\documentclass[letterpaper, 12pt]{amsart}
\usepackage{amsmath,amsthm,amssymb}
\usepackage{graphicx,float}
\usepackage[in]{fullpage}
\usepackage{amsmath}
\usepackage{amsthm,amsbsy}

\usepackage{amssymb}
\usepackage{graphicx}
\usepackage{graphicx}
\usepackage{dcolumn}
\usepackage{bm}
\usepackage{amsfonts}
\usepackage{latexsym}
\usepackage{pdfsync}
    \usepackage{fullpage} 
\usepackage{colordvi}
\usepackage{color}
\usepackage{booktabs}
\usepackage{graphicx}
\usepackage{subfigure}
\usepackage{stmaryrd}
\usepackage{cite}
\input xy
\xyoption{all}

\usepackage{color}

\newcommand{\commentout}[1]{}
\newcommand{\nwc}{\newcommand}

\nwc{\red}{\color{red}}
\nwc{\blue}{\color{blue}}

\nwc{\nn}{\nonumber}

\nwc{\nwt}{\newtheorem}
\nwt{cor}{Corollary}
\nwt{proposition}{Proposition}
\nwt{corollary}{Corollary}
\nwt{theorem}{Theorem}
\nwt{summary}{Summary}
\nwt{lemma}{Lemma}
\nwt{definition}{Definition}
\nwt{remark}{Remark}

\nwc{\FF}{\mathcal{F}}
\nwc{\PP}{\mathcal{P}}
\nwc{\xx}{\mathbf{x}}
\nwc{\CC}{\mathbb{C}}
\nwc{\ZZ}{\mathbb{Z}}
\nwc{\RR}{\mathbb{R}}
\nwc{\bk}{\mathbf{k}}
\nwc{\bz}{\mathbf{z}}
\nwc{\bt}{\mathbf{t}}
\nwc{\bom}{\boldsymbol\omega}
\nwc{\bn}{\mathbf{n}}
\nwc{\bN}{\mathbf{N}}
\nwc{\PO}{\mathcal{P}_{\rm o}}
\nwc{\PF}{\mathcal{P}_{\rm f}}
\nwc{\QO}{\mathcal{Q}_{\rm o}}
\nwc{\QF}{\mathcal{Q}_{\rm f}}
\nwc{\PT}{\mathcal{T}}
\nwc{\real}{\text{re}}
\nwc{\imag}{\text{im}}
\nwc{\ep}{\epsilon}
\nwc{\vep}{\varepsilon}
\nwc{\tvep}{\tilde{\vep}}
\nwc{\lt}{\left}
\nwc{\rt}{\right}

\nwc{\mf}{\mathbf}
\nwc{\mb}{\mathbf}
\nwc{\ml}{\mathcal}
\nwc{\bj}{{\mb j}}
\nwc{\bh}{{\mb h}}
\nwc{\bA}{{\mb A}}
\nwc{\IA}{\mathbb{A}} 
\nwc{\bi}{\mathbf i}
\nwc{\bo}{\mathbf o}
\nwc{\IS}{\mathbb{S}}
\nwc{\IC}{\mathbb{C}} 
\nwc{\ID}{\mathbb{D}} 
\nwc{\IM}{\mathbb{M}} 
\nwc{\IP}{\mathbb{P}} 
\nwc{\bI}{\mathbf{I}} 
\nwc{\IE}{\mathbb{E}} 
\nwc{\IF}{\mathbb{F}} 
\nwc{\IG}{\mathbb{G}} 
\nwc{\IN}{\mathbb{N}} 
\nwc{\IQ}{\mathbb{Q}} 
\nwc{\IR}{\mathbb{R}} 
\nwc{\IT}{\mathbb{T}} 
\nwc{\IZ}{\mathbb{Z}} 
\nwc{\IV}{\mathbb{V}}
\nwc{\IX}{\mathbb{X}}
\nwc{\IY}{\mathbb{Y}}
\nwc{\II}{\mathbb{I}}

\nwc{\cE}{{\ml E}}
\nwc{\cP}{{\ml P}}
\nwc{\cQ}{{\ml Q}}
\nwc{\cL}{{\ml L}}
\nwc{\cX}{{\ml X}}
\nwc{\cW}{{\ml W}}
\nwc{\cZ}{{\ml Z}}
\nwc{\cO}{{\ml O}}
\nwc{\cV}{{\ml V}}
\nwc{\cT}{{\ml T}}
\nwc{\crV}{{\ml L}_{(\delta,\rho)}}
\nwc{\cC}{{\ml C}}
\nwc{\cA}{{\ml A}}
\nwc{\cK}{{\ml K}}
\nwc{\cB}{{\ml B}}
\nwc{\cD}{{\ml D}}
\nwc{\cF}{{\ml F}}
\nwc{\cS}{{\ml S}}
\nwc{\cM}{{\ml M}}
\nwc{\cG}{{\ml G}}
\nwc{\cH}{{\ml H}}
\nwc{\bF}{{\mathbf F}}
\nwc{\bG}{{\mathbf G}}
\nwc{\br}{{\mb r}}
\nwc{\bp}{{\mb p}}
\nwc{\bq}{{\mb q}}
\nwc{\bR}{{\mb R}}
\nwc{\bM}{{\mb M}}
\nwc{\cbz}{\overline{\cB}_z}
\nwc{\supp}{{\hbox{\rm supp}}}
\nwc{\fR}{\mathfrak{R}}
\nwc{\bY}{\mathbf Y}

\nwc{\pft}{\cF^{-1}_2}
\nwc{\bU}{{\mb U}}
\nwc{\bPhi}{{\mb \Phi}}
\nwc{\bPsi}{{\mb \Psi}}
\nwc{\im}{{\rm i}}
\nwc{\om}{\omega}
\nwc{\bdhat}{\hat{\mb d}}

\nwc{\bw}{{\mathbf w}}
\nwc{\mbm}{{\mathbf m}}
\nwc{\lbr}{\textlbrackdbl}
\nwc{\rbr}{\textrbrackdbl}
\nwc{\vzero}{{\mathbf 0}}
\nwc{\cN}{{\mathcal N}}
\nwc{\rbra}{\textrbrackdbl}
\nwc{\lbra}{\textlbrackdbl}
\nwc{\conv}{\hbox{conv}}
\nwc{\rank}{\hbox{rank}}

\nwc{\beq}{\begin{eqnarray}}
\nwc{\beqn}{\begin{eqnarray*}}
\nwc{\eeqn}{\end{eqnarray*}}
\nwc{\eeq}{\end{eqnarray}}

\commentout{
\theoremstyle{remark}               
\newtheorem*{lemma}{\bf Lemma}

\newtheorem*{theorem}{\bf Theorem}

}

\renewcommand{\tilde}{\widetilde}
\renewcommand{\hat}{\widehat}
\newcommand{\e}{\mathrm{e}} 
\renewcommand{\i}{\mathrm{i}} 
\newcommand{\diag}{\operatorname{diag}}

\renewcommand{\d}{\,\!\operatorname{d}\!} 
\renewcommand{\@}{\partial}

\renewcommand{\vec}[1]{\boldsymbol{#1}}

\newcommand{\ind}{\operatorname{I}} 

\newcommand{\Exp}{\mathbb{E}} 
\renewcommand{\Pr}{\IP}

\newcommand{\paren}[1]{\left({#1}{}_{}^{}\right)} 
\newcommand{\bracket}[1]{\left[{#1}{}_{}^{}\right]} 
\newcommand{\abs}[1]{\left| {}_{}^{} {#1}{}_{}^{} \right|}
\newcommand{\norm}[1]{\left\|{#1}\right\|}
\newcommand{\set}[1]{\left\{#1\right\}} %

\nwc{\pdfi}{{f}}
\nwc{\pdfs}{{f^{\rm s}}}
\nwc{\pdfii}{{f_1^{\rm i}}}
\nwc{\pdfsi}{{f_1^{\rm s}}}
\nwc{\chis}{{\chi^{\rm s}}}
\nwc{\mbx}{\mathbf{X}}
\nwc{\bX}{\mathbf X}
\nwc{\bZ}{\mathbf Z}
\nwc{\bE}{\mathbf E}
\nwc{\chii}{{\chi^{\rm i}}}
\nwc{\bB}{\mathbf B}
\nwc{\bH}{\mb H}

\begin{document}
\title{Compressive Radar with Off-Grid  Targets: A Perturbation Approach}
\author{Albert Fannjiang$^{1*}$ and Hsiao-Chieh Tseng$^2$}
 \address{$^1$Department of Mathematics, University of California, Davis, CA 95616-8633, USA.}
\address{
$^2$Department  of Land, Air, \& Water Resources, University of California, Davis, CA 95616.}
\address{
$^*$Corresponding author: fannjiang@math.ucdavis.edu}

\begin{abstract}
Compressed sensing (CS) schemes are proposed
for monostatic as well as synthetic aperture radar (SAR) imaging with chirped signals and 
Ultra-Narrowband (UNB) continuous waveforms.  In particular, a simple, perturbation 
method is developed to reduce the gridding error for off-grid targets. A coherence bound is obtained
for the resulting measurement matrix. A greedy pursuit algorithm, Support-Constrained Orthogonal Matching Pursuit (SCOMP), 
is proposed to take advantage of the support constraint in the perturbation  formulation
and proved to have the capacity of determining  the off-grid targets to the grid accuracy under favorable conditions.
Alternatively, the Locally Optimized Thresholding (LOT) is proposed to
enhance the performance of the CS method, Basis Pursuit (BP). 
For the advantages of higher signal-to-noise ratio and 
signal-to-interference ratio, it is proposed that Spotlight SAR imaging be
implemented with CS techniques and multi-frequency UNB waveforms. 
Numerical simulations show promising results of the proposed approach and
algorithms.  \end{abstract}
\maketitle

  \section{Introduction}
Advances in  compressed sensing (CS) and radar processing have provided tremendous impetus to each other. On the one hand, 
the two CS themes of sparse reconstruction and low-coherence, pseudo-randomized data acquisition are longstanding
concepts in radar processing.
On the other hand,  CS contributes provable performance
guarantees for sparse recovery algorithms and
informs refinement of these algorithms. These and other
important issues relevant to CS radar are thoroughly reviewed in \cite{Ender10, PEPC10} (see also the references therein).

Target sparsity, a main theme in  CS, arises naturally in radar processing. According to the geometrical theory of diffraction \cite{Kel},  the
scattering response of a target at radio frequencies can often be approximated as a sum of responses from individual reflectors.
These scattering centers provide a concise, yet physically
relevant, representation of the target  \cite{GPG}.   A spiky reconstruction of reflectivity may thus be highly valuable for automatic
target recognition. More generally radar images are compressible 
by means of 
either parametric models of physical scattering behaviors
or transform coding \cite{PEPC10}. 

In the present work,
we focus on the case of off-grid point targets which do not sit on a regular grid.
A main drawback  of the standard CS framework is the reliance on a underlying well-resolved  grid \cite{BS07, HS09}. In reality, the dominant scattering centers can not be
assumed to be positioned exactly at the imaging  grid 
points. Indeed,  the standard CS methods  break
down if the effects of off-grid targets are not accounted for \cite{Ender10,PEPC10}.
The problem is,  to reduce gridding error, the grid has to be
refined, giving rise to high coherence of the measurement matrix
which is detrimental  to standard CS methods  \cite{FL1,FL2,Asilomar12}. 
\commentout{
As pointed out in \cite{PEPC10}, ``The suppression of
image side-lobes by $\ell^1$ or greedy algorithms may invite a
qualitative claim of ``super-resolution"; however, existing
CS results are silent regarding resolution and agnostic
regarding bias and variance of parameter estimates in the
underlying continuous parameter space. Indeed, super-resolution
implies that mutual coherence must be large.
}

Can CS approach be extended to the case of  arbitrarily located targets? Several approaches
have been proposed to address this critical question \cite{CEN11,DB11,FL1, FL2}.
In this paper we propose a simple, alternative approach,
based on improved  measurement matrices 
as well as improvement in reconstruction algorithms, including
 a greedy pursuit algorithm, called Support-Constrained Orthogonal Matching Pursuit (SCOMP),  to take advantage the support
constraint arising in the new formulation  (Section \ref{sec:fixed}). 
We obtain coherence bounds for the measurement matrices with
the linear chirp (Lemma \ref{lem1}). We prove that
the greedy algorithm can determine  the targets to the grid accuracy under
favorable conditions and obtain an error bound for the target amplitude
recovery (Theorem \ref{thm:omp}). 

In Section \ref{sec:sar} we consider   the Spotlight mode of
Synthetic Aperture Radar (SAR). We extend the approach for
off-grid targets to  Spotlight SAR and propose sparse sampling schemes based on multi-frequency 
Ultra-Narrowband (UNB) waveforms  (Section \ref{sec:narrow}). We extend  the performance guarantee for SCOMP to Spotlight SAR imaging
(Theorem \ref{thm22}).  
Finally we present numerical experiments demonstrating 
the effectiveness of our approach (Section \ref{sec:num}) and draw conclusion 
(Section \ref{sec7}).  

  \section{Monostatic signal model}
 Let us begin by reviewing  the signal model for a mono-static radar with co-located transmit and receive antennas. A complex waveform $f$ with the carrier frequency $\omega_0$ is transmitted. 
Let $r$ and  $v$ denote the range and the radial velocity,
respectively. 
  We parameterize the complex scene by the reflectivity function $\rho(\tau,u)$ 
  where the delay $\tau= 2r/c_0$ is the round-trip propagation time
  and    $
    u={2v \omega_0}/{c_0}$ is
     the Doppler shift. 
Under the far-field and  {narrow-band approximations} \cite{CB08}, 
the scattered signal is given by 
\beq
    y(t) &=\iint
    x(\tau,u)
    f(t-\tau)
    \e^{-2\pi\i u t}
    \d  u \d \tau +w(t)
    \label{00}
 \eeq
 where 
     \[
     x(\tau, u)=\rho(\tau, u)\e^{-\pi\i u \tau }
     \]
     and $w(t)$  represents the measurement noise. 
     
     In the present work, we focus on the case of  immobile targets,  
  $\rho(\tau, u) = \rho(\tau)\delta( u)$. Eq. (\ref{00}) becomes
      \begin{equation}
    y(t) = \int_{-\infty}^{\infty}
     \rho(\tau) f(t - \tau) \d \tau +w(t).
     \label{SM_conv}
     \end{equation}
 
 For the transmitted signal, let
  $ \ind_T$ be the indicator function of duration  $[0,T]$
  of transmission. 
By far the most commonly used waveform is the linear frequency-modulated  chirp 
\beq
f_{\text{LC}}(t) =  
    \exp\bracket{2\pi\i\paren{\frac{\alpha_1}{2} t^2 + \omega_0 t}} \ind_T(t)
    \eeq
    owing to the simplicity  in implementation. 
    The bandwidth of linear chirp is
    $B= \alpha_1 T$.

   With discrete targets located at $\{\tau^*_k: k=1,\cdots,s\}$ and sampling times $\{t_j:j=1,\cdots,m\}$, the signal model is given by
  \beq
  \label{3'}   y(t_j) &=&  \sum^s_{k=1} \rho_k f_{\rm LC}(t_j - \tau^*_k) +w(t_j)\\
  &=& \sum^s_{k=1} \rho_{k} 
   \exp\bracket{2\pi\i\paren{\frac{\alpha_1}{2} (t_j-\tau^*_k)^2 + \omega_0 (t_j-\tau^*_k)}} +w(t_j).\nn
   \eeq
  \subsection{On-grid targets}  
 
 Suppose that the targets are located exactly on the grid points of spacing $\Delta\tau$, i.e. each $\tau_k^*$ is an integer multiple of $\Delta\tau$ . Then it
 is natural to extend $\{\rho_k\}$ to the entire imaging grid,  with value zero when a target is absent, and turn (\ref{3'}) into a linear inversion problem as follows.  
 Let $\tau_k = k\Delta\tau$,
 $k=1,\ldots,n$ and 
    \[
  \bar{t}_j = t_j/T \in [0,1].
  \]
be the normalized sampling times.

 We have from (\ref{3'}) that 
      \beq
 y(t_j)  &= &  f_{\text{LC}}(t_j) \sum_{k=1}^n\rho_{k} f_{\text{LC}}(-\tau_k) 
   \exp\bracket{-2\pi\i \alpha_1 \tau_k t_j}+w(t_j),\quad j=1,...,m\label{7''}    \label{SM_conv_pt}
  \eeq
{where $n$ is the total number of grid points in the range and $m$ is the number of observed data.} In the absence of a target at a grid point $\tau_k$, the corresponding target amplitude
$\rho_k=0$. 

 The main point of CS is to recover the targets, $\{\rho_k,\tau_k\}_{k=1}^n$, from $\{y(t_j)\}_{j=1}^m$  with $m$ much smaller than $n$. 

 Suppose that  the bandwidth $B$ satisfies  \beq
 Q= B\Delta \tau=\alpha_1 T\Delta\tau\in\mathbb{N}\label{B}
  \label{33}
  \eeq
  where   $Q$ is the resolution-time-bandwidth product. 
Then with 
\beq
Y_j &=& {y(t_j)}/{   f_{\text{LC}}(t_j) } \label{6}\\
E_j&=&w(t_j)/f_{\rm LC}(t_j)\label{6'}\\
X_k &=& \rho_{k} f_{\text{LC}}(-\tau_k) \label{7}\\
  F_{jk} &= &\exp\bracket{-2\pi\i \alpha_1 \tau_k t_j}
= \exp\bracket{-2\pi\i Q k \bar{t}_j}\label{8}
\eeq
 we can write the signal model (\ref{SM_conv_pt}) as the  linear system 
 \beq
 Y=\bF X+E.\label{10'}
 \eeq 

A main thrust of CS is the performance guarantee for
the Basis Pursuit (BP):
\beqn
\hat X=\hbox{arg}\min \|Z\|_1,\quad \|\bF Z-Y\|_2\leq\epsilon
\eeqn
under the assumption of the restricted isometry property (RIP):
\beqn
a(1-\delta_k)\|Z\|_2\leq\|\bF Z\|_2\leq a(1+\delta_k)\|Z\|_2
\eeqn
for some constant $a>0$ and all $k$-sparse $Z$ where $\delta_k$ is the $k$-th order
restricted isometry constant. 
More precisely, we have the following statement for $Q=1$ 
 \cite{Candes08,Rauhut08}. 

  \begin{proposition}
 Let $ \bar{t}_j\in[0,1]$, $j=1,2,\ldots,m$ be independent uniform random variables. If
   \beq
   \label{60}
   \frac{m}{\ln m} \geq Cs\ln^2 s\ln n\ln\frac{1}{\beta}\ , \quad \beta\in(0,1)
   \eeq
   for some universal constant $C$ and sparsity level $s$, then the random partial 
   Fourier measurement matrix $[\exp{(-2\pi\i k\bar{t}_j)}]$, $k=1,\ldots,n$,  satisfy the RIP  with $\delta_{2s}<\sqrt{2}-1$ and the BP solution $\hat X$ satisfies
   $$\norm{\hat{X}-X}_2\leq C_0 \frac{1}{\sqrt{s}}\norm{X^{(s)}-X}_1 + C_1\norm{E}_2
   \ , \quad \norm{\hat{X}-X}_1\leq  C_0 \norm{X^{(s)}-X}_1 + C_1\norm{E}_2$$
   for some constants, $C_1$, with probability at least $1-\beta$.
   Here $X^{(s)}$ is the best $s$-sparse approximation of $X$. 
   \label{thm1}
  \end{proposition}
 The assumption of independent uniform random variables underlies 
 the  important role of random sampling in CS. Random sampling also induces 
 incoherence (see Lemma \ref{lem1} below).  Depending on the nature of measurement matrix  certain random measurements tend to yield the best performance in reconstruction 
 with sparse sampling. 
   
  According to the above result, the (normalized) sampling times
  should be chosen randomly and uniformly in $[0,1]$ and the number
  of time samples $m$ on the order of the target sparsity $s$, up to a logarithmic factor.  
  Proposition \ref{thm1} is useful  as long as the point targets are located exactly
  on the grid points which is an unrealistic assumption. 
  \subsection{Off-grid targets}  \label{sec:fixed}
  In practice, the time delays $\{{\tau}^*_k\}$ do not sit exactly on the grid. The mismatch between the actual signal and
  the signal model creates the gridding  error leading  to
  poor performance  of the standard CS methods  \cite{CSPC11, FL1,FL2}. 
  
  To remedy  this problem and   reduce the gridding error,  we modify  the  signal model as follows.   

  Let $ {\tau}_k = (k+\xi_k)\Delta\tau $, with
  $\abs{\xi_k}<1/2$, which is meant to capture the actual target time delays. 
 We modify  (\ref{7''}) to obtain 
  \beq
   &&y(t_j) =   f_{\text{LC}}(t_j) \sum_k \rho_{k} f_{\text{LC}}(-k\Delta\tau-\Delta\tau\xi_k)
   \exp{(-2\pi\i Q\bar t_j k)}\exp{(-2\pi\i Q\xi_k \bar t_j)}+w(t_j).   \label{SM_Range_pt}
     \eeq
For small $B\max |\xi_k|$  we can write 
  \beq
  \label{second}
 \e^{-2\pi\i Q\xi_k \bar{t}_j} 
  = \e^{-\pi i Q \xi_k} \Big(1 - 2\pi\i Q \xi_k (\bar{t}_j-1/2)+ \mathcal{O}\big( Q^2\abs{\xi_k}^2 \big)\Big).
  \eeq
Let 
  \beq
  \label{sig}
    \sigma=\Big(m^{-1}\sum_l|\bar{t}_l-1/2|^2\Big)^{1/2}
  \eeq
  be the time sample variation. With 
\beq 
   Y_j &= &{y(t_j)}/{   f_{\text{LC}}(t_j) }\label{8'} \\
   X_k &=& \rho_{k} f_{\text{LC}}(-k\Delta\tau-\xi_k\Delta\tau) e^{-\pi i Q\xi_k}
  \\
  X'_k &=& -2\pi \i {  \sigma}{Q} \xi_k X_k,
 \label{9}\\   
   F_{jk} &=& \exp\bracket{-2\pi\i Q k \bar{t}_j}
   \ , \quad 
   G_{jk} = F_{jk}(\bar{t}_j-1/2){  \sigma^{-1}}\label{10}
\eeq
 the linear system takes the form 
  \begin{equation*}Y_j = \sum_k \big( F_{jk}X_k + G_{jk}X'_k \big)+E_j
  \end{equation*}
{ or equivalently } 
\beq
  {Y} = \begin{bmatrix} \bF & \bG \end{bmatrix} \begin{bmatrix} X \\ X'\end{bmatrix}+E \label{CS_offgridfix}\label{7'}
  \eeq
  where the error term 
  \beq
  \label{error}
  E_j=w(t_j)/f_{\rm LC}(t_j)+   \sum_k 
  F_{jk}X_k \mathcal{O}\big( Q^2|\xi_k|^2 \big)     
  \eeq
  contains not only the measurement noise but also
  the gridding error due to neglect of the second order term in (\ref{second}). 
  
  From (\ref{9}) we see that the magnitude of $X'$ is directly proportional to $Q$. 
  Moreover, increasing $Q$ also increases the error in the approximation (\ref{second})
  and hence the gridding error for the system (\ref{7'}). 
  
  After $X$ and $X'$ are solved from the system, we  can estimate $\{\xi_k\}$ and $\{\rho_k\}$, respectively, by
  \[
  \hat \xi_k={i\Delta\tau X_k'\over 2\pi{\sigma} Q X_k}
  \]
   and 
  \[
 \hat  \rho_k= {X_k \e^{\pi i Q\hat \xi_k}\over f_{\text{LC}}(-k\Delta\tau-\Delta \tau \hat \xi_k) }.
  \] 
  
  It is generally difficult to establish RIP for
  matrices other than random partial Fourier matrices and
  random matrices of independently and identically distributed (i.i.d.) entries. 
  An alternative notion is the mutual coherence. The mutual coherence $\mu$ of a matrix $\bA$ is defined by
  the maximum normalized inner product between columns of $\bA$:
  \beq
  \mu(\bA)=\max_{i\neq j}{|A_i^*A_j|\over \|A_i\|_2\|A_j\|_2}.
  \eeq
In CS one seeks low level of mutual coherence in the measurement matrix.  
  
  The following lemma states a coherence bound for
  the system  (\ref{7'}). 
   \begin{lemma}  Let $ \bar{t}_j\in[0,1]$, $j=1,2,\ldots,m$ be independent uniform random variables. Suppose 
  $2n<\delta \exp\bracket{K^2/2}$ where $\delta$ and $K$ are two arbitrary numbers. Then the mutual coherence $\mu$ of the
  combined sensing matrix $\bA=[\bF \,\, \bG]$ satisfies  
  \beq
  \label{19}
  \mu \leq  C \bracket{ \frac{\sqrt{2}K}{\sqrt{m}} + \frac{1}{2\pi Q} }
  \eeq
  for some universal constant $C$, with probability greater than $(1-\delta)^2-4e^{-m/18}$.\label{lem1}
  \end{lemma}
   The proof of the lemma is given
 in appendix \ref{app1}. 
 This lemma says that to reduce the mutual coherence of
 the sensing matrix one should increase the number of data and $Q$.  
 The $Q$-dependent second term on the right hand side of (\ref{19}) is
 due to the presence of the perturbation matrix $\bG$ in the signal model (\ref{7'}) while
 the mutual coherence of the primary matrix  $\bF$ is the first $\cO(m^{-1/2})$ term.

 \commentout{
We then obtain an error bound from
 Lemma \ref{lem1}  in conjunction with 
 the following result (Theorem 3.1 of \cite{DET06}).
       \begin{proposition} Suppose the data is noisy and the $\ell^2$-norm
      of the noise is less than $\epsilon$. 
 If the sparsity of $X$ satisfies 
 \beq
 \label{20}
 s<{1\over 4} (1+{1\over \mu})
 \eeq
    then the BP solution $\hat X$
    satisfies the error bound
    \begin{equation}
     \norm{\hat X -X}^2_2\leq \frac{4\epsilon^2}{1-(4s-1)\mu} .     
    \end{equation}
\label{thm2}
  \end{proposition}
}

\subsection{  Support-constrained OMP}
Let $\supp(X)$ denote the support set of $X$ which is the set of
index $j$ with $X_j\neq 0$. 
 Note that  the support constraint 
  \beq
  \label{sc}
  \supp(X')\subseteq\supp(X)
  \eeq
   can be utilized in the greedy pursuit such as Orthogonal Matching Pursuit (OMP) as follows. A common stage for any greedy pursuit is to choose the index corresponding to 
  the column(s) of the maximum coherence with the residual vector. Since
  ${X},{X}'$ have the same sparse structure, one may utilize the a priori information: choose the
  $k$-th columns of $\bF$ and $\bG$, and test the size of the projected vector from ${Y}$ on the span
  of the two columns.
 
 {
    \begin{center}
   \begin{tabular}[width=4in]{||l||}
   \hline
   \centerline{{\bf Algorithm 1.}\quad  Support-Constrained OMP (SCOMP)} \\ \hline
   Input: $\bF,\bG, Y, \|E\|_2$\\
 Initialization:  $X^0 = 0, R^0 = Y$ and $\cS^0=\emptyset$ \\ 
Iteration:\\
\quad 1) $i_{\rm max} = \hbox{arg}\max_{i}\Big(|F^*_{i}R^{k-1}|+|G^*_{i}R^{k-1} |\Big)$\\
 \quad      2) $\cS^k= \cS^{k-1} \cup \{i_{\rm max}\}$ \\
  \quad  3) $(X^k, X^{'k}) = \hbox{arg} \min\|
     \bF Z+\bG Z'-Y\|_{2}$ s.t. \hbox{supp}($Z'$) $\subseteq$  \hbox{\supp}($Z$) $\subseteq S^k$ \\
  \quad   4) $R^k = Y- \bF X^k-\bG X^{'k}$\\
\quad  5)  Stop if $\|R^k\|_{2}\leq \|E\|_2$.\\
 Output: $\hat X=X^k, \hat X'=X^{'k}$. \\
 \hline
   \end{tabular}
\end{center}
\bigskip

We have the following performance guarantee for SCOMP.
\begin{theorem} Suppose that the columns of $\bA$  have same 
  2-norm. 
 Let  $\hbox{supp} (X)=\{J_1,\ldots, J_s\}$
and $$X_{\rm max}=|X_{J_1}|+|X'_{J_1}|\geq |X_{J_2}|+|X'_{J_2}|\geq \cdots\geq |X_{J_s}|+|X'_{J_s}|=X_{\rm min}. $$. 
Suppose \beq
\label{snr2}
(4s-1)\mu + \frac{4\|E\|_2}{X_{\rm min}} < 1
\eeq
and let $\hat X$ and $\hat X'$ be the SCOMP estimates. Then 
\beq
\label{exact}
\supp(\hat X)=\supp(X)
\eeq
and
\beq
\label{err:omp}
\|\hat X-X\|^2_{2}+\|\hat X'-X'\|_2^2\leq {{2}\|E\|_2^2\over {1-\mu (2s-1)}}.
\eeq
\label{thm:omp}\label{thm2}
\end{theorem}
The proof is given in Appendix \ref{app:omp}. 

\begin{remark}\label{rmk1}
Since $E$ contains the griding error, the error bound (\ref{err:omp}) may be
too crude to be useful. To improve the accuracy of $\hat X$ and $\hat X'$  we can perform nonlinear least squares (NLS) 
on (\ref{SM_Range_pt}) subject to the exact recovery of the target support (\ref{exact}).
In other words, we solve for
\beq
&&\arg\min\sum_j\Big|f(t_j)-  f_{\text{LC}}(t_j) \sum_{k\in \hbox{\tiny \rm supp} (\hat X)} \rho_{k} f_{\text{LC}}(-k\Delta\tau-\Delta\tau\xi_{k})
   \e^{-2\pi\i Q\bar t_j k}\e^{-2\pi\i Q \xi_{k} \bar t_j}\Big|^2\label{nls}
   \eeq
in  the set of all $\{\rho_{k}:k\in \supp (\hat X)\}\subset \IC^s$ and $\{\xi_{k}: k\in \supp (\hat X)\}\subset (-0.5,0.5)^s$. 
It is natural to use 
 the SCOMP output $\hat X, \hat X'$ 
as the initial guess for iterative methods (e.g. the Gauss-Newton method or gradient methods) for (\ref{nls}). 
\end{remark}
}

\begin{remark}\label{rmk3}
Lemma \ref{lem1} and Theorem \ref{thm:omp} together suggest  that to enhance
the performance of SCOMP one should increase $Q$. On the other hand,
larger $Q$ also tends to correspond to a larger gridding error for the system (\ref{7'}).  As we shall see
in Section \ref{sec:num}, $Q=1$ yields the best result. 
As the $Q$-dependence of the coherence estimate (\ref{19}) is due to the perturbation matrix $\bG$,  
we speculate that the actual  performance of SCOMP has more to do with
the mutual coherence of the primary matrix $\bF$ which is $Q$ independent and decays
like $m^{-1/2}$. 

\end{remark}

{ Before ending this section, we note that the gridding error term in (\ref{error})
has the appearance of the matrix perturbation problems studied 
in \cite{HS2,ZLG11}. The analogy, however, is superficial as
$\xi_k$ in (\ref{error}) are part of the unknown and hence
the gridding error is {\em cubic}, not linear,  in the unknown.  
}

 \commentout{ 
The sparsity constraint (\ref{20}) can be relaxed if the targets are
randomly distributed in distance (time delay) as stated in   the following proposition \cite{CP09}.
    \begin{proposition}
Suppose  
 the data noises  are additive  Gaussian noise of
      variance $\sigma^2$. 
      Assume that the targets are  uniformly randomly distributed
      and $s$-sparse such that
      \beq
      \label{21}
       s \leq \frac{c_1 n m}{\norm{\bA}_2^2 \log n}. 
      \eeq
  Assume also \beq\label{23}
\mu \leq \frac{c_2}{\log n}.
\eeq
      If 
   \beq
   \label{22}
      \min_{k\in S}\abs{X_k} > 8\sigma\sqrt{2\log n}
  \eeq
      then the solution $\hat{X}$ of the LASSO
      \beq
      \label{lasso}
       \min_{Z}\ \lambda\sigma\norm{Z}_1 + \frac{1}{2}\norm{Y-\bA Z}_2^2 ,\quad \lambda = 2\sqrt{2\log n}      
      \eeq
    obeys $
       \supp(\hat{X}) = \supp(X)       
$
      with probability at least $1-\mathcal{O}(n^{-1})$.\label{prop2}
    \end{proposition}
    Once the targets are determined to the grid accuracy  the target strengths can be computed by a simple least-squares step to obtain a nearly
    optimal reconstruction. 
    
On the one hand, (\ref{21}) demands little on the mutual coherence. On the other hand, in order for (\ref{21}) to improve over (\ref{20}) we need a sufficiently strong spectral norm bound. We prove
the following spectral norm bound  in appendix \ref{app2}. 

\begin{lemma} The sensing matrix in  (\ref{10}) satisfies
the spectral norm bound 
\[
\norm{\bA}_2^2 \leq 2n
\] 
with probability greater than 
\[
\left(1-{C (m-1)\over n}\right)^{m(m-1)} 
\]
where $C$ is an absolute constant. 
\label{lem2}
\end{lemma}
Lemmas \ref{lem1}, \ref{lem2} and Proposition \ref{prop2} imply that
LASSO can achieve nearly optimal reconstruction
of randomly distributed targets  whose number is on the order of the
data number $m$, modulo a logarithmic factor. 
 This is stated in the following theorem.
 
  \begin{theorem}
  Suppose  
 the data noises  are additive  Gaussian noise of
      variance $\sigma^2$. 
      Assume that the targets are  uniformly randomly distributed
      and $s$-sparse such that 
      \[
      s\leq {c_1m\over 2\log n}.
      \]
      Suppose that
      \[
      \frac{\sqrt{2}K}{\sqrt{m}} + \frac{1}{2\pi Q} \leq {c_2\over C\log n}.
      \]
      If (\ref{22}) holds, then the LASSO (\ref{lasso})  recovers the target support
      exactly     with probability at least $1-\mathcal{O}(n^{-1})$.
      \label{new-thm2}
  \end{theorem}
  }

    \section{Spotlight SAR}
    \label{sec:sar} \label{sec:narrow}
   
   In this section, we consider the Spotlight SAR   for a stationary scene,  
represented by the reflectivity  $\rho(\br)$.
For simplicity of the presentation, we focus on the case of two dimensions $\br=(r_1,r_2)$. 
The adaption to three dimensions is straightforward. 

 In standard radar processing, the received signal, upon receive, is typically  deramped by mixing the echo with the
reference transmitted chirp \cite{Jak}.  Under the start-stop approximation and a far-field
assumption the deramp processing produces  samples of the Fourier transform of the Radon
projection, orthogonal to the
radar look direction, of the scene reflectivity multiplied by a quadratic phase term. Furthermore, if  the time-bandwidth product $TB=\alpha_1T^2$  is significantly  larger than the total number $n$ of resolution cells, 
 the quadratic phase term can be neglected and the deramped
 signal can be written simply as \cite{MOJ}
   \beq
 && y(\nu,\theta)=\cF[\rho](\nu\cos\theta,\nu\sin\theta) +w(\nu,\theta) \label{45}
\eeq 
where $\cF$ is the 2-d Fourier transform, $\theta$  the look angle, $\tau_0$ the round-trip travel time to the scene center, $w$ the measurement noise
and 
 \beq
 \nu(t)={2\over c_0}\lt(\omega_0+\alpha_1 (t-\tau_0)\rt)
 \label{45'}
\eeq 
 the spatial frequency.
For a sufficiently small scene,  $t$ is effectively limited to $[\tau_0,\tau_0+T]$ and hence
$\nu(t)$  is restricted to 
\beq
\nu\in [\nu_0,\nu_*],\quad \nu_0={2\omega_0/ c_0},\ \  \nu_*=\nu_0+{2\alpha_1 T/c_0}. 
\label{50'}\label{59}
 \eeq
 
 Alternatively, the SAR tomography  (\ref{45}) can be implemented by multi-frequency, Ultra-Narrowband (UNB) continuous waveforms \cite{cis-siso}. A multi-frequency UNB SAR has many practical advantages such as 1)  relatively simple, low cost transmitters are deployed,  2) SNR is increased
 as reduced bandwidth results in less unwanted thermal noise, 3) UNB signals provide 
 relief when the available electromagnetic spectrum is eroded by other civilian and military
 radar applications. For UNB SAR, the spatial frequency   $\nu$ in (\ref{45}) 
 is related to the carrier frequency $\omega$ of continuous waveform by
 $\nu=2\omega/c_0$.  UNB multi-frequency SAR is particularly appealing from the point
 of view of compressed sensing as the associated multiple spatial frequencies can be
 viewed as sparse sampling of the continuous range $[\nu_0,\nu_*]$  of spatial frequencies.

 Let the imaging domain be the finite square lattice
\beq
\label{pix}
\cL=\lt\{\ell(p_1, p_2): p_1,p_2=1,...,\sqrt{n}\rt\}.
\eeq
The total number of  cells
 $n$ is  a perfect square. For the off-grid  targets represented 
 by
 \[
 \rho(\br)=\sum_{\bp\in \IZ^2} \rho_\bp \delta (\br-\ell\bp-\ell \bh_\bp), \quad \bh_\bp=(h_{1\bp}, h_{2\bp}),\,\, |h_{1\bp}|, |h_{2\bp}|< 1/2
 \]
 the signal model (\ref{45}) becomes
 \beqn
 y(\nu,\theta)&=&\sum_{\bp\in \IZ^2}\rho_\bp \exp{[-2\pi \i \ell \nu \bdhat\cdot
 (\bp+\bh_\bp)]} +w(\nu,\theta) \eeqn
   where $\bdhat=(\cos\theta,\sin\theta)$ denotes
the direction of look.
Following  the same perturbation technique
\[
\e^{-2\pi \i \ell \nu \bdhat\cdot
 (\bp+\bh_\bp)}=\e^{-2\pi \i \ell \nu \bdhat\cdot
 \bp}\Big(1-2\pi\i \ell\nu \bdhat \cdot\bh_\bp+\cO( |\ell\nu \bdhat \cdot\bh_\bp|^2)\Big)
 \]
 we consider the signal model
 \beq
 y(\nu,\theta)
 &=&\sum_{\bp\in \IZ^2}\rho_\bp \e^{-2\pi \i \ell \nu \bdhat\cdot
 \bp}(1-2\pi\i \ell\nu \bdhat \cdot\bh_\bp)+e(\nu,\theta)
 \eeq
 where the error term 
 \beq
 \label{55}
   e(\nu,\theta)=w(\nu,\theta)+\sum_\bp \rho_\bp \e^{-2\pi \i \ell \nu \bdhat\cdot
 \bp}\cO( |\ell\nu \bdhat \cdot\bh_\bp|^2)
 \eeq
 includes the measurement noise $w$
 and the gridding error. 
 
 We shall distinguish two regimes: the Fully Diversified  Multi-Frequency (FDMF) SAR with
 $\nu_0=0$ and the Partially Diversified  Multi-Frequency (PDMF) SAR with $\nu_0>0$. 
 
 First we describe a general sampling scheme applicable to both regimes.\\

{\bf SAR scheme A:} We independently  select
$\theta_k, k=1,\ldots,m_1$ according to a probability density function $\phi$ on  $[0,2\pi]$ 
and then, for each $\theta_k$, independently select
$\nu_{kl},l=1,\ldots, m_2,$ according to a probability density
function $g$ on $ [\nu_0,\nu_*]$. The simplest case is with 
$\phi=1/(2\pi), g=1/(\nu_*-\nu_0)$.\\ 

Let 
\beq
\label{sig2}
 \sigma_1=\Big({1\over m}\sum_j\sum_k\nu_{kj}^2\cos^2\theta_k\Big)^{1/2},\quad
\sigma_2=\Big({1\over m}\sum_j\sum_k\nu_{kj}^2\sin^2\theta_k\Big)^{1/2},\,\, m=m_1m_2
\eeq
be the sample variations of spatial frequency.
 Define the primary and secondary target vectors by
 \[
 X_l=\rho_\bp,\quad X'_l=-2\pi \i \ell h_{1\bp} {  \sigma_1} X_l,\quad X''_l=-2\pi \i \ell h_{2\bp} {  \sigma_2 } X_l, \quad l=(p_2-1)\sqrt{n}+p_1. \]
The signal model takes
 the form 
 \beq
 \label{53}
 Y=\bF X+\bG X'+{\mb H} X''+E
 \eeq
 subject to the support constraint
 \beq
 \label{sc2}
 \supp(X')\subseteq \supp(X),\quad \supp(X'')\subseteq \supp (X)
 \eeq 
where the measurement matrix is given by
 \beq
 \label{2dmatrix}
 && F_{il}=e^{-2\pi \i \ell \nu_{kj}\bdhat_k\cdot
 \bp},\quad G_{il}=e^{-2\pi \i \ell \nu_{kj} \bdhat_k\cdot
 \bp}\nu_{kj}\cos{\theta_k} {  \sigma^{-1}_1},\quad H_{il}=e^{-2\pi \i \ell \nu_{kj} \bdhat_k\cdot
 \bp}\nu_{kj}\sin{\theta_k}{  \sigma_2^{-1}}
 \eeq
with  $i=j+(k-1)m_2$.

In the extreme case, we select the spatial frequencies $\nu_l, l=1,...,m_2$ independently
of $\theta_k, k=1,...,m_1.$ The number of degrees of diversity is $m_1+m_2$ now  instead of $m=m_1m_2$ as for SAR scheme A.

\commentout{
First let us consider the high frequency, narrow band regime
 \beq
 \label{hfnb}
 \nu_0\gg 1/\ell,\quad  \nu_0\gg 2\alpha_1T/c_0
 \eeq
 where $\ell$ is the spacing of the imaging grid (see below).
 }
 \commentout{
 We will first discuss the ideal case of full spatial frequency band
 \beq
 \nu\in [0,\nu_*]
  \label{59}
 \eeq
 and consider the partial bandwidth case in Section \ref{sec5.3}. 
 We will return to the narrow band case (\ref{50'}) with
 \[
 \nu_0=2\om_0/c_0\gg 2\alpha_1T/c_0
 \]
 in Section \ref{sec5.4}. 
 }


\commentout{
 Note also that in the narrow-band limit 
 \beq
 \label{nb}
 (\nu_*-\nu_0)/\nu_0\ll 1, 
 \eeq
 there is little frequency diversity and 
 we have the approximation 
  \[
G_{il}\approx e^{-2\pi \i \ell \nu_{kj} \bdhat_k\cdot
 \bp}{m_1^{1/2}\cos{\theta_k}\over (\sum_k\cos^2\theta_k)^{1/2}},\quad H_{il}\approx e^{-2\pi \i \ell \nu_{kj} \bdhat_k\cdot
 \bp} {m_1^{1/2}\sin{\theta_k}\over (\sum_k\sin^2\theta_k)^{1/2}}
 \]
}

\subsection{FDMF SAR}\label{sec:wide}
For  PDMF SAR, we can also use the specialized scheme:\\

{\bf SAR scheme  B:} For $k=1,...,m$ we select  $\nu_k$ and $\theta_k$ together by
solving 
 \beq
 \ell\nu_k\cos\theta_k=Qa_k,\quad \ell\nu_k\sin\theta_k=Qb_k,
 \label{48}
 \eeq
for a fixed $Q\in\IN$ where  $(a_k,b_k), k=1,\ldots,m,$ are i.i.d.
 uniform random variables on $[-1/2,1/2]^2$.\\
 
  Eq. 
(\ref{48}) always has a solution
 in (\ref{59}) under the condition
 \beq
 \nu_0=0,\quad \ell\nu_*\geq Q/\sqrt{2}.
 \eeq
 On the other hand,  for PDMF SAR, $\nu_0\neq 0$ and eq. (\ref{48}) may not have a solution for sufficiently small $a_k$ and $b_k$.


With (\ref{48}), the measurement matrix is given by 
\beq
 \label{2dmatrix'}
 &&F_{kl}=e^{-2\pi \i Q (a_k,b_k)\cdot 
 \bp},\quad G_{kl}=e^{-2\pi \i Q (a_k,b_k)\cdot 
 \bp} a_k{  \sigma^{-1}_1},\quad H_{kl}=e^{-2\pi \i Q(a_k,b_k)\cdot 
 \bp} b_k{  \sigma_2^{-1}}
 \eeq
 with $l=(p_2-1)\sqrt{n}+p_1$ and 
 \beq
\label{sig2'}
 \sigma_1=\Big({1\over m}\sum_{k=1}^ma_k^2\Big)^{1/2},\quad
\sigma_2=\Big({1\over m}\sum_{k=1}^mb_k^2\Big)^{1/2}.
\eeq
\commentout{
In this setting, the primary and auxiliary  targets  are given by 
 \[
 X_l=\rho_\bp,\quad X'_l=-2\pi \i \ell h_{1\bp} {  \sigma_1} X_l,\quad X''_l=-2\pi \i \ell h_{2\bp} {  \sigma_2 } X_l.
 \]
 }
The measurement matrix  $\bA=[\bF\ \bG \ \bH]$ with (\ref{2dmatrix'}) is a two-dimensional version of (\ref{10}) and hence  
satisfies the coherence bound analogous to
Lemma \ref{lem1}.

\commentout{

  Now we state the coherence bound on the system (\ref{53}). 
 
  \begin{lemma}\label{lemm4}
Assume that look angles $\theta_k\in [0,2\pi]$ are
i.i.d. with the probability density function $\phi(\theta)$.

Suppose
\beq
\label{m-2}
n\leq {\delta\over 8} \e^{K^2/2},\quad \delta, K>0.
\eeq
Then   the sensing matrix satisfies the coherence bound
\beq
\label{mut}
\mu< 
\bar\mu+{\sqrt{2} K\over \sqrt{m_1}}
\eeq
 with probability greater than $(1-\delta)^2$
 where in general $\chii$ (resp. $\chis$) satisfies the bound
 \beq
 \label{21-3}
&\bar\mu\leq {c_\gamma}{ {(1+\nu_0\ell)}^{-1/2}} \|\phi\|_{\gamma,\infty}
\label{21-4}
 \eeq
 where $\|\cdot\|_{\gamma,\infty}$ is the H\"older norm
 of order $\gamma>1/2$ and the constant $c_\gamma$  depends only on $\gamma$.
 \end{lemma}
Lemma \ref{lemm4} is analogous to the single frequency coherence bound
given in \cite{cis-simo}.

 Note that the above coherence bound does not require
full, circular view of the scene, but the smoothness of
the sampling density function $g$ which depends
indirectly on the size of the support of $g$. 

To see how spatial frequency sampling  can improve
the coherence bound, consider the case of circular SAR with
randomly and uniformly distributed looks on $[0,2\pi]$. 
Expanding $\phi$ in the Fourier series
\[
\phi(\theta)=\sum_{l} c_le^{\i l\theta}
\]
and denoting the angle between $\bp'$ and $\bp$ by $\measuredangle{\bp'\bp}$
we can bound $\bar\mu$ (cf. (\ref{39-3})) by 
\beq
\label{56}
&&\max_{\bp\neq \bp'}\max_{j}\Big|\sum_l c_l \e^{\i l \measuredangle{\bp'\bp}} \int^{2\pi }_0a_j(\theta) e^{\i l\theta}
e^{2\pi\i \ell \nu_j |\bp-\bp'|\cos\theta }d\theta\Big|\nn
\eeq
where
\[
a_{j}(\theta)\in \{1,\ \nu_j\cos\theta/\sigma_1,\ \nu_j\sin\theta/\sigma_2,\ \nu_j^2\cos^2\theta/\sigma_1^2,\ 
  \nu_j^2\sin^2\theta/\sigma_2^2,\ \nu_j^2\cos\theta\sin\theta/(\sigma_1\sigma_2)\}
    \]
which involves  $J_l(2\pi\ell\nu_j |\bp-\bp'|)$, the Bessel function of order $l$, and their derivatives. 

On the one hand, the worst case bound is  
\beq
\label{asym}
J_l(2\pi \ell \nu_j |\bp-\bp'|) < {c\over \sqrt{\ell \nu_0}},\quad
\forall \bp\neq \bp', \quad\forall j,\quad \forall l
\eeq
for some $c>0$ in consistence with (\ref{21-3}). 
On the other hand, the large argument asymptotic for $J_l$ 
(corresponding to the regime of high frequency, $\nu_0\ell\gg 1$)
\beq
\label{304}
J_l(z)=\sqrt{2\over \pi z}\lt\{ \cos{(z-l\pi/2- \pi/4)}
+\cO(|z|^{-1})\rt\},\quad z\gg 1
\eeq
suggests a similar reduction of $\bar\mu$
under (\ref{303}) by judicious choice of
spatial frequencies. 
If for example the standard resolution criterion
\beq
\label{303}
\ell\alpha_1T/c_0\geq 1/2
\eeq
is satisfied, then the argument $2\pi\ell\nu_j|\bp-\bp'|$ ranges over
at least one period, with $\nu_j$ in (\ref{50'}) 
and $\bp,\bp'\in \cL$,  and 
one can deliberately select a sequence
of spatial frequencies to minimize (\ref{56})
 in view of
the sinusoidal nature of the leading asymptotic in (\ref{asym}).
Ideally  the leading order term in (\ref{asym}) 
may 
be cancelled  out yielding the improved bound
\[
\bar\mu=\cO((\ell\nu_0)^{-1}).
\]

}

    \commentout{
For immobile targets, the signal model \eqref{SM_conv} 
takes the form of  convolution. Assume the reflectivity function 
  $\rho(t)$ has a $s$-sparse coefficient representation $\vec{x}=\set{x_j}$, $j=1,2,\ldots,n$, under an ONB $\Phi = \set{\varphi_j}$; the sparse coefficients can be recovered
  by random convolution \cite{Romberg09}.
  
  Let 
  $$f(t) = \sum_{j=1}^n \hat{f}_j\e^{2\pi\i t\omega_j}\cdot\ind_{[0,T]}(t).$$
  where $\omega_j = (j-1)\Delta \omega$ be the frequency samples, 
  $B = 2\pi n\Delta\omega$ the bandwidth of $f$, and $\hat{f}_j$, $j=1,2,\ldots,n$ are defined as follows
  \cite{Romberg09}:
  \begin{itemize}
   \item $j=1$ ($\omega=0$): Define the DC component $\hat{f}_1 \sim \pm 1$ with equal probability.
   \item $j=2,3,\ldots,\frac{n}{2}$ ($0<2\pi\omega<\frac{B}{2}$): 
    $\hat{f}_j = \e^{2\pi\i x_j}$ where $x_j\sim\text{unif}[0,1)$.
   \item $j=\frac{n}{2}+1$ ($2\pi\omega=\frac{B}{2}$):  $\hat{f}_{\frac{n}{2}+1} \sim \pm 1$ with equal probability.
   \item $j=\frac{n}{2}+2,\ldots,n$  ($\frac{B}{2}<2\pi\omega<B$): $\hat{f}_j = \hat{f}_{n-j+2}^\ast$.
  \end{itemize}
  The sensing scheme can be written as a underdetermined system
  $\vec{y} = A \vec{x}$ where
  $A=R_m \frac{1}{\sqrt{n}}F^\ast\Sigma F \Phi \in\mathbb{C}^{m\times n}$ is the sensing matrix, with
  $R_m:\mathbb{C}^{n\times n}\to \mathbb{C}^{m\times n}$ the projection operator that select
  $m$ different rows from a matrix; $\Sigma = \diag(\hat{f}_1,\ldots,\hat{f}_n)\in\mathbb{C}^{n\times n}$, 
  and $F\in\mathbb{C}^{n\times n}$, $F_{j,\ell} = \e^{-2\pi\i(j-1)(\ell-1)/n}$ is
  the discrete Fourier matrix,
  $j,\ell=1,2,\ldots,n$. 
  \begin{proposition}
   \cite{Romberg09} Consider the noiseless case.  If
   $$m\geq C\cdot(s\log n + \log^3 n)$$ where $C$ is a universal constant, then the Basis Pursuit 
   recovers $X$ exactly with probability exceeding $1-\mathcal{O}(n^{-1})$.
  \end{proposition}

  Within the setting, we design an wave envelope with duration $T$ that covering all possible time delay region 
  of the targets. Given the sampling interval $\Delta t = T/n$, it yields that $\Delta\omega = \frac{2\pi n}{T}$
  w.r.t.\ the time-frequency relation of discrete Fourier transform, and
  bandwidth $B = \frac{(2\pi n)^2}{T}$. The operator $R_m$ randomly picks $m$ points from $n$ time-grids.
  }

 \subsection{ 2D SCOMP}
{SCOMP with the two-dimensional support constraint is given as follows.
     \begin{center}
   \begin{tabular}[width=4in]{||l||}
   \hline
   \centerline{{\bf Algorithm 2.}\quad  2D SCOMP } \\ \hline
   Input: $\bF,\bG, \bH, Y, \|E\|_2$\\
 Initialization:  $X^0 = 0, R^0 = Y$ and $\cS^0=\emptyset$ \\ 
Iteration:\\
\quad 1) $i_{\rm max} = \hbox{arg}\max_{l}\Big(|F^*_{l}R^{k-1}|+|G^*_{l}R^{k-1} |+|H^*_{l}R^{k-1} |\Big)$\\
 \quad      2) $\cS^k= \cS^{k-1} \cup \{i_{\rm max}\}$ \\
  \quad  3) $(X^k, X^{'k}, X^{''k}) = \hbox{arg} \min\|
     \bF Z+\bG Z'+\bH Z''-Y\|_{2}$,\\
      \hspace{4cm} s.t. $\supp(Z'), \supp(Z'')\subseteq  \supp(Z) \subseteq S^k$ \\
  \quad   4) $R^k = Y- \bF X^k-\bG X^{'k} -\bH X^{''k}$\\
\quad  5)  Stop if $\|R^k\|_{2}\leq \|E\|_2$.\\
 Output: $\hat X=X^k, \hat X'=X^{'k}, \hat X''=X^{''k}$. \\
 \hline
   \end{tabular}
\end{center}
\bigskip
From $\hat X, \hat X', \hat X''$, we can recover the off-grid perturbation by
\[
\hat \bh_\bp={i\over 2\pi \ell X_l} \Big({X_l'\over \sigma_1}, {X_l''\over \sigma_2}\Big), \quad l=(p_2-1)\sqrt{n}+p_1.
\]
 
 Performance guarantee similar
 to Theorem \ref{thm2} follows the same line of argument 
 given in Appendix \ref{app:omp}. 

 \begin{theorem}Suppose  that the columns of $\bA=[\bF\ \bG\ \bH]$  have same 
  2-norm. 
 Let  $\hbox{supp} (X)=\{J_1,\ldots, J_s\}$
and 
\beqn
X_{\rm max}=|X_{J_1}|+|X'_{J_1}|+|X''_{J_1}|&\geq& |X_{J_2}|+|X'_{J_2}|+|X''_{J_2}|\geq \cdots \\
&\geq &|X_{J_s}|+|X'_{J_s}|+|X''_{J_s}|=X_{\rm min}. 
\eeqn
Suppose \beq
\label{snr3}
(6s-1)\mu + \frac{6\|E\|_2}{X_{\rm min}} < 1
\eeq
and let $\hat X$ and $\hat X'$ be the output of Algorithm 2. Then 
\[
\supp(\hat X)=\supp(X)
\]
and
\beq
\label{err:omp3}
\|\hat X-X\|^2_{2}+\|\hat X'-X'\|_2^2+\|\hat X''-X''\|_2^2 \leq {{3}\|E\|_2^2\over {1-\mu (3s-1)}}.
\eeq
\label{thm22}
\end{theorem}
\begin{remark}
Analogous to Remark \ref{rmk1}, we can improve the accuracy of recovery
by performing the nonlinear least squares
\beq
\label{nls2}
\arg \min \sum_{k, j}\Big| y(\nu_{kj},\theta_k)-\sum_{\bp\in \hbox{\tiny \rm supp} (\hat X)}\rho_\bp \e^{-2\pi \i \ell \nu_{kj} \bdhat_k\cdot
 (\bp+\bh_\bp)} \Big|^2
\eeq
in the set of all $\{\rho_\bp: \bp\in \supp (\hat X)\}\subset \IC^s$ and
$\{\bh_\bp:\bp\in \supp (\hat X)\} \subset (-0.5,0.5)^{2s}$ with
the SCOMP estimates as initial guess. 
\end{remark}
}

 The greedy algorithm for the 3-dimensional setting and its performance guarantee
 can be  analogously formulated. For the sake of brevity, we will not pursue them here.  

  \section{  Numerical experiments}
  \label{sec:num}
 
  In the following simulations, we use $s=10$ complex-valued targets with
 random  amplitudes
 \[ 1 + (c_1+ i c_2)/\sqrt{8}
 \]
 where $c_1,c_2$ are standard normal random variables.
 We set the grid spacing $\Delta\tau, \ell =1$ 
and let off-grid perturbations $\{\xi_k\}$ be  i.i.d. uniform random variables in $[-0.4, 0.4]$. 
 In all our simulations, we add $1\%$ external noise
to the data and so the signal-to-noise ratio (SNR) is 100. 
  
  The  BP estimates (solved with YALL1 \cite{YALL1}) tend to be ``bushy" and require ``pruning." To take
  advantage of the prior knowledge of sparsity and the support constraint (\ref{sc})
  we apply the technique of Locally Optimized Thresholding (LOT) to the BP estimates $\hat X, \hat X'$
  as follows. In addition to pruning (i.e. thresholding), LOT also locally adjusts
 the reconstruction to minimize the residual subject to the support constraint. 
    \begin{center}
   \begin{tabular}[width=5in]{||l||}
   \hline 
   \centerline{{\bf Algorithm 3.}\quad Locally Optimized Thresholding (LOT)}  \\ \hline
    Input: $\hat X, \hat  X'$, $\bA=[\bF\ \bG], Y, s=\hbox{\rm target sparsity}$.\\
Iteration:  Set $\cS^0=\emptyset$. For $n=1,2,...,s$\\
\quad 1) $i_n= \hbox{arg}\,\,\max_j \big( |\hat X_j|+|\hat X_j'|(2\pi\sigma Q)^{-1}\big),\ \hbox{\rm s.t.} \  j\not\in \cS^{n-1} $.\\
 \quad 2)  $\cS^n=\cS^{n-1}\cup\{i_n\}$.\\
 
Output: 
$(\tilde X, \tilde X') = \hbox{arg}\,\,\min_{z}\|\bF Z+\bG Z'-Y\|_2,$ \
s.t. $ \supp (Z')\subseteq \cS^s,  \supp (Z)\subseteq \cS^s. $\\
     \hline
   \end{tabular}
\end{center}
\bigskip

\begin{remark}
As in (\ref{nls}) we can improve the accuracy of the LOT estimates by performing the nonlinear least squares
\beq
\min\sum_j\Big|f(t_j)-  f_{\text{LC}}(t_j) \sum_{k\in \cS^s} \rho_{k} f_{\text{LC}}(-k\Delta\tau-\Delta\tau\xi_{k})
   \e^{-2\pi\i Q\bar  t_j k}\e^{-2\pi\i Q \xi_{k} \bar t_j}\Big|^2\label{nls3}
   \eeq
   in  the set of all $\{\rho_{k}:k\in \cS^s\}\subset \IC^s$ and $\{\xi_{k}: k\in \cS^s\}\subset (-0.5,0.5)^s$ with 
the LOT estimates as  the initial guess
for iterative methods (e.g. the Gauss-Newton method or gradient methods) for (\ref{nls}). 
\end{remark}

      \begin{figure}[t]\centering
  \subfigure[BP]{ \includegraphics[width=8cm]{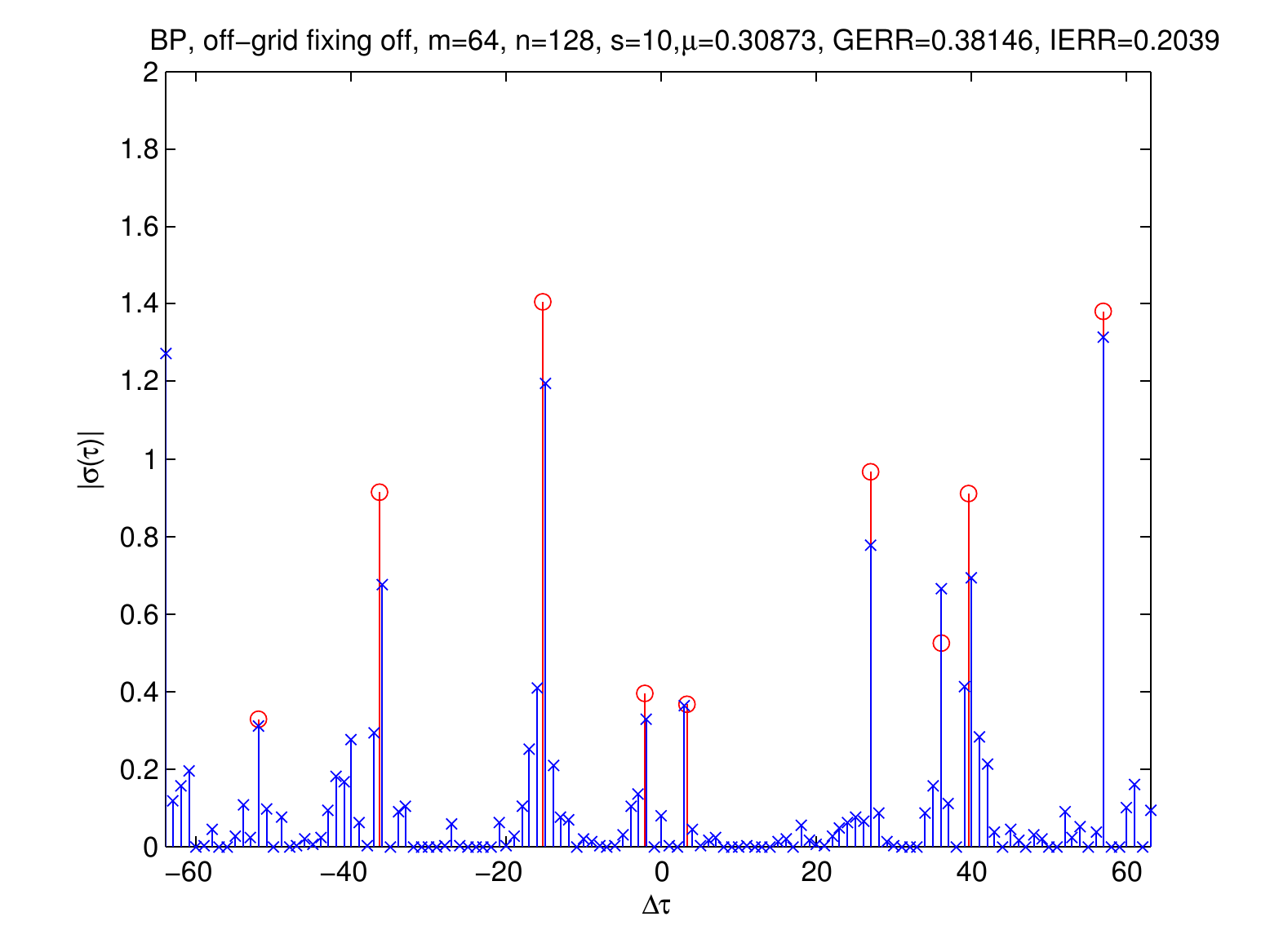}}
   \hspace{-1cm}
       \subfigure[OMP]{        \includegraphics[width=8cm]{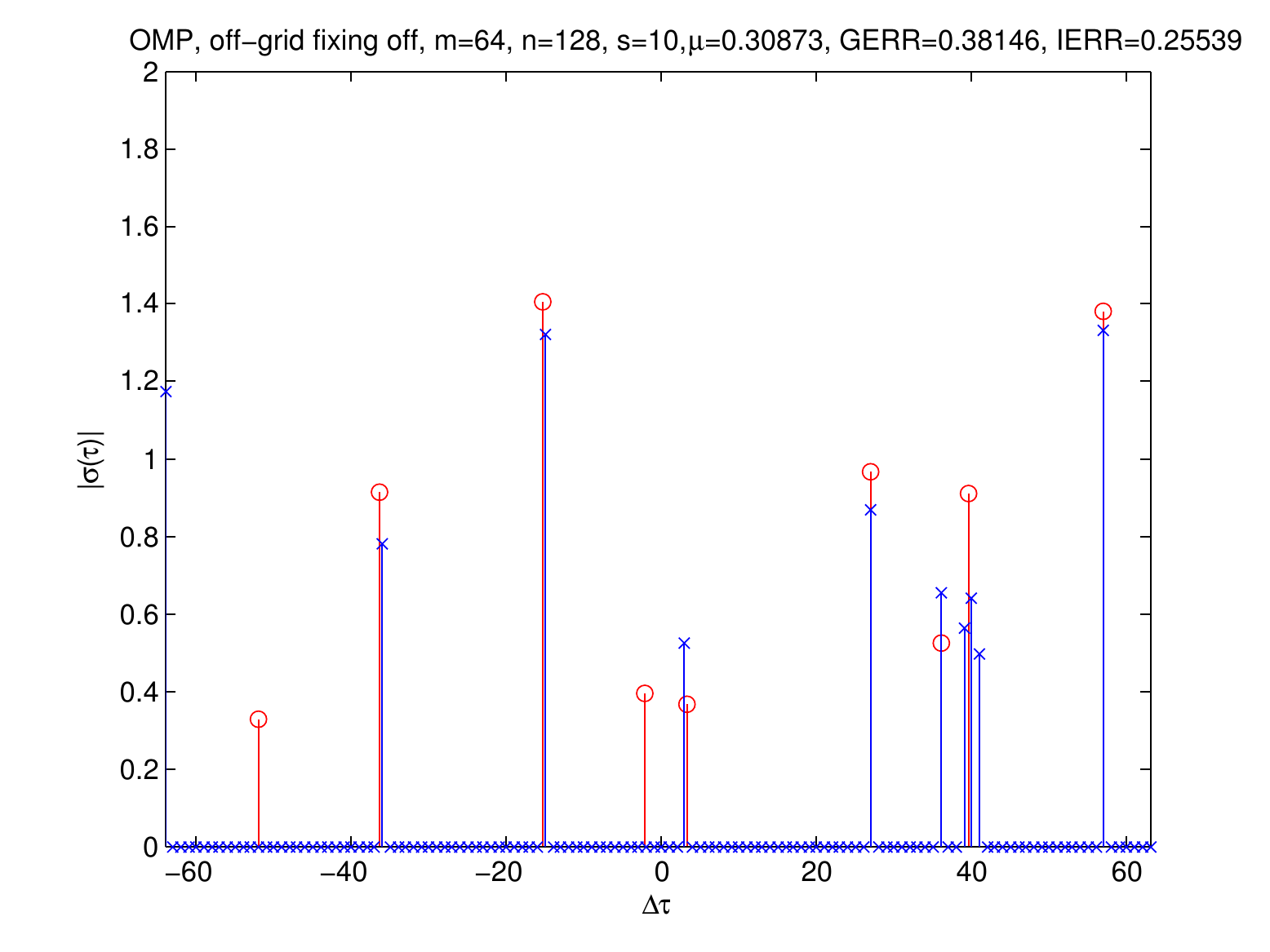} 
}
     \subfigure[BPLOT]{    \includegraphics[width=8cm]{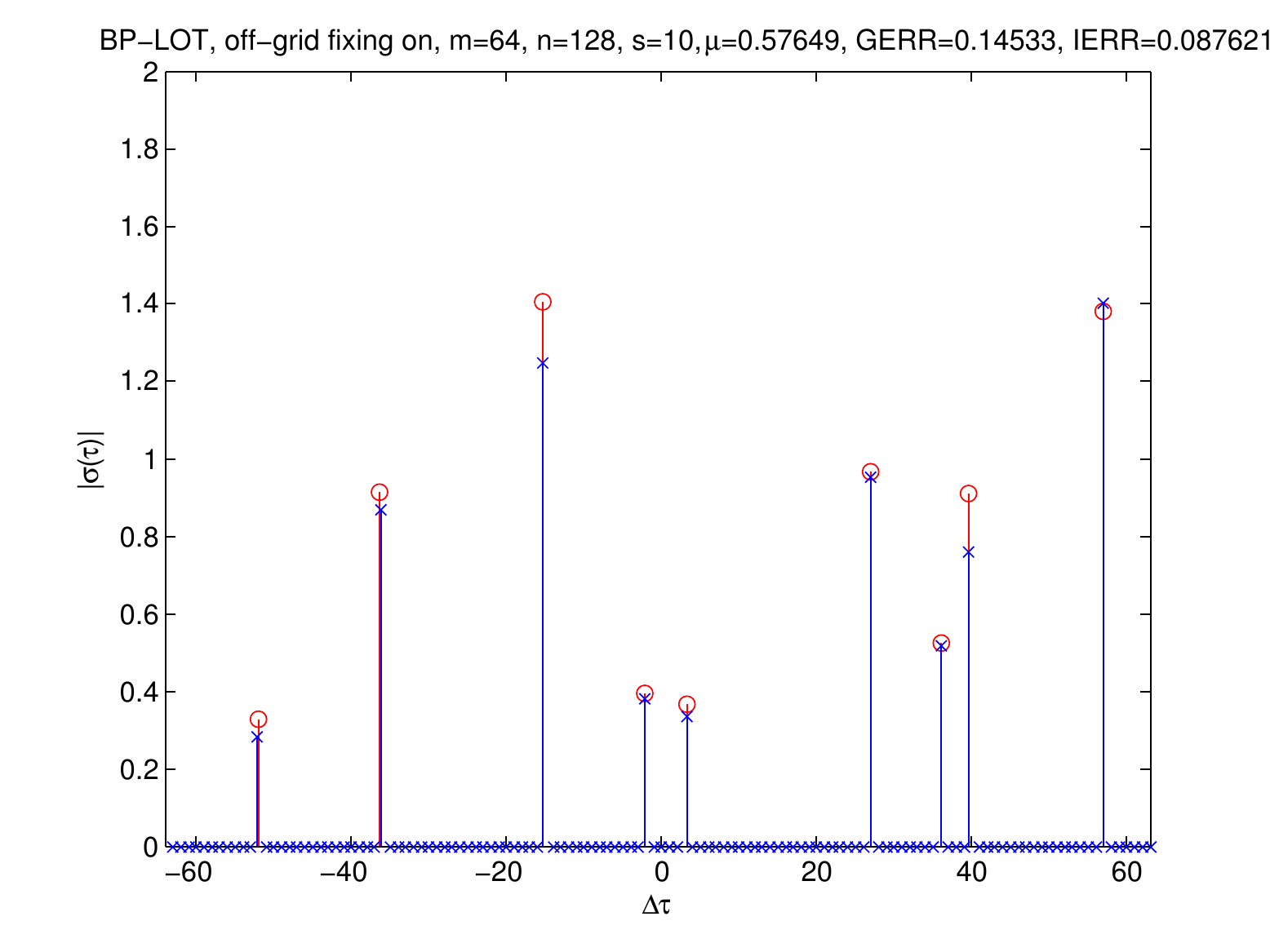}}
\hspace{-1cm}
\subfigure[SCOMP]{  \includegraphics[width=8cm]{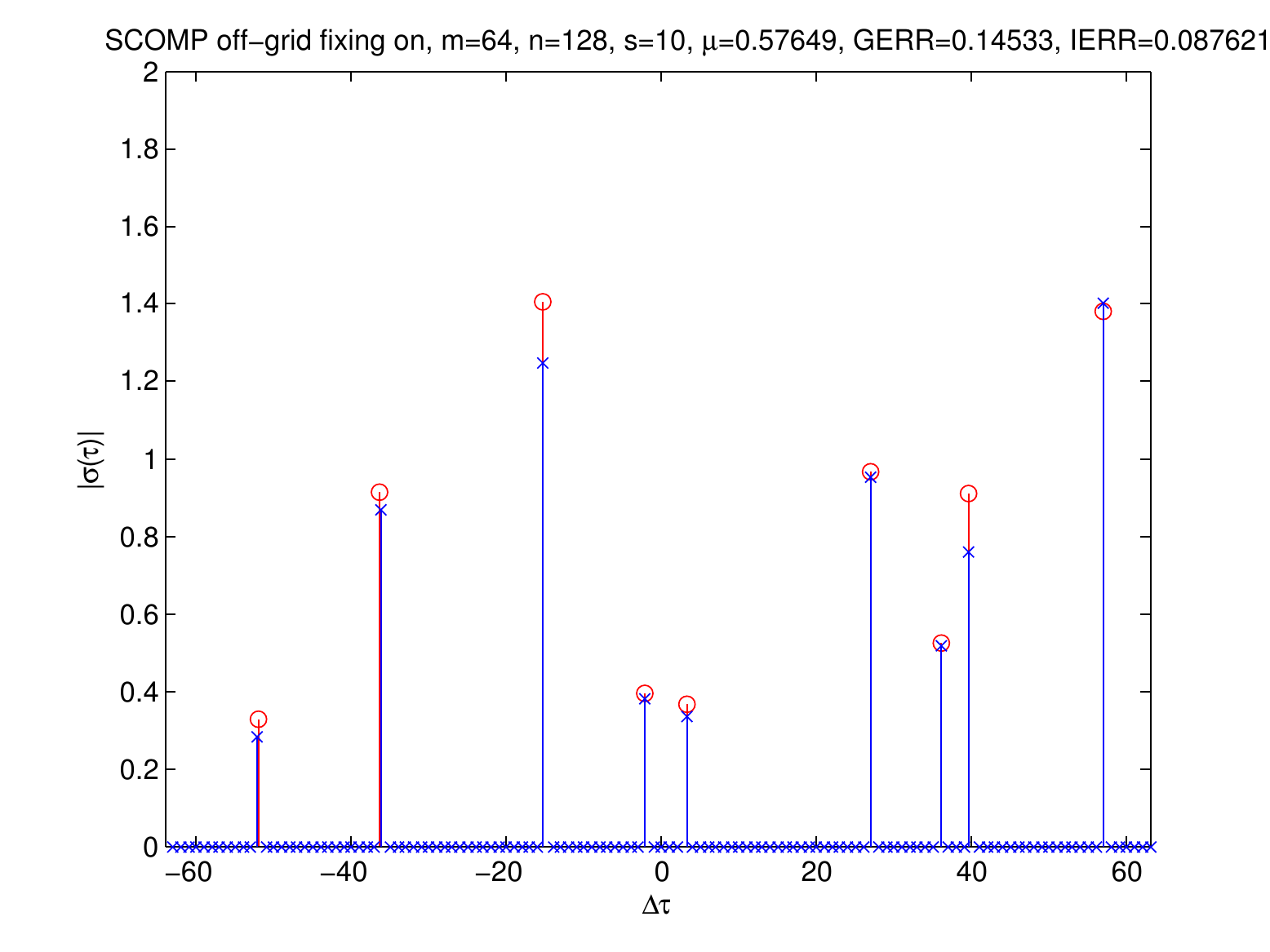} }
   \caption{Radar ranging  (blue crosses) with $Q=1$ of off-grid targets (red circles) by  (a) BP,  (b) OMP with eq. (\ref{10'}) and (c) BPLOT, (d) SCOMP
with eq. (\ref{7'}).  
    }
   \label{fig-new1-1}
  \end{figure}

The idea of LOT is similar to  that  of the Band-excluded Locally Optimized Thresholding (BLOT)
proposed in \cite{FL2,Asilomar12} except without the band-exclusion step which is not needed here since  the grid is well resolved.  For brevity, we shall denote the combined
algorithm of BP followed by LOT
as BPLOT. 

Successful recovery for OMP, SCOMP, BPLOT is defined as the recovery of target support
to the grid accuracy, i.e. 
$\supp(\hat X)=\supp(X)$.  For BP, a recovery is counted as successful if 
$\supp(\hat X^s)=\supp(X)$ where $\hat X^s$ is  the best  $s$-sparse
approximation 
of the BP recovery $\hat X$ (i.e. thresholded BP).
We distinguish two versions of thresholded BP: the grid-corrected version
and the uncorrected version (``fixON BP" and ``fixOff BP", respectively, in the legend of  Fig. \ref{fig-new12-2}, \ref{fig4-2}, \ref{fig6}(b) and \ref{fig5-2}(b)).

When  recovery is successful, we
measure the degree of success by the  (relative) recovery error  $\|X-\hat X\|_2/\|X\|_2$ in the case of OMP, SCOMP, BPLOT
and by $\|X-\hat X^s\|_2/\|X\|_2$ in the case of BP. Note that for the system (\ref{7'}) 
\[
\|X-\hat X\|_2/\|X\|_2= \sum_k\Big| {f_{\rm LC}(-k\Delta \tau- \xi_k\Delta\tau)\over f_{\rm LC} 
(-k\Delta\tau-\hat\xi_k\Delta\tau)} e^{\pi i Q(\hat \xi_k -\xi_k)} \rho_k -\hat\rho_k\Big|^2/\|\rho\|_2
\]
while for the system (\ref{53}) $\|X-\hat X\|_2/\|X\|_2=\|\hat\rho-\rho\|_2/\|\rho\|_2.$ 

      \begin{figure}[t]\centering
  \subfigure[BP]{ \includegraphics[width=8cm]{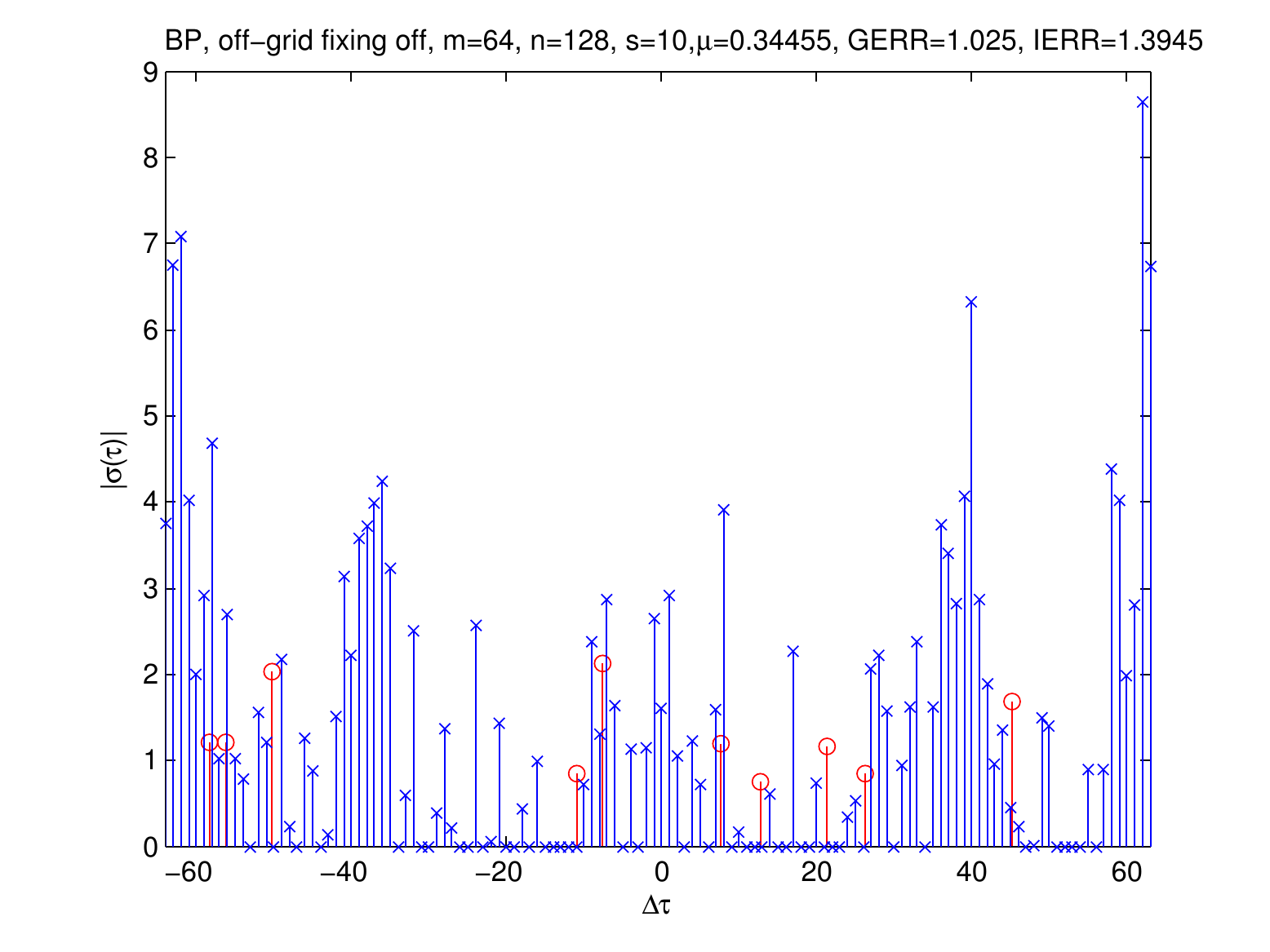}}
   \hspace{-1cm}
       \subfigure[OMP]{        \includegraphics[width=8cm]{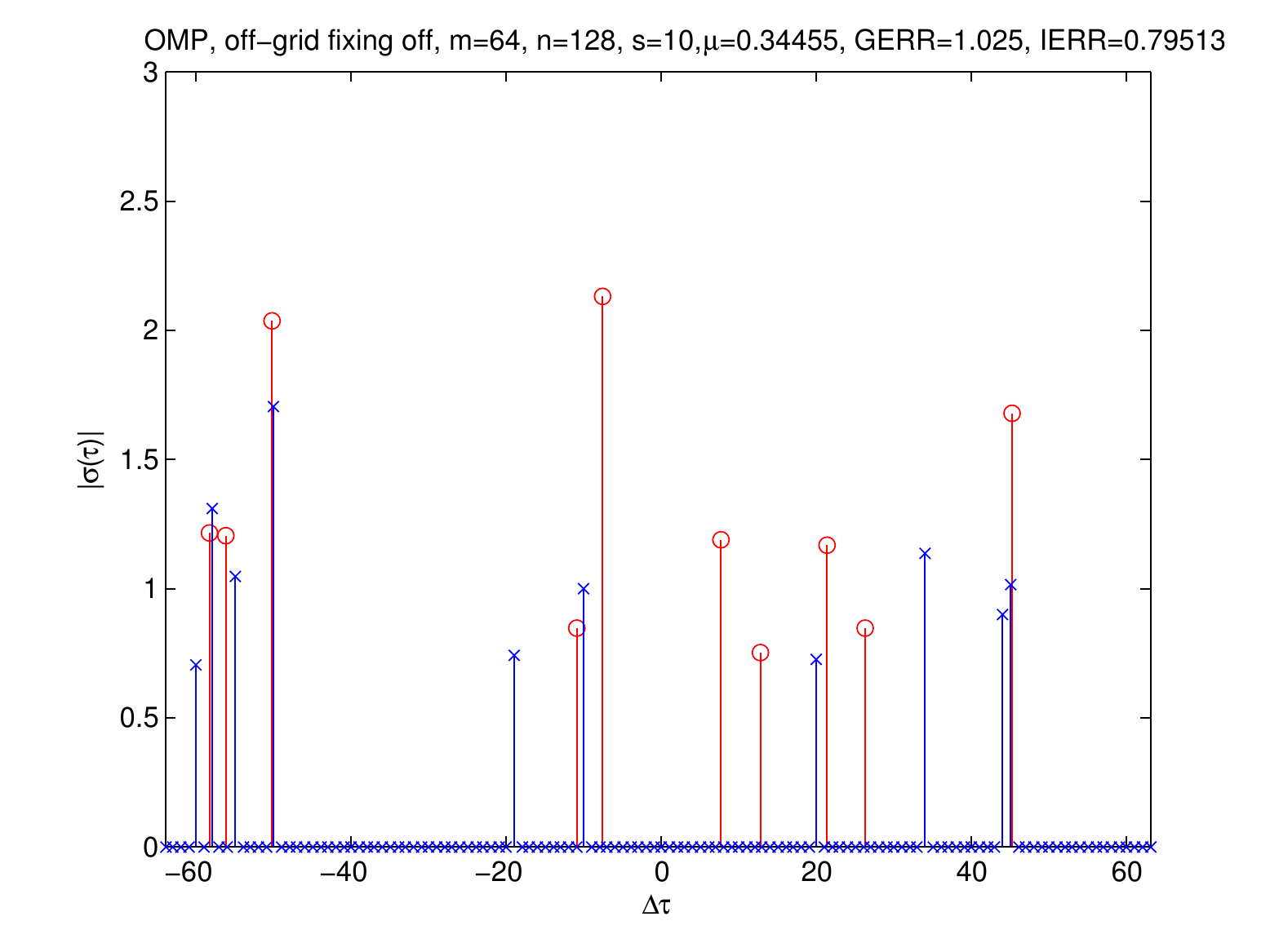} 
}
     \subfigure[BPLOT]{    \includegraphics[width=8cm]{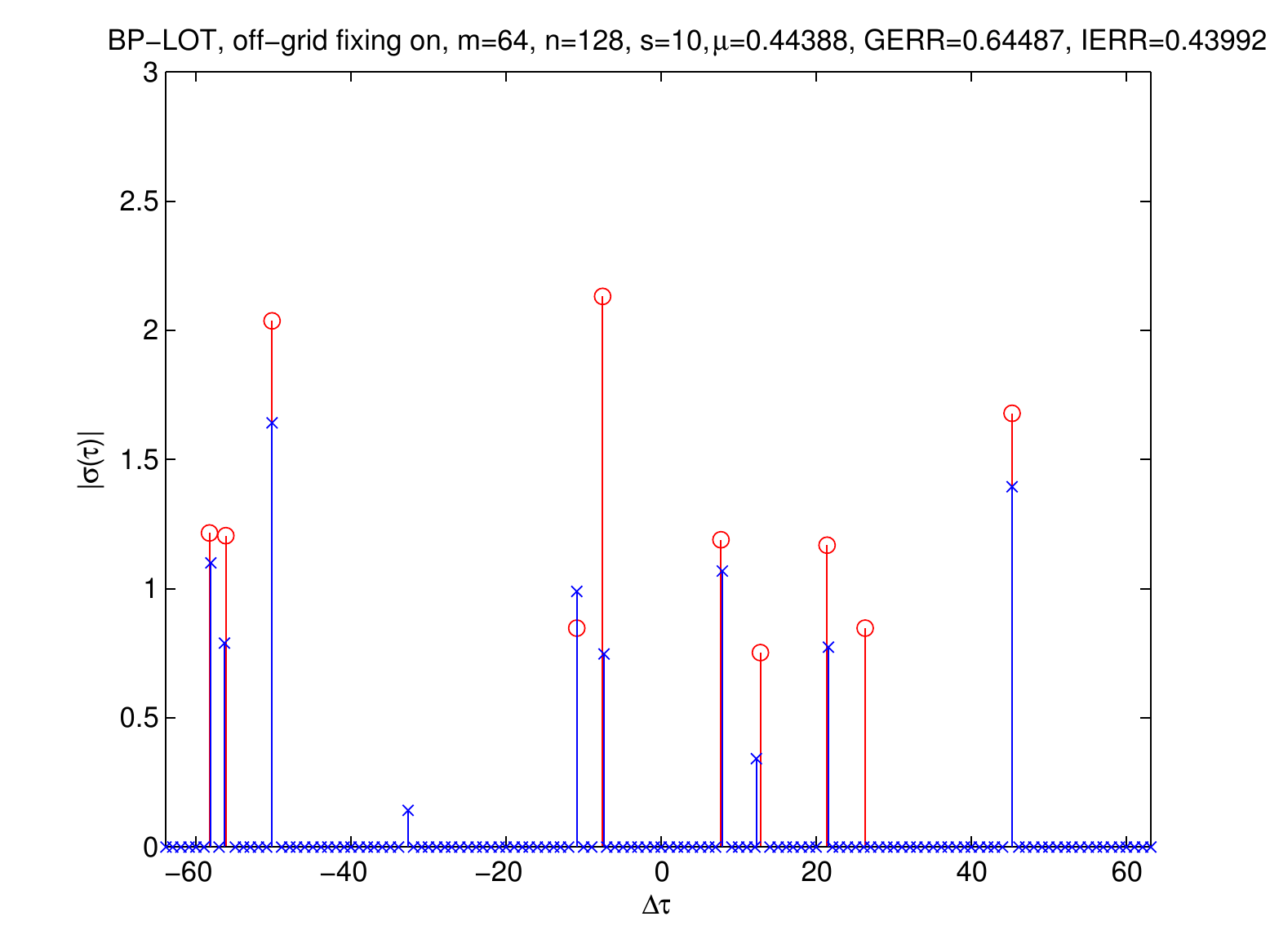}}
\hspace{-1cm}
\subfigure[SCOMP]{  \includegraphics[width=8cm]{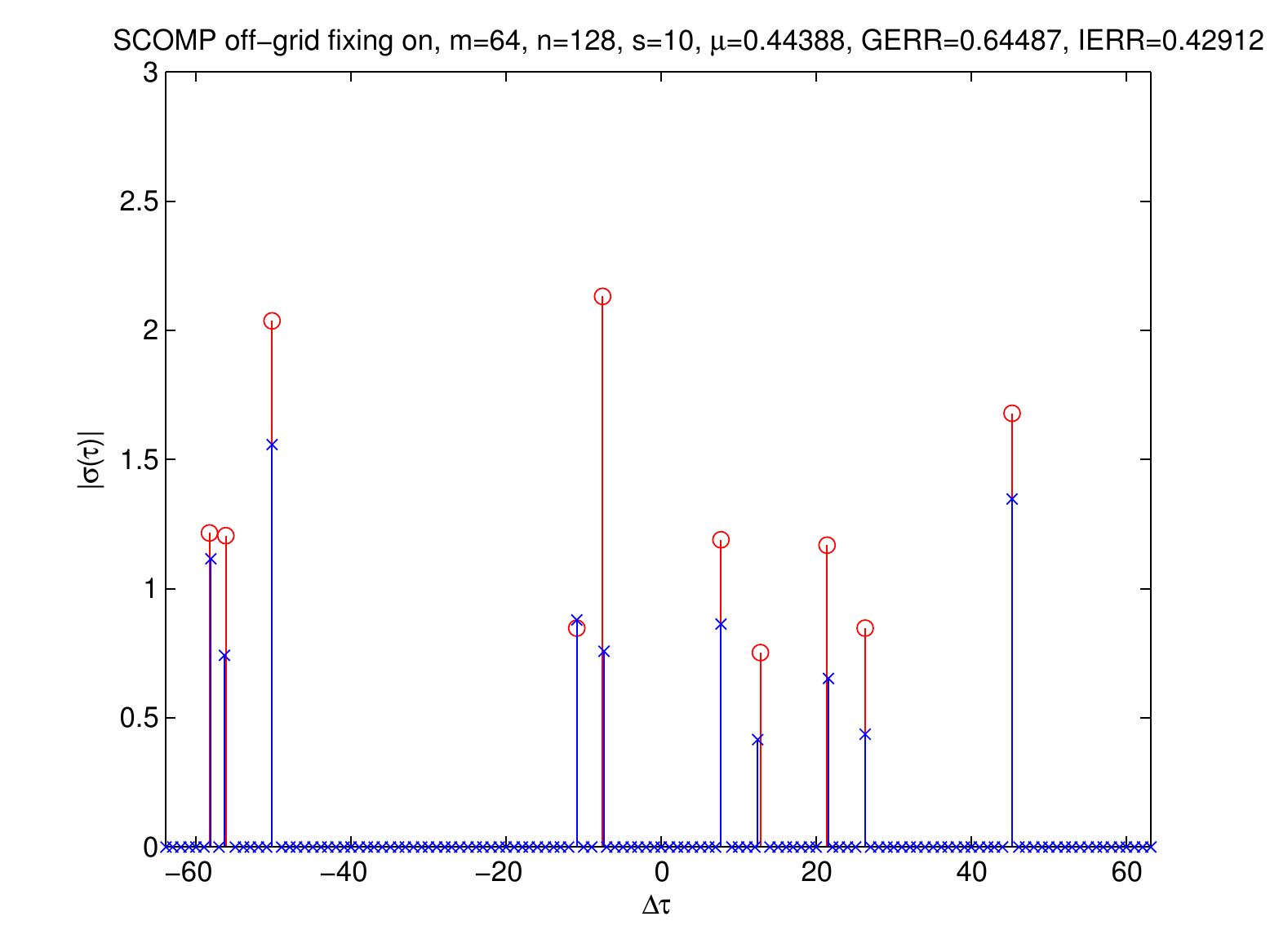} }
   \caption{Radar ranging with $Q=2$ by  (a) BP,  (b) OMP without grid correction  and (c) BPLOT, (d) SCOMP
with grid correction.  
    }
   \label{fig-new1-2}
  \end{figure}

For radar ranging \eqref{SM_Range_pt},  we set the parameters $m=64, n=128$ and $Q=1, 2$  (Fig.\ref{fig-new1-1}-\ref{fig-new12-2}). 
  The gridding error for the formulation (\ref{10'}) is a whopping $38.1\%$ for $Q=1$ and $103\% $ for $Q=2$  while
  that for (\ref{7'}) is  $14.5\%$ for $Q=1$ and 64.5\% for $Q=2$ which still seem large. But surprisingly  BPLOT (Fig.\ref{fig-new1-1}(c)) and
  SCOMP (Fig.\ref{fig-new1-1}(d) \&  \ref{fig-new1-2}(d)) can locate the targets  to the grid accuracy, producing error of $8.8\%$ for Fig.\ref{fig-new1-1}(c) \& (d) and $43\%$ \ref{fig-new1-2}(d). Note that the second target from the right is missed by BPLOT in Fig.\ref{fig-new1-2}(c). 
 By contrast   BP (Fig.\ref{fig-new1-1} (a) \& \ref{fig-new1-2} (a))
  and OMP (Fig.\ref{fig-new1-1}(b) \& \ref{fig-new1-2} (b)) poorly locate  the targets for both $Q=1\&2$.   
  
  Fig.\ref{fig-new12-1} shows how the NLS technique can further improve
  the performance of SCOMP.  The error for SCOMP-NLS is
  2.2\% and 31.2\%, respectively, for $Q=1$ (Fig.\ref{fig-new12-1}(a)) and  $Q=2$ 
  (Fig.\ref{fig-new12-1}(b)). The worsening performance as $Q$ increases from 1 to 2 is probably 
  due to the increasing  gridding error. 
  
Fig.\ref{fig-new12-2} shows   the success rate  computed out of 100 independent trials as a function of the target sparsity 
  with the support recovery to the grid accuracy  as the criterion for success. For each trial,
  the target support, amplitudes, time samples and external noise are independently chosen with
  the same sparsity.

\commentout{
Since we only seek within the grid accuracy for support recovery, we
compute the relative error by first convolving the error with 
the compactly supported bell-shaped function

\beq
&&g_a(x) =\Big\{\begin{matrix}
 \exp{[- a^2/( a^2 - x^2)]} / (0.4439938162a), & |x|<a\\
             0 , &\hbox{else}
             \end{matrix} 
             \eeq
   which satisfies 
$\int g_a(x) d x =1$.  We set $a=0.4$. 
}

For $Q=1$ (Fig.\ref{fig-new12-2}(a)) BPLOT has the best performance
while for $Q=2$ (Fig.\ref{fig-new12-2}(b))  SCOMP is the best performer.  
Both BPLOT and SCOMP outperform 
 both BP and  OMP without grid correction.

    \begin{figure}[t]\centering
  \subfigure[SCOMP-NLS with $Q=1$]{  \includegraphics[width=8cm]{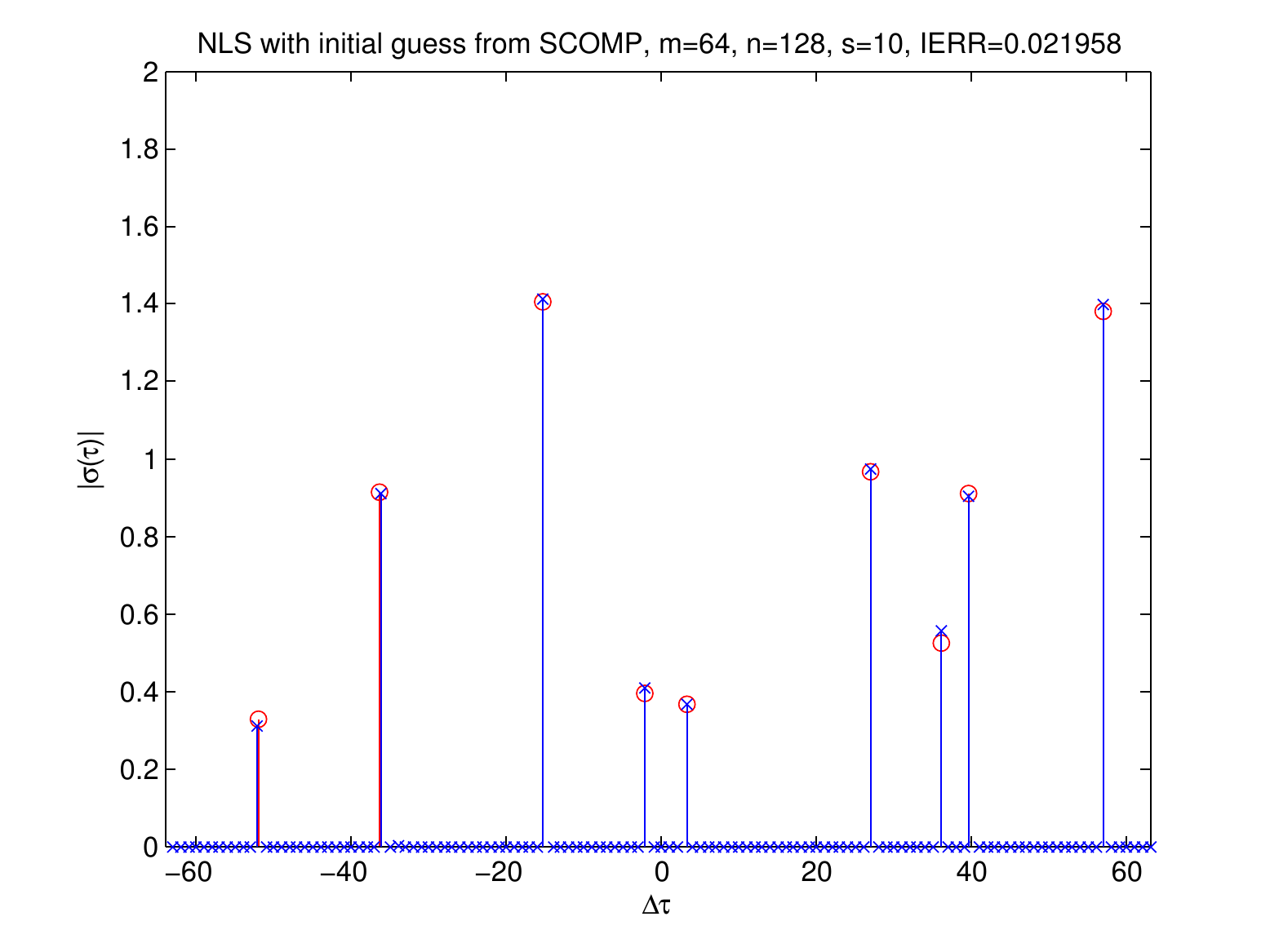}}
    \subfigure[SCOMP-NLS with  $Q=2$]{  \includegraphics[width=8cm]{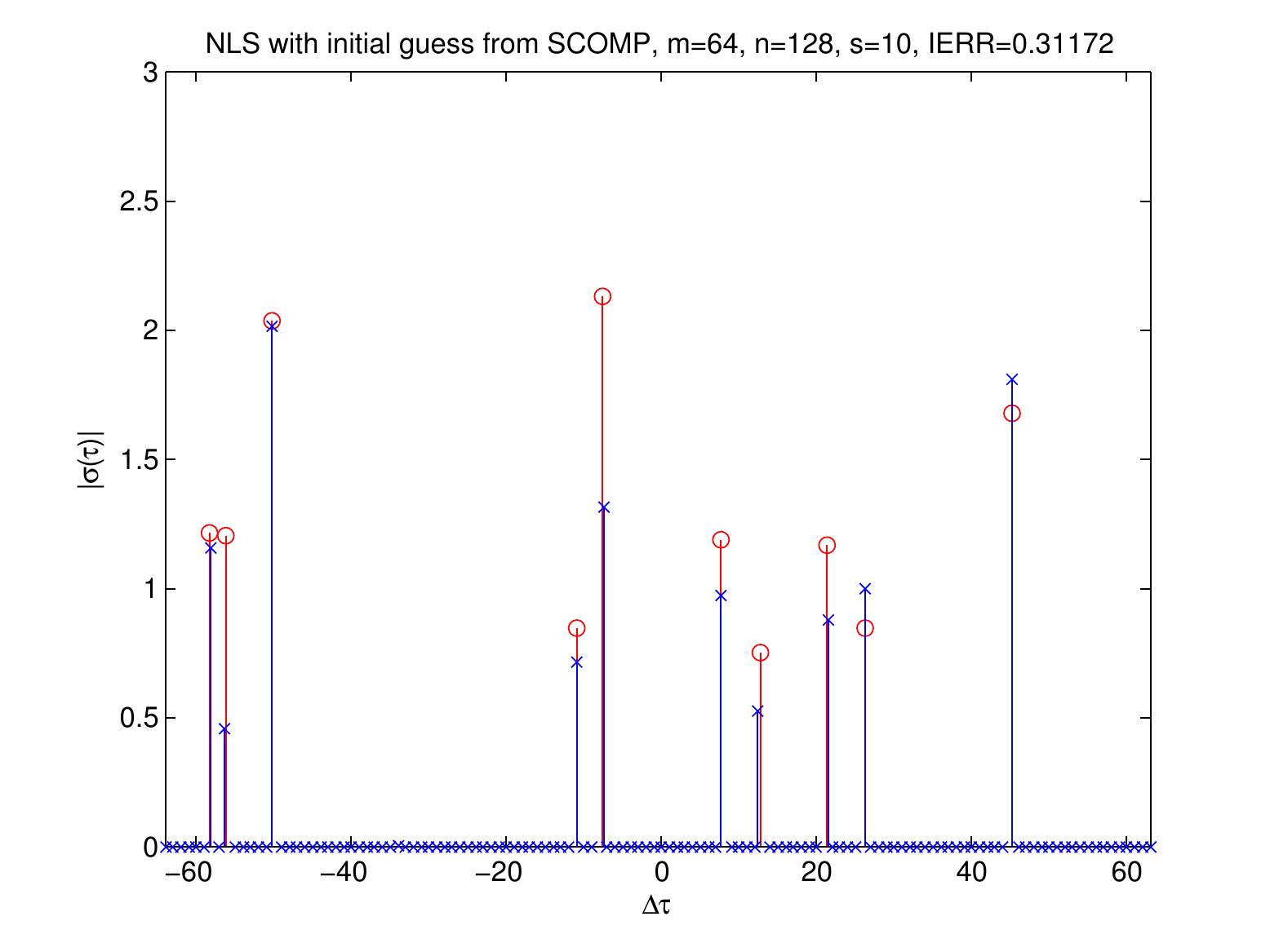}}
  \caption{SCOMP-NLS  produces ranging error of (a) $2.2\%$ with $Q=1$ and (b) $31.2\%$ with $Q=2$.}
  \label{fig-new12-1}
  \end{figure}
  \begin{figure}[h!]\centering
  \subfigure[$Q=1$]{\includegraphics[width=8cm]{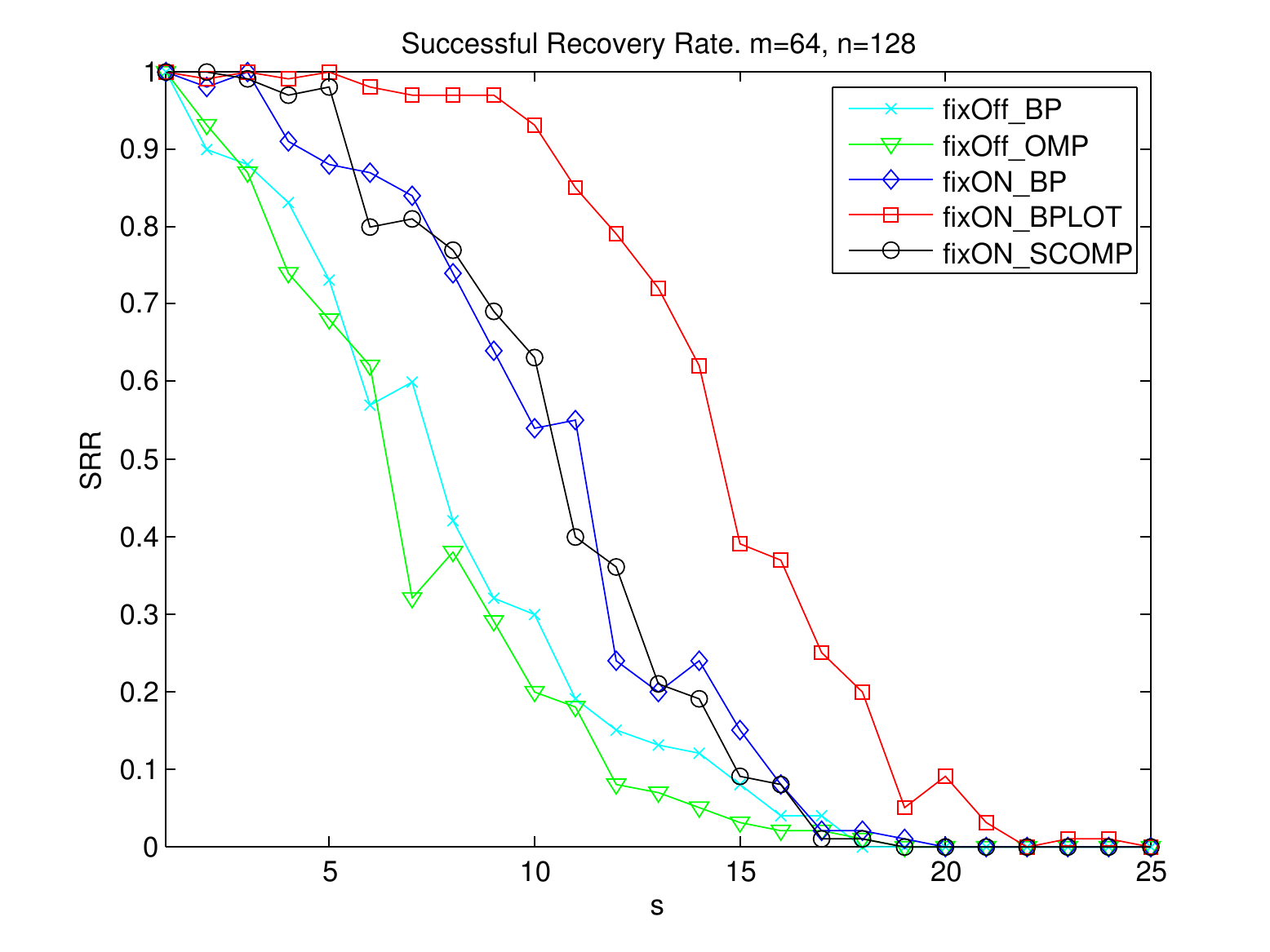}}
    \subfigure[$Q=2$]{\includegraphics[width=8cm]{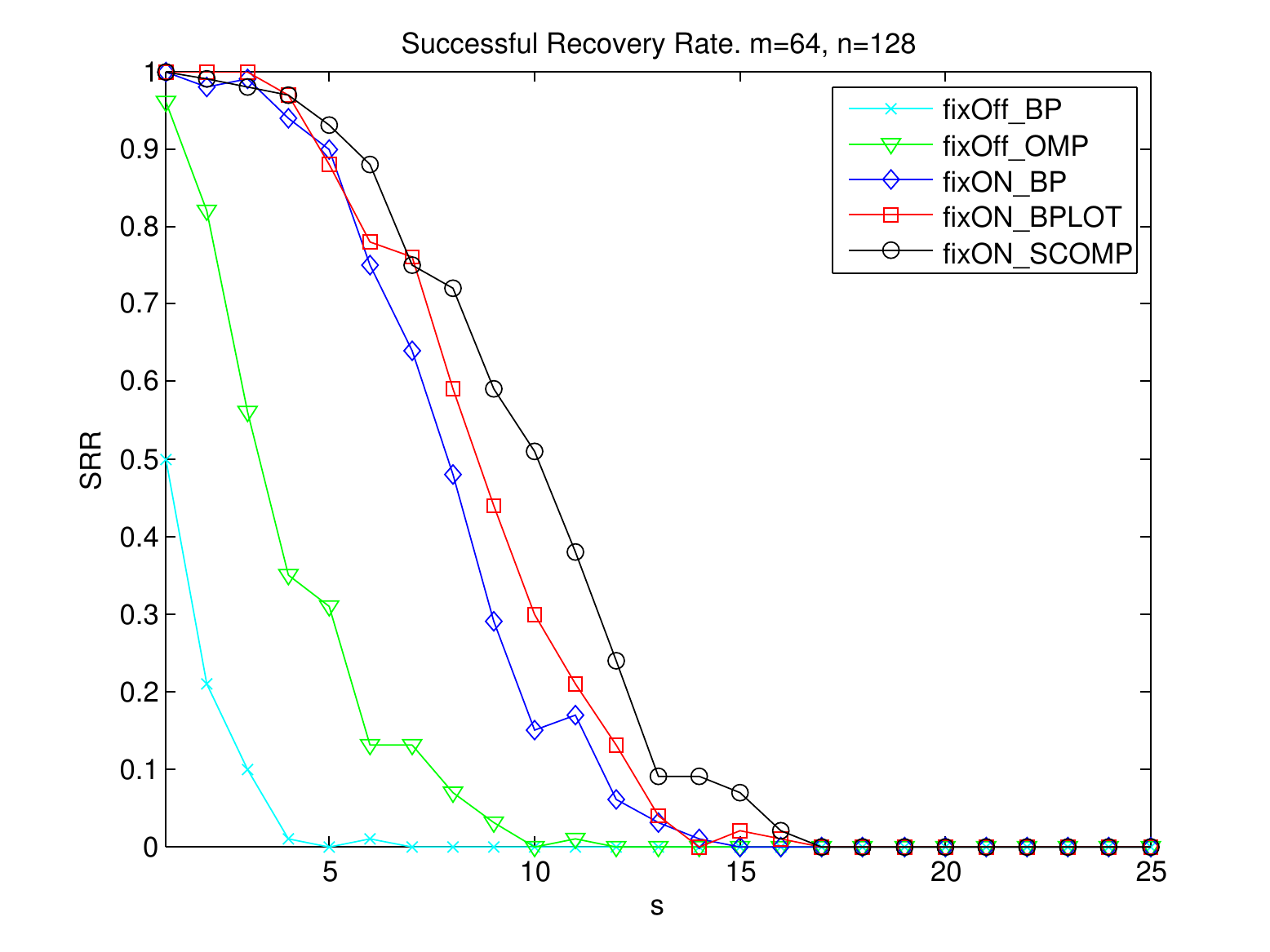}}
  \caption{Success rate of  ranging  versus sparsity  with (a) $Q=1$ and (b)  $Q=2$.
   In the legend, ``fixOff" means recovery {\em without} grid correction 
and ``fixOn" means recovery {\em with} grid correction.  }
  \label{fig-new12-2}
  \end{figure}

\commentout{We add complex circularly Gaussian white noise
to the data: $\vec{e}=\set{e_{j,R}+\i e_{j,I}}$, $j=1,2,\ldots,m$,
     $e_{j,R},e_{j,I} = (1\%)\cdot\sigma_{\mathrm{N}}\frac{\norm{\vec{y}}_2}{\sqrt{m}}$ where 
     $\sigma_{\mathrm{N}}\sim\mathcal{N}(0,1)$ is i.i.d.\ chosen. 
}

  \commentout{
  For the Spotlight SAR imaging, we set 
 $ n = 625,  m_1=m_2=8, 
\nu_0 = 4, \nu_*= 5, \ell = 1, c_0=1$ so that $\nu_0\ell =4>1$ and  (\ref{303}) is satisfied with equality. The look angles are i.i.d. uniform random variables in $[0,2\pi]$ and the
spatial frequencies are i.i.d. uniform random variables in $[4, 5]$.  In this high-frequency regime, the gridding error is amplified in proportion to
the relative bandwidth $(\nu_*-\nu_0)/\nu_0$ which is 4 in this case. Consequently, we
only allow  the maximal off-grid perturbations of size $0.05\ell $ in the simulation. 

As shown in Fig.\ref{fig6}, only SCOMP can determine  the targets to the grid accuracy in this instance. 
The success rate (out of 100 independent trials) is shown in Fig.\ref{fig7}(a) where
the advantage of SCOMP appears  primarily in the ``shoulder" part of the curve. 
The nonlinear least squares (\ref{nls2}) improves the quality of the SCOMP recovery, reducing
the relative error of $26.6\%$ in Fig.\ref{fig6}(d) to $0.4\%$ in Fig.\ref{fig7}(b). 
}

  \begin{figure}[t]\centering
 \subfigure[BP]{ \includegraphics[width=8cm]{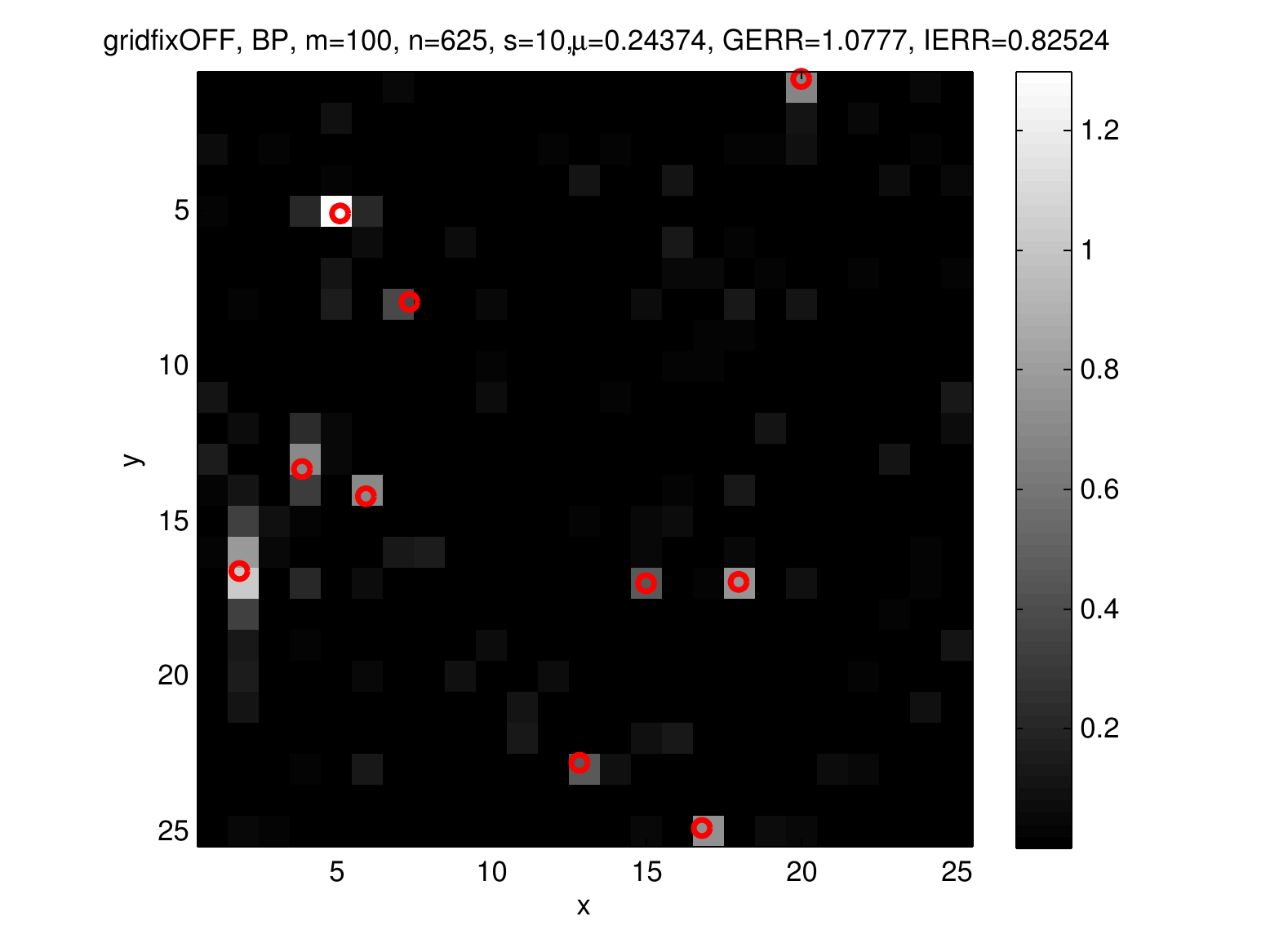}}
 \subfigure[OMP]{\includegraphics[width=8cm]{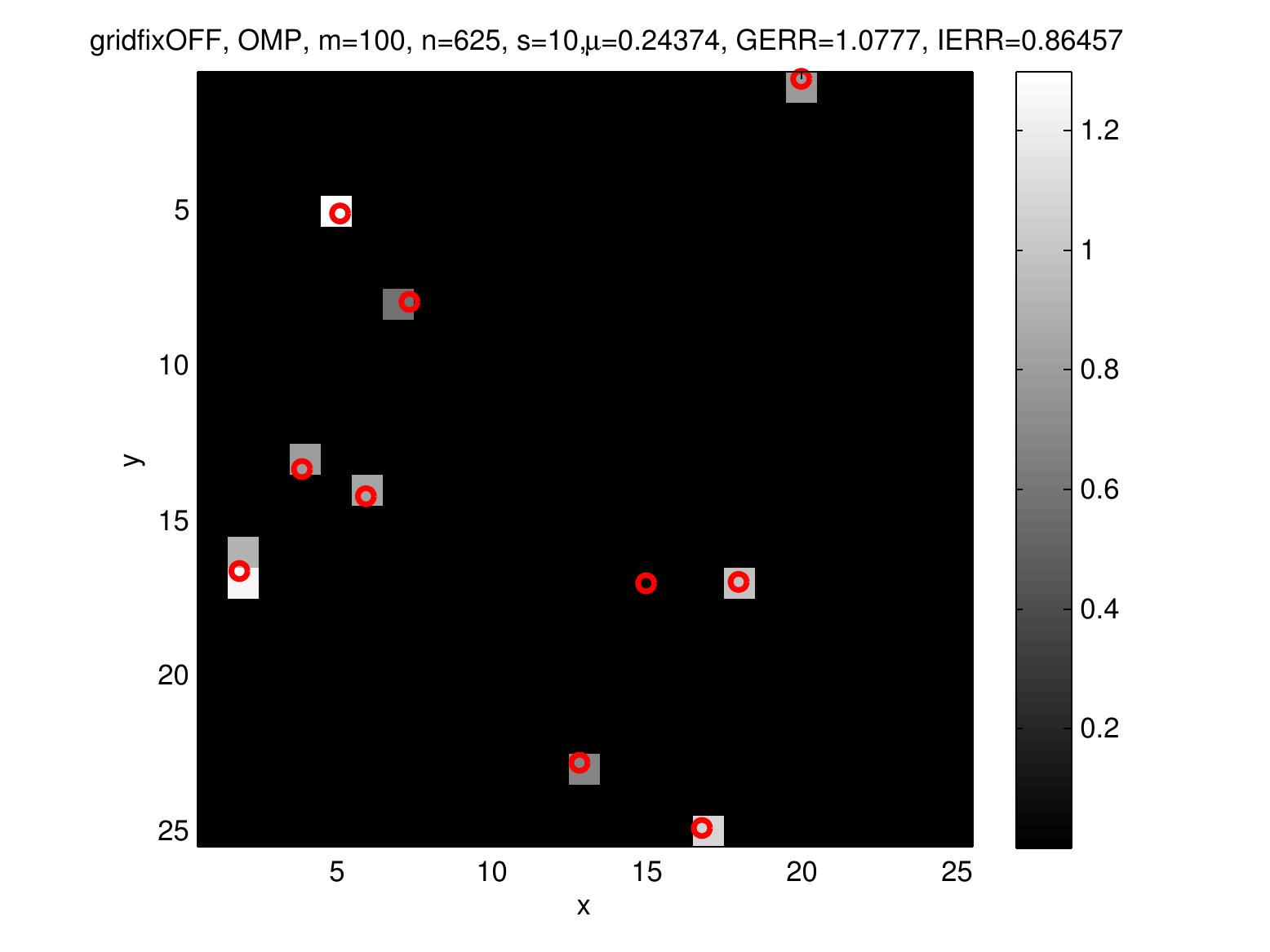}}
 \subfigure[BPLOT]{\includegraphics[width=8cm]{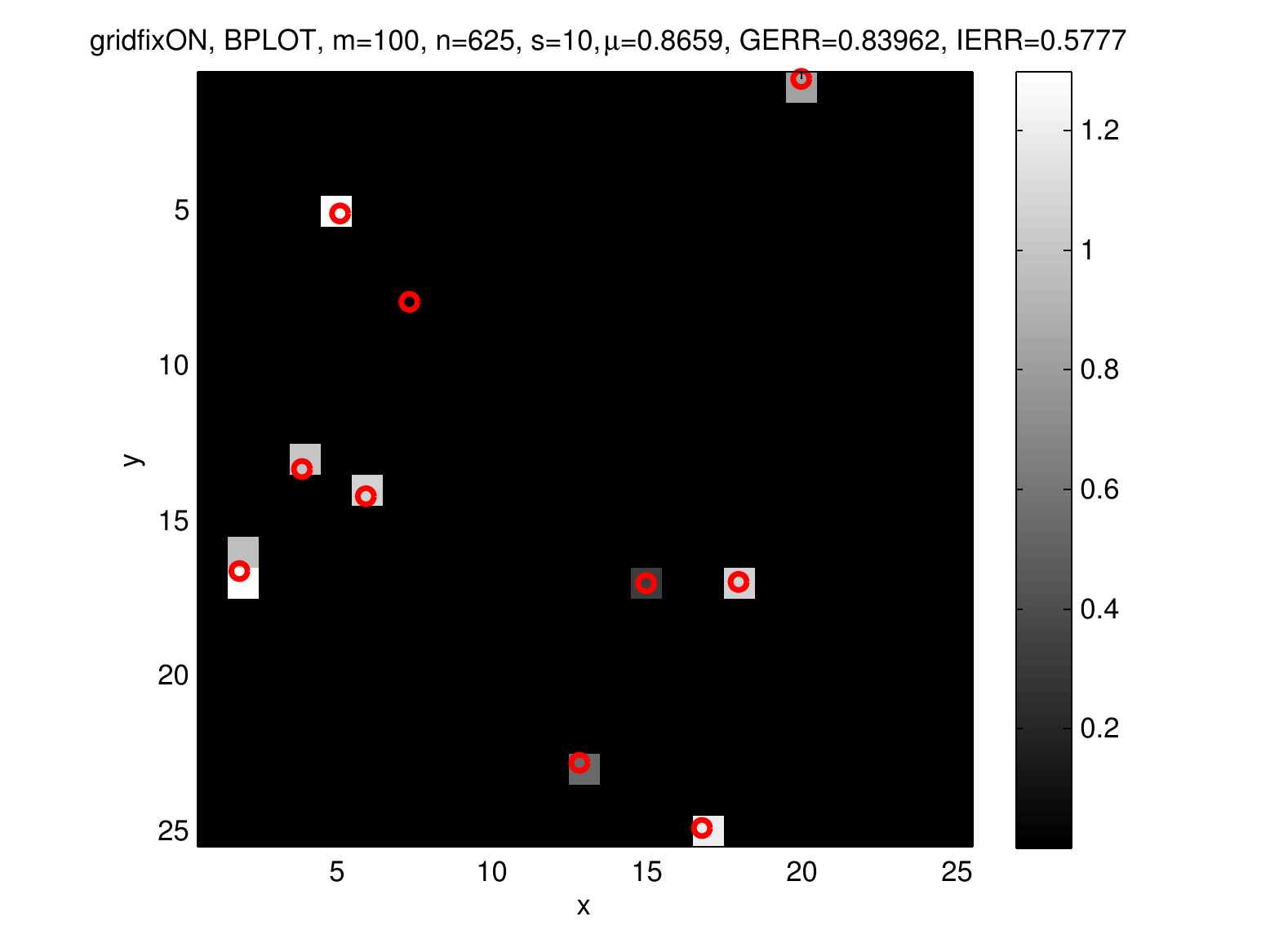}}
  \subfigure[SCOMP]{\includegraphics[width=8cm]{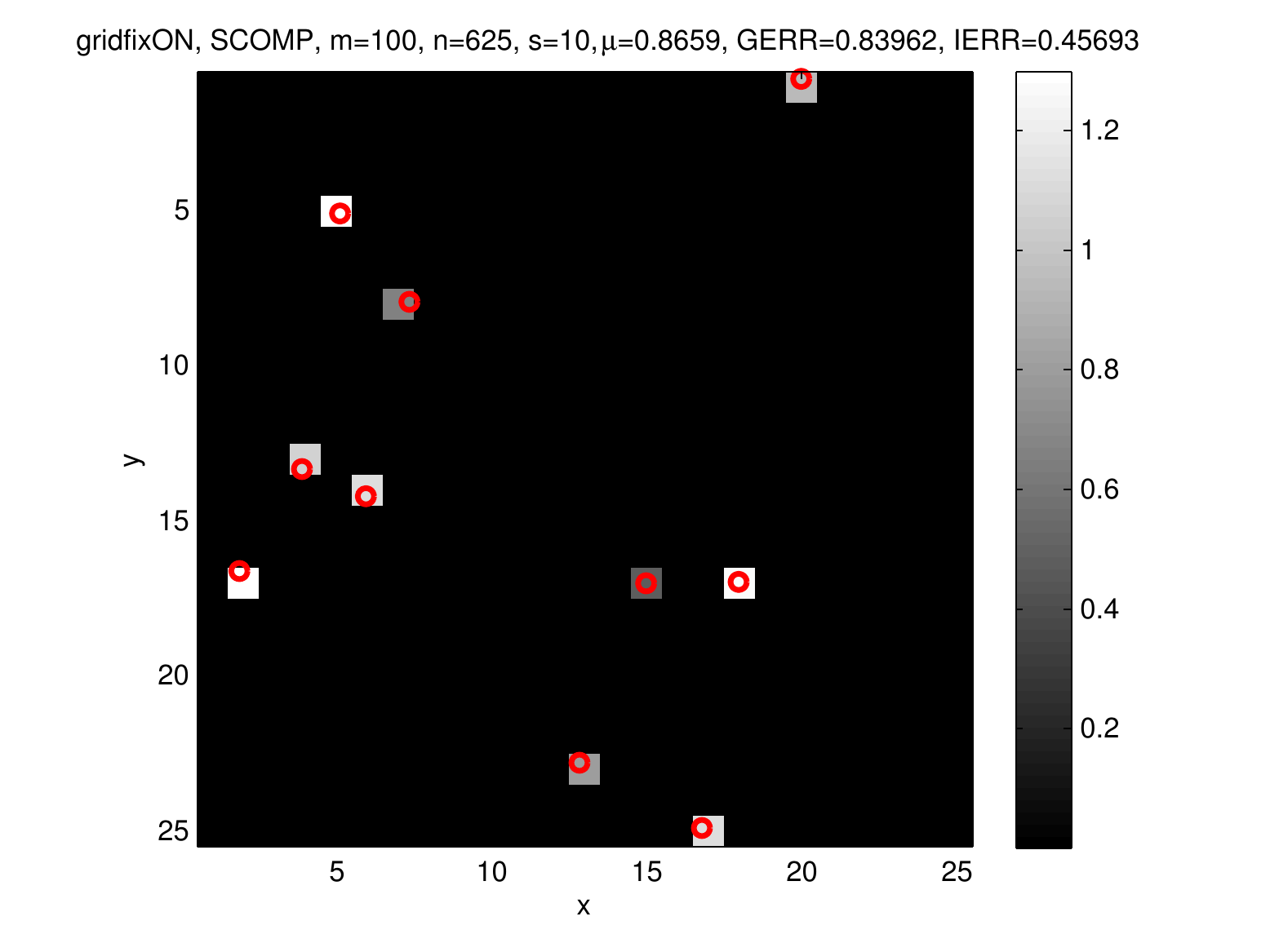}}
  \caption{FDMF  SAR imaging (white spots) of off-grid targets (red circles) with SAR scheme B and $Q=1$ by (a) BP, (b) OMP, both without
  grid correction, 
   and (c) BPLOT, (d) SCOMP, both with grid correction. }
   \label{fig3-1}
  \end{figure} 

For 2D Spotlight SAR in the FDMF regime, we use SAR schemes A \& B (Sections \ref{sec:narrow}) with
  $ n = 625,  m=100, 
\nu_0 = 0, c_0=1$   (Fig.\ref{fig3-1}-\ref{fig4-2}). For $Q=1$, we set $\nu_*=1/\sqrt{2}$. For $Q=2$, we set $\nu_*=\sqrt{2}$. For SAR scheme A, we set
$\phi=1/(2\pi), g=1/(\nu_*-\nu_0)$ here and below.

First we consider the SAR scheme B. Fig.\ref{fig3-1} shows that only SCOMP locates the targets to the grid accuracy. Note that OMP and BPLOT miss the target located at around (15, 17) in (b) \& (c), respectively. 
In Fig.\ref{fig3-2} with $Q=2$, BPLOT and SCOMP have the same results, locating the targets
to the grid accuracy. The relative error is 45.7\% for Fig.\ref{fig3-1}(d) and
24.9\% for Fig.\ref{fig3-2}(d).  After applying NLS to the SCOMP estimates, the error is reduced to
0.7\% for $Q=1$ and 0.3\% for $Q=2$ (Fig.\ref{fig4-1}). This is a rate instance where the gridding and recovery errors are smaller with $Q=2$ than $Q=1$ and reminds us the subtle dependence of the gridding error on target and measurement configurations.

Fig.\ref{fig4-2} shows the success rate versus sparsity computed out of 100 independent trials.
For both $Q=1$ and $Q=2$, SCOMP has the best performance. It is also clear
from Fig.\ref{fig4-2}, the results with $Q=1$ are better than those with $Q=2$
for all tested methods, despite the fact that the former's bandwidth  $1/\sqrt{2}$ 
is smaller than the latter's $\sqrt{2}$.

Fig.\ref{fig6} shows the results of FDMF SAR  ($m_1=m_2=10, \nu_0=0,\nu_*=1$) with the  SAR scheme A
 which is easier to implement
than the SAR scheme B (Section \ref{sec:narrow}). The purpose is to compare the performance of
the two sampling schemes. From Fig.\ref{fig4-2}(a) and \ref{fig6}(b) we find that  with the SAR  scheme A,
the performance of SCOMP worsens  while the performances of  grid-corrected thresholded BP and BPLOT  improve.  Note, however, that the bandwidth ($=1$) for Fig. \ref{fig6}(b) is larger
than that ($=1/\sqrt{2}$) for Fig. \ref{fig4-2}(a). 

For PDMF Spotlight SAR in Fig.\ref{fig5-1}-\ref{fig5-2}, we use the SAR scheme A with 
  $ n = 625,  m_1=m_2=14, 
\nu_0 = 1/2, \nu_*=1, c_0=1$, resulting in the fractional bandwidth $2/3$. Note that the total number of data $m=196$ almost doubles  that for the FDMF case.  

In Fig.\ref{fig5-1},  only SCOMP manages to locate the targets to the grid accuracy (BPLOT misses
the target located around (19, 22)), yielding an error of $49.9\%$.  The error is reduced
to 25.4\% by NLS (Fig.\ref{fig5-2}(a)). The success rate plot in Fig.\ref{fig5-2}(b) shows that
BPLOT and SCOMP have a similar, best performance, with the grid-corrected thresholded BP   trailing closely
behind. 

From  Fig.\ref{fig6}(b) and \ref{fig5-2}(b), we find that grid-corrected thresholded BP, BPLOT and SCOMP have
comparable performances with the SAR scheme A. 
Also, the similarity of the success rates (for grid-corrected thresholded BP, BPLOT and SCOMP)  between Fig.\ref{fig6}(b) and \ref{fig5-2}(b) indicates that increasing the spatial diversity and the number
of data can compensate the deficiency in frequency diversity, up to a point. 

For the purpose of  comparison, Fig. \ref{fig7-2} shows the results of PDMF SAR with 
(a) $m_1=28, m_2=7$ and (b) $  m_1=7, m_2=28$, and other parameters the same as in Fig. \ref{fig5-2}. The number of  degrees of diversity (=196)
is the same  for both Fig. \ref{fig5-2} and \ref{fig7-2}.  Clearly, the performances of grid-corrected thresholded BP, BPLOT and SCOMP improve (slightly) in Fig. \ref{fig7-2}(a) but
degrade  in Fig. \ref{fig7-2}(b)
 relative to Fig. \ref{fig5-2}(b). This means that,  for a fixed bandwidth and number of degrees of diversity, there is an optimal distribution between the frequency diversity and
the angular diversity. For example, for a smaller bandwidth, the frequency diversity should be decreased (and the angular diversity be increased)
accordingly.

\commentout{
On the other hand, Fig. \ref{fig7-2}(a) is another instance where the effect of
 target configuration becomes important, especially when the success rate
 ($\leq 20\%$) is low. The gridding error and the recovery error are
 both smaller in Fig. \ref{fig7-2}(a) than Fig. \ref{fig5-2}(a). 
}

In the case of extreme deficiency in frequency diversity $(\nu_*-\nu_0)/\nu_0\ll 1$, the gridding error dominates the data and our methods eventually break down.  
  
    \begin{figure}[t]\centering
 \subfigure[BP]{ \includegraphics[width=8cm]{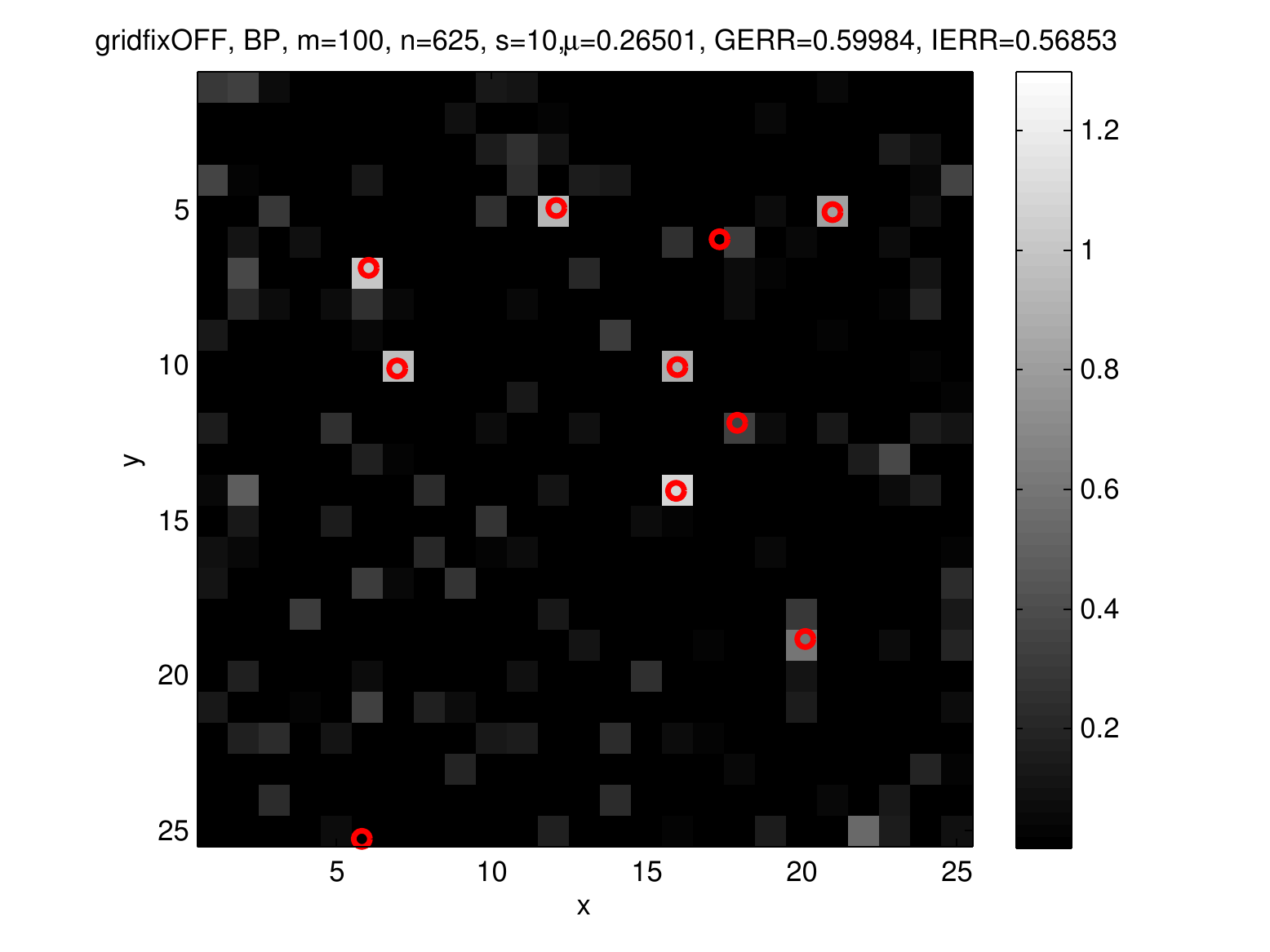}}
 \subfigure[OMP]{\includegraphics[width=8cm]{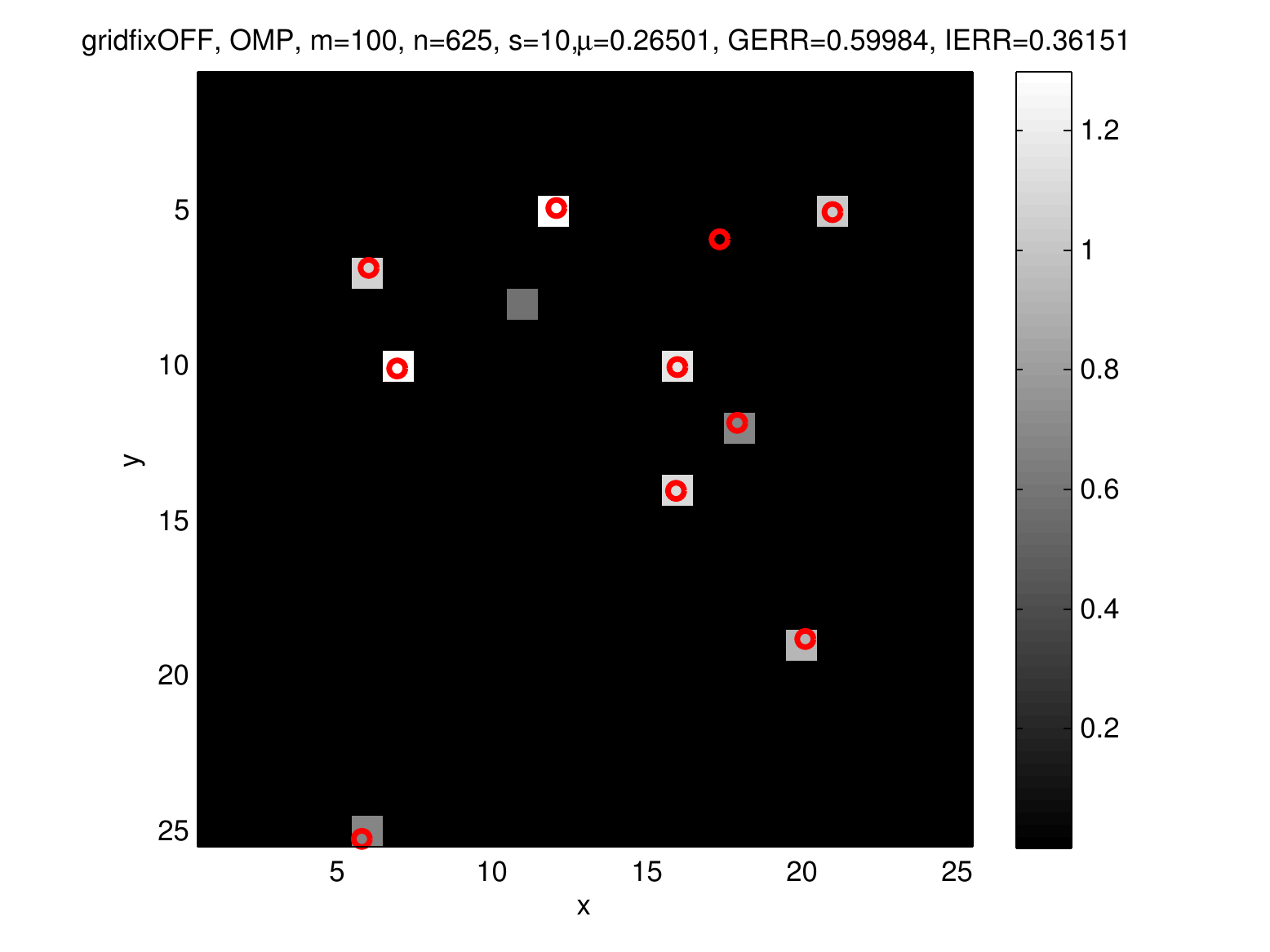}}
 \subfigure[BPLOT]{\includegraphics[width=8cm]{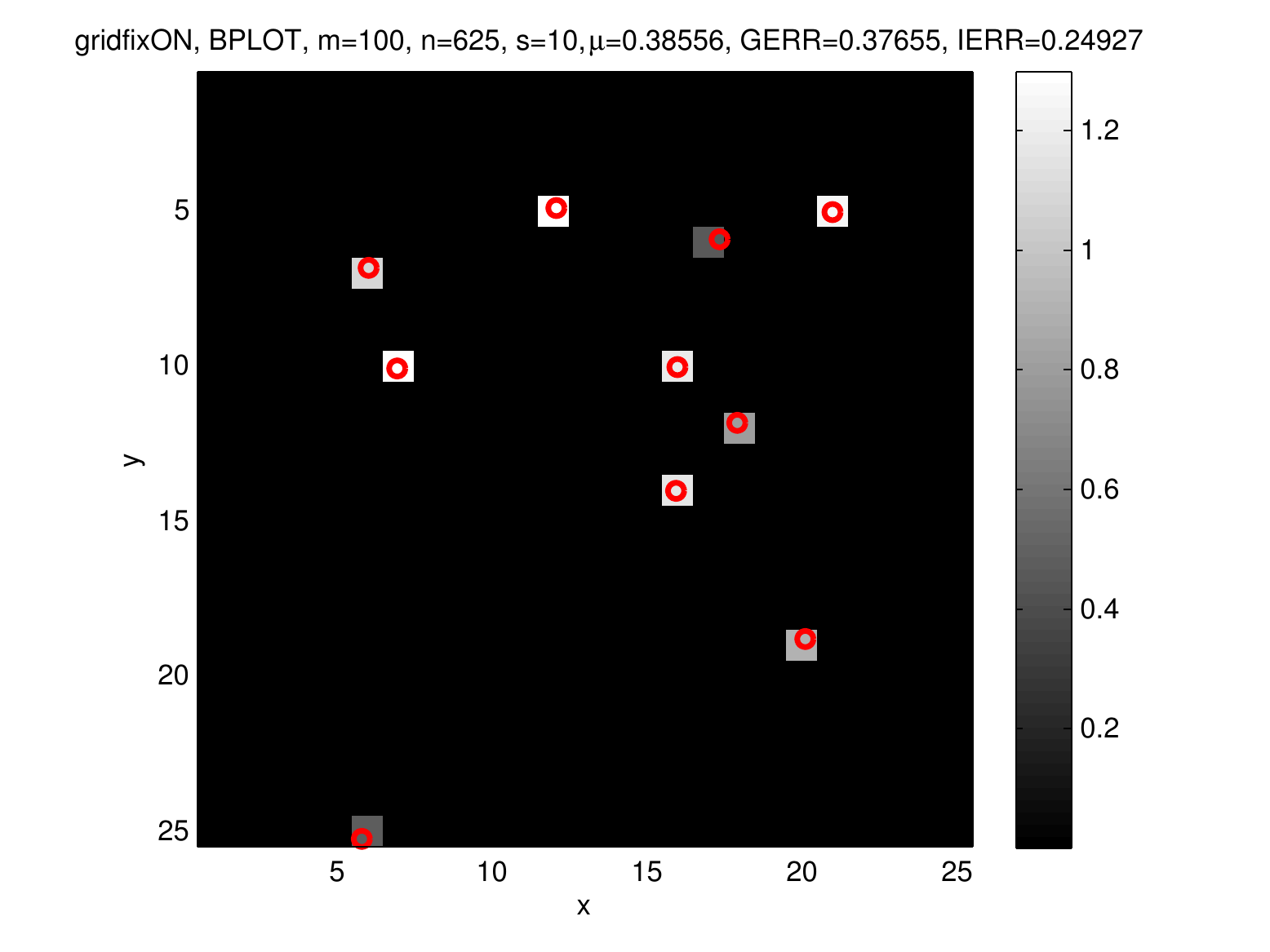}}
  \subfigure[SCOMP]{\includegraphics[width=8cm]{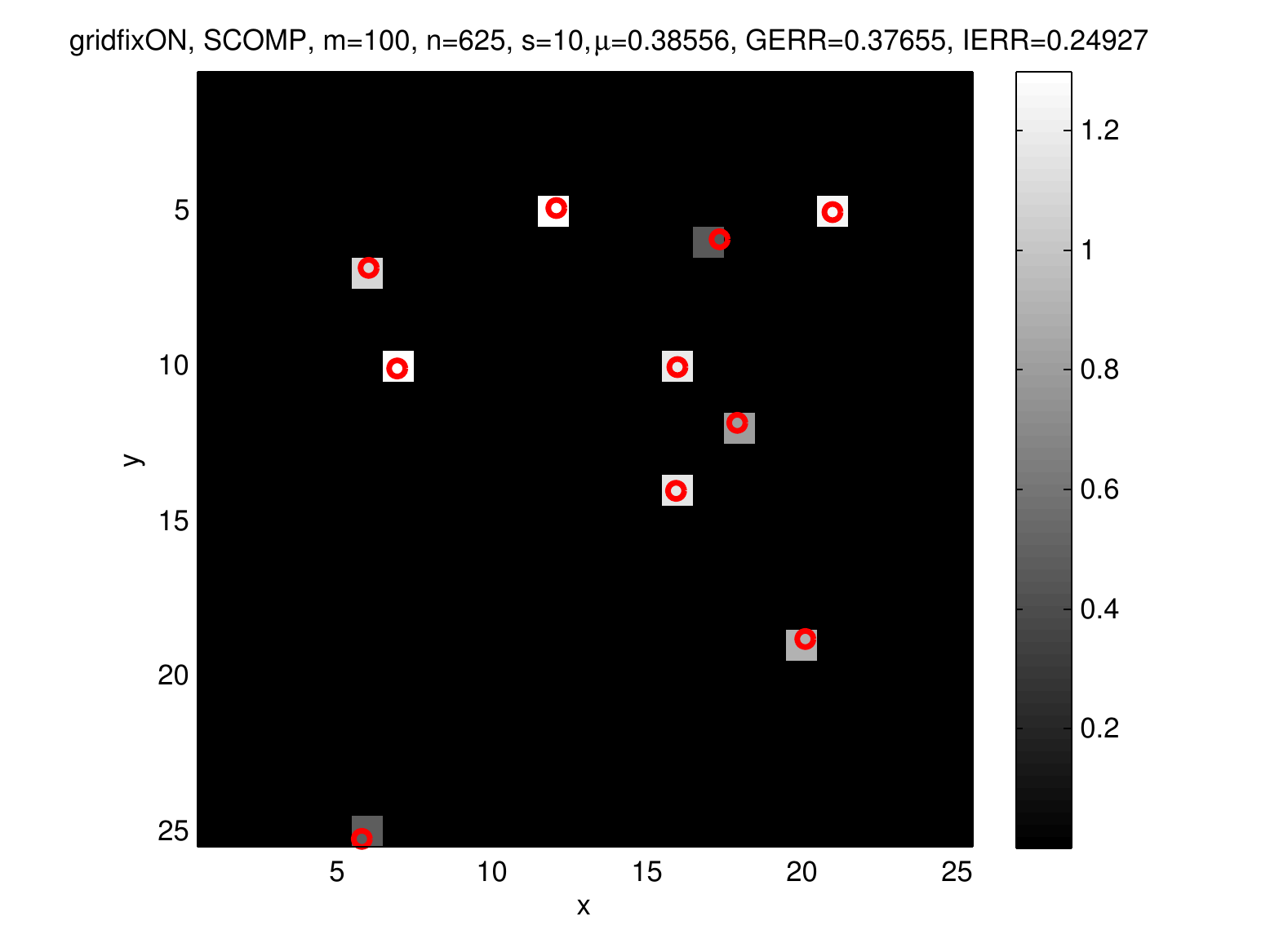}}
  \caption{FDMF  SAR imaging with SAR scheme B and $Q=2$ by (a) BP, (b) OMP, without
  grid correction, 
   and (c) BPLOT, (d) SCOMP, with grid correction. }
   \label{fig3-2}
  \end{figure}

    \begin{figure}[t]\centering
          \subfigure[SCOMP-NLS with $Q=1$]{\includegraphics[width=8cm]{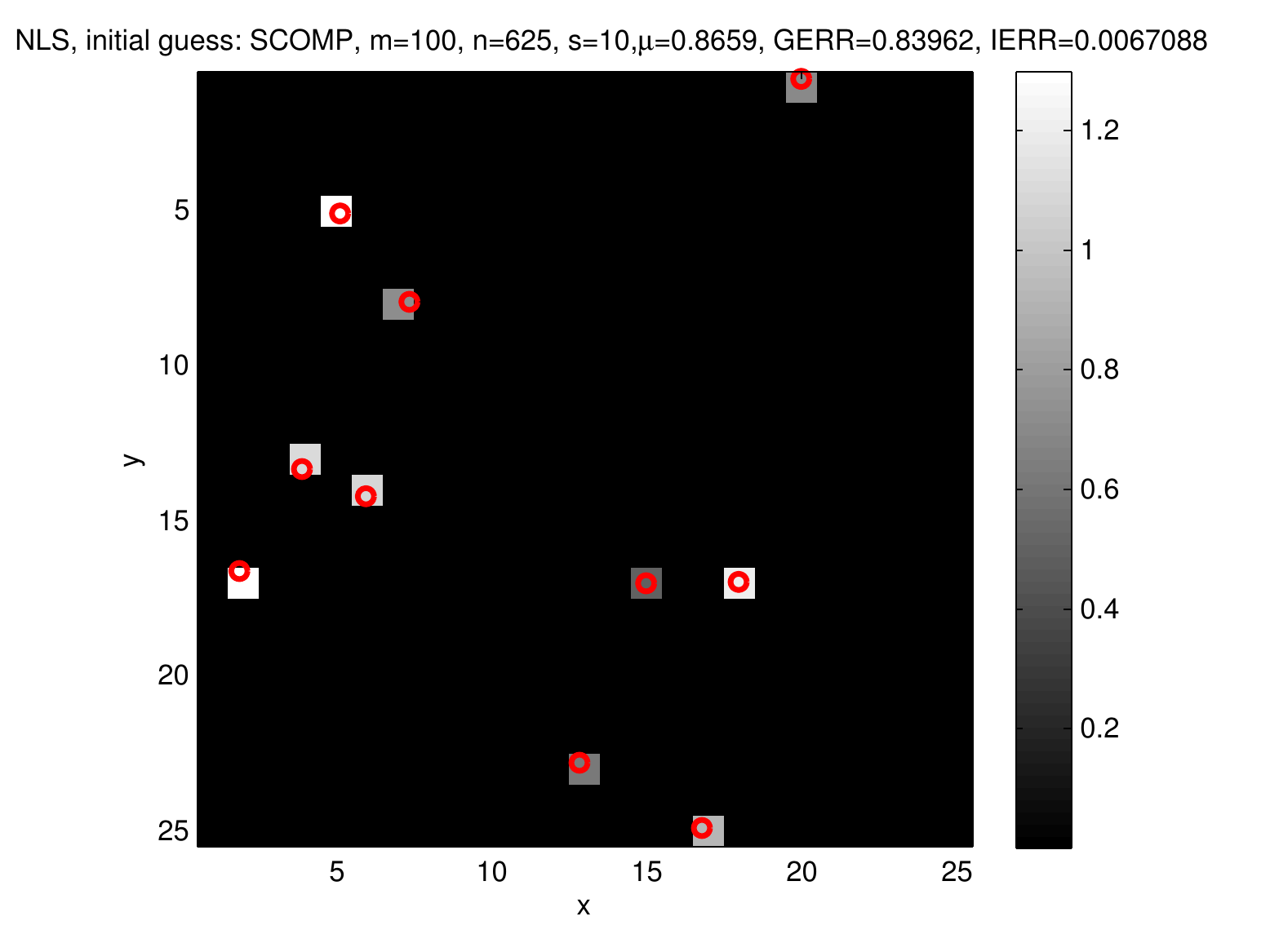}}
  \subfigure[SCOMP-NLS with $Q=2$]{\includegraphics[width=8cm]{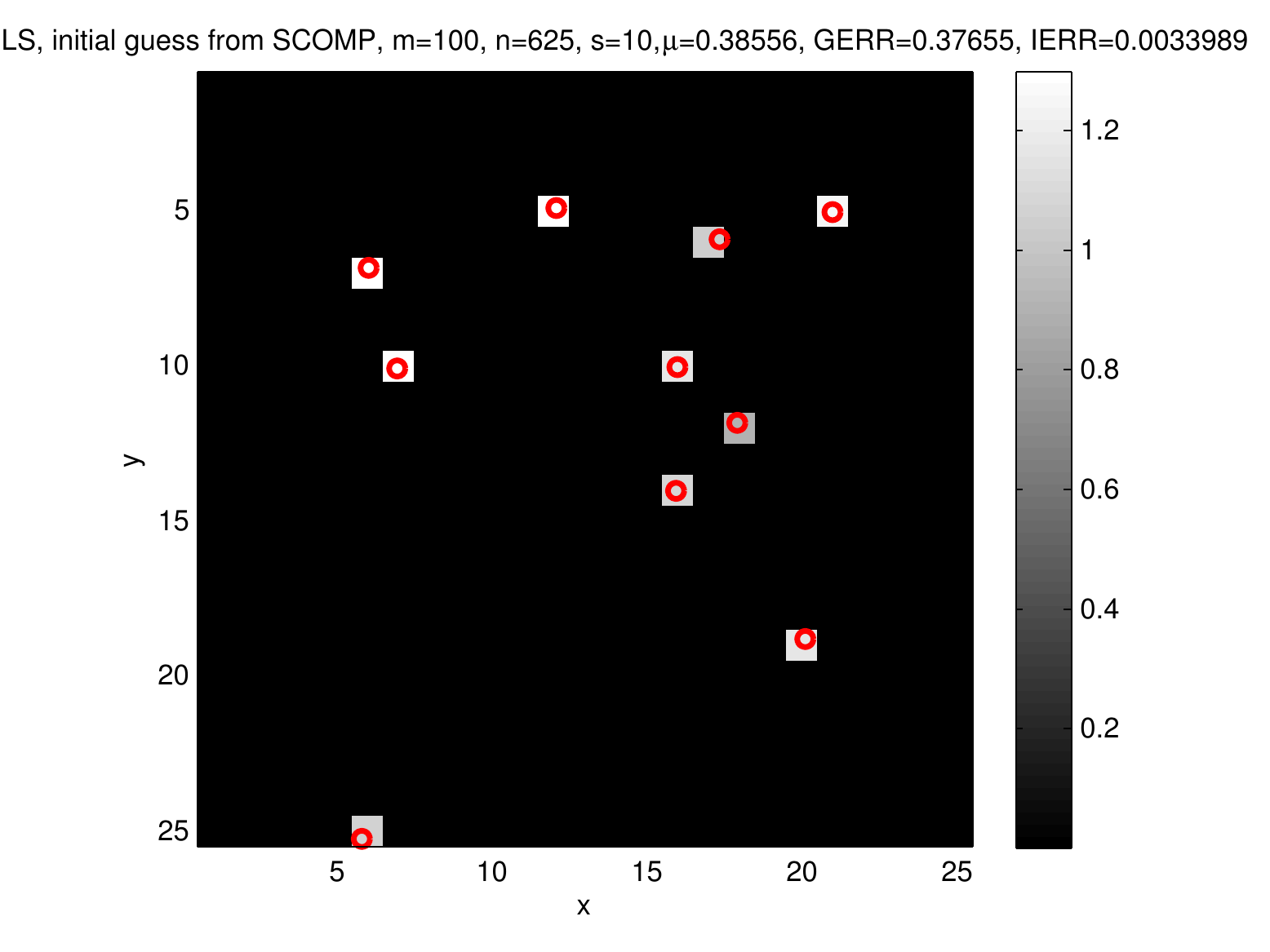}}
  \caption{  SCOMP-NLS for FDMF SAR scheme B produces error of (a) $0.7\%$ with $Q=1$ and (b) $0.3\%$ with $Q=2$. }
  \label{fig4-1}
  \end{figure}

  \begin{figure}[t]\centering
    \subfigure[$Q=1$]{\includegraphics[width=8cm]{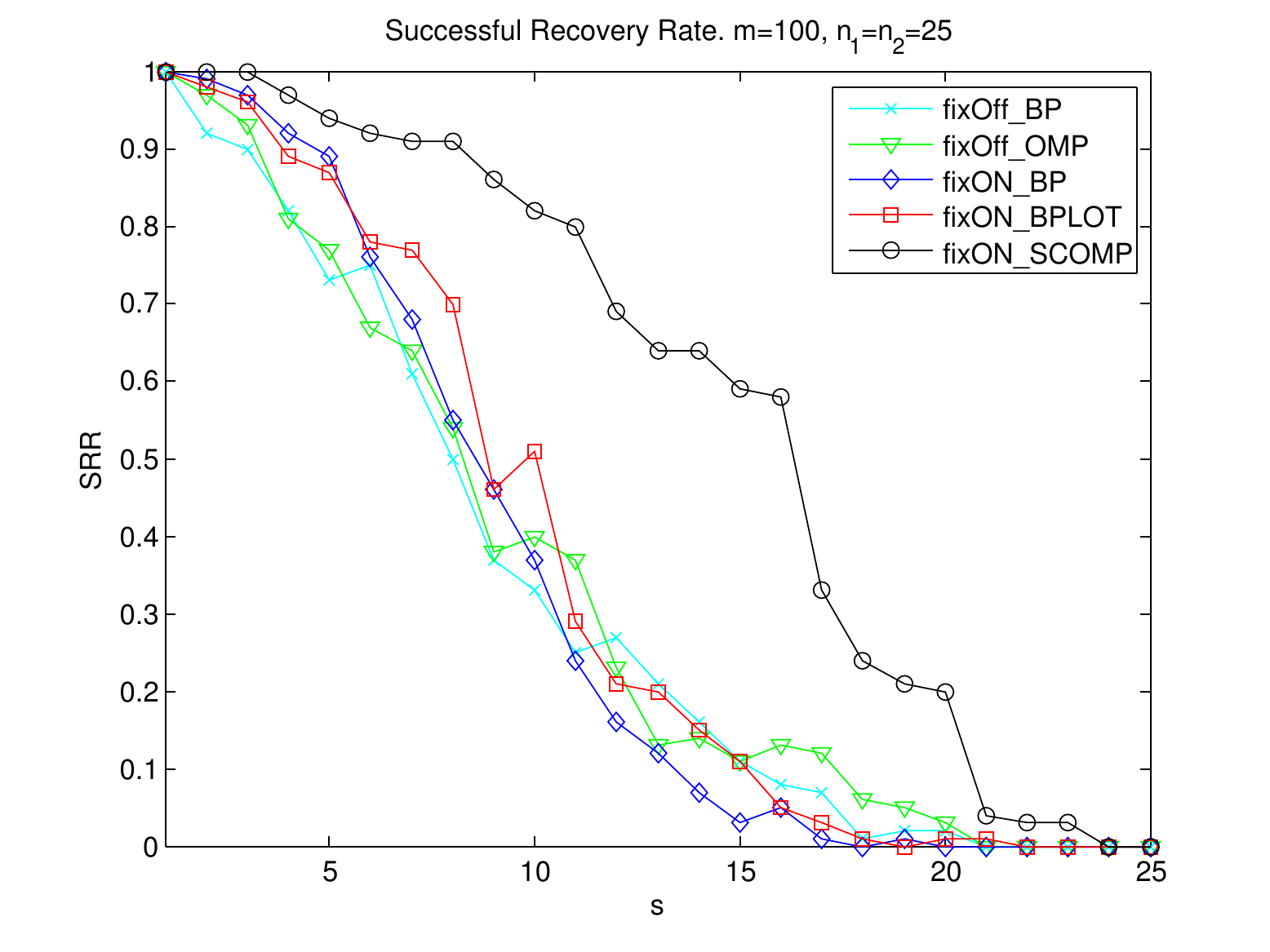}}
        \subfigure[ $Q=2$]{\includegraphics[width=8cm]{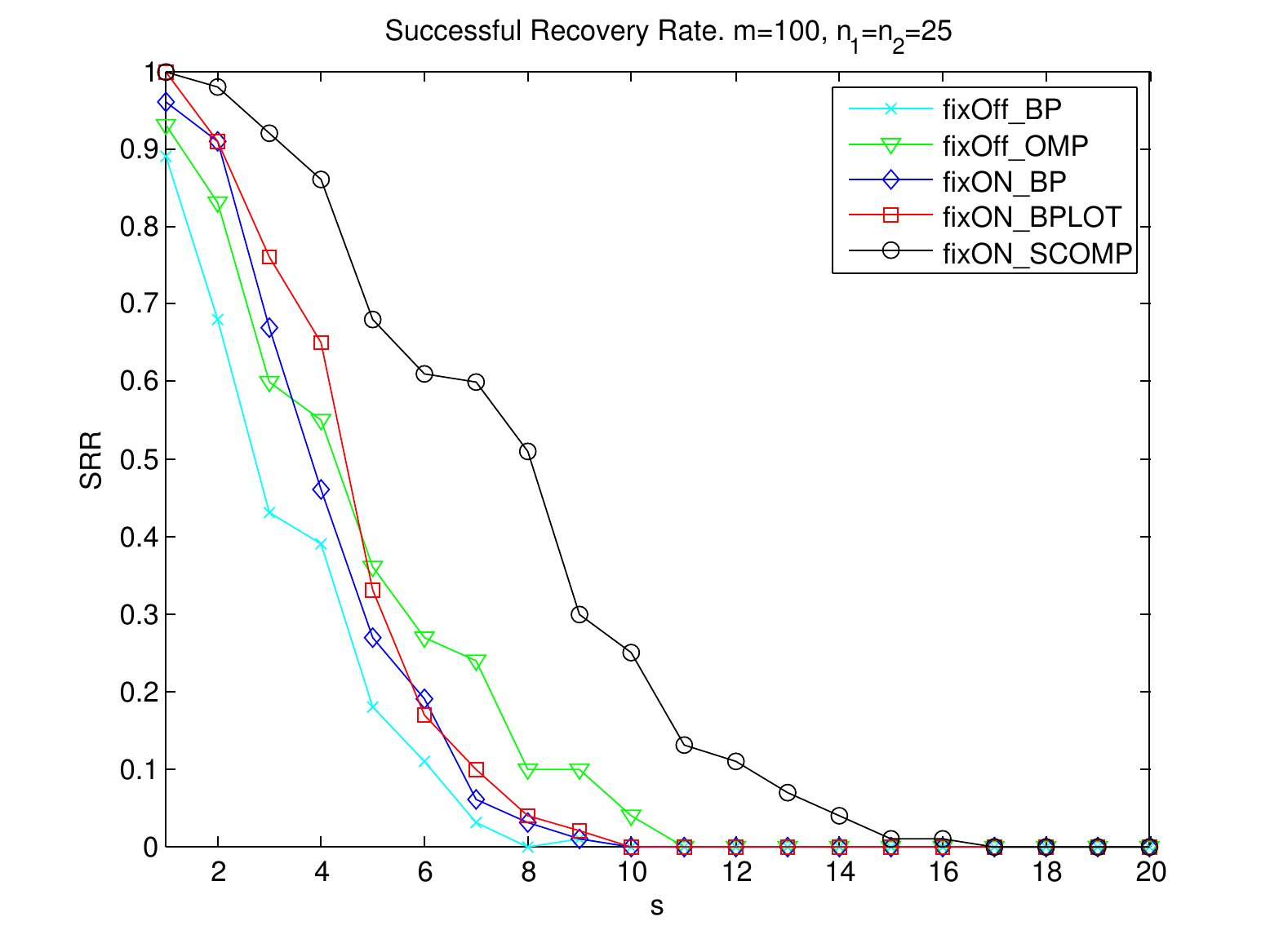}}
  \caption{ Success rate versus sparsity for FDMF SAR scheme B  with (a) $Q=1$ and (b) $Q=2$.  The legend is same as in Fig. \ref{fig-new12-2}.  }
  \label{fig4-2}
  \end{figure}
  
    \begin{figure}[t]\centering
    \subfigure[SCOMP-NLS]{\includegraphics[width=8cm]{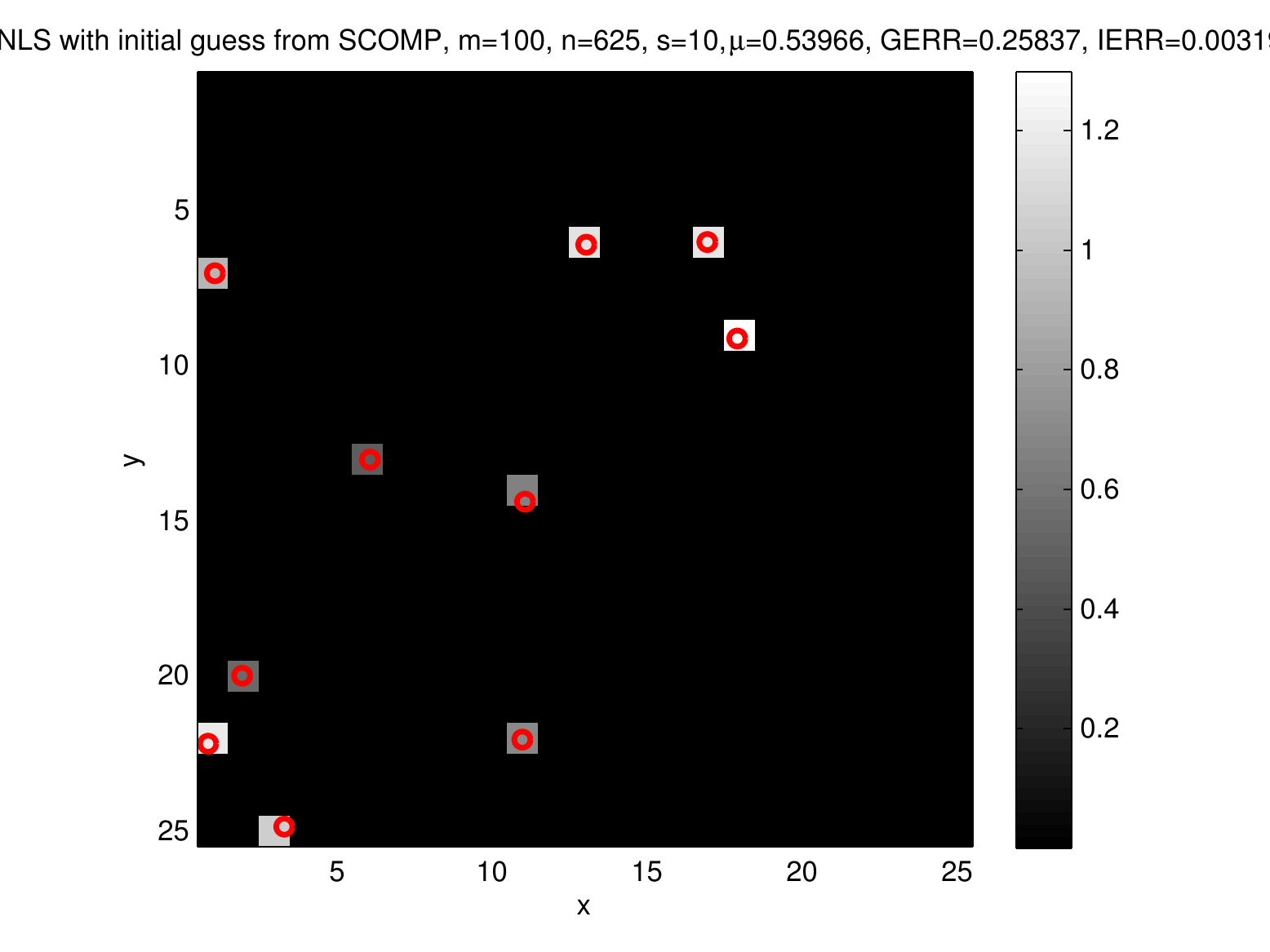}}
        \subfigure[Success rate vs. sparsity]{\includegraphics[width=8cm]{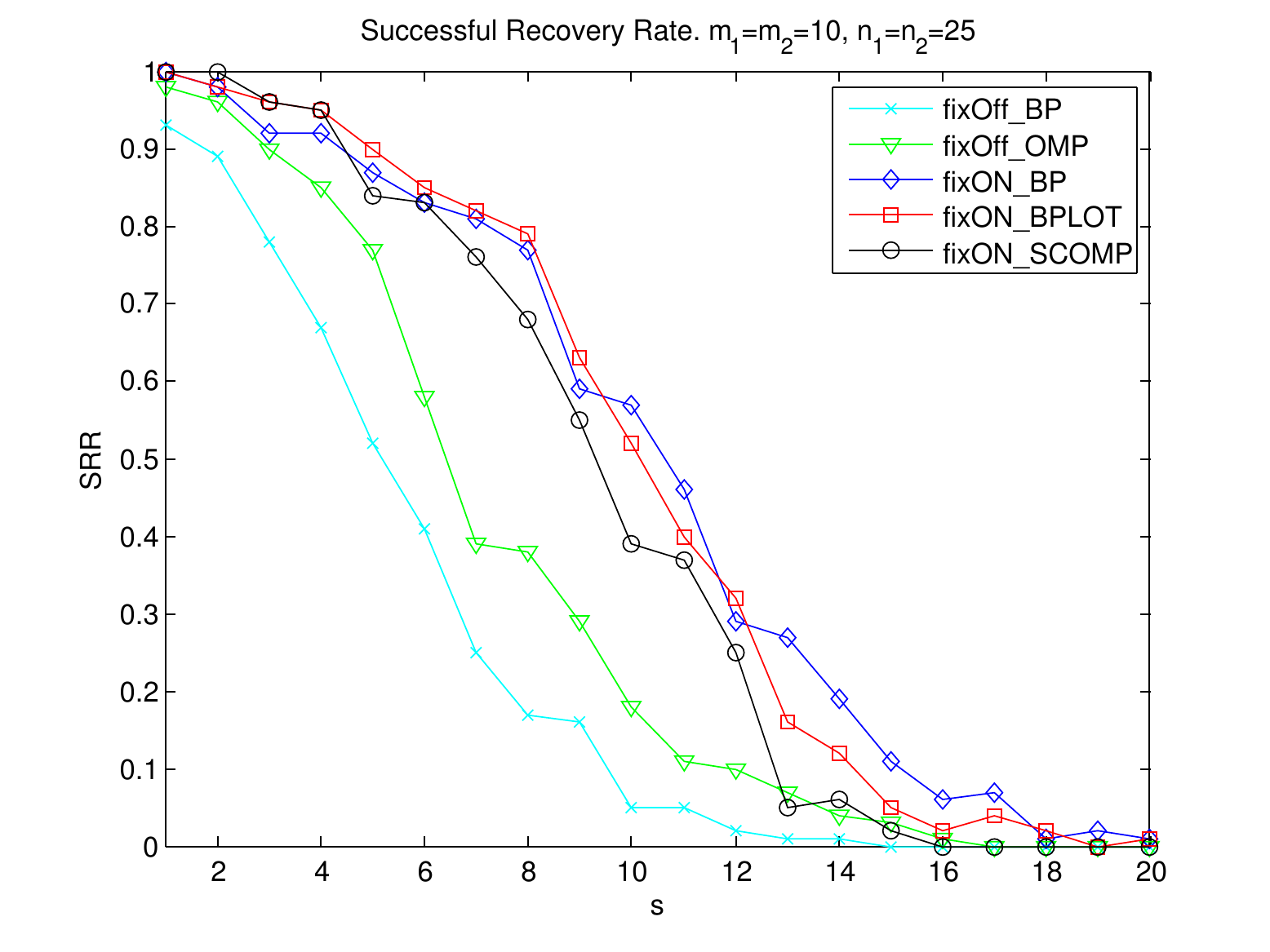}}
  \caption{ FDMF  SAR imaging with SAR scheme A ($m_1=m_2=10, \nu_0=0, \nu_*=1$).  (a) SCOMP-NLS produces error of $0.3\%$.  (b) Success rate versus sparsity. The legend is same as in Fig. \ref{fig-new12-2}.  }
  \label{fig6}
  \end{figure}

    \begin{figure}[t]\centering
    \subfigure[BP]{\includegraphics[width=8cm]{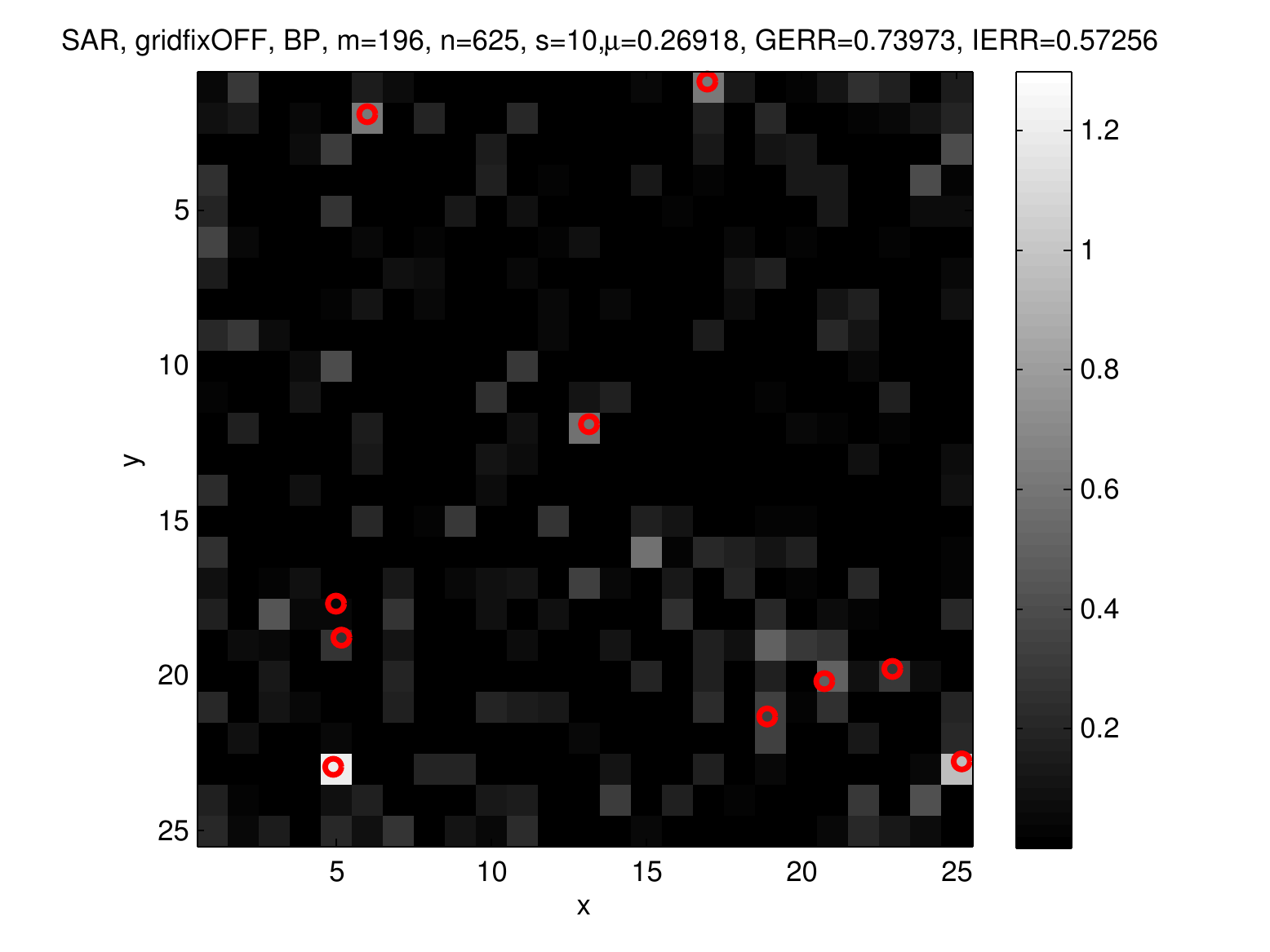}}
        \subfigure[OMP]{\includegraphics[width=8cm]{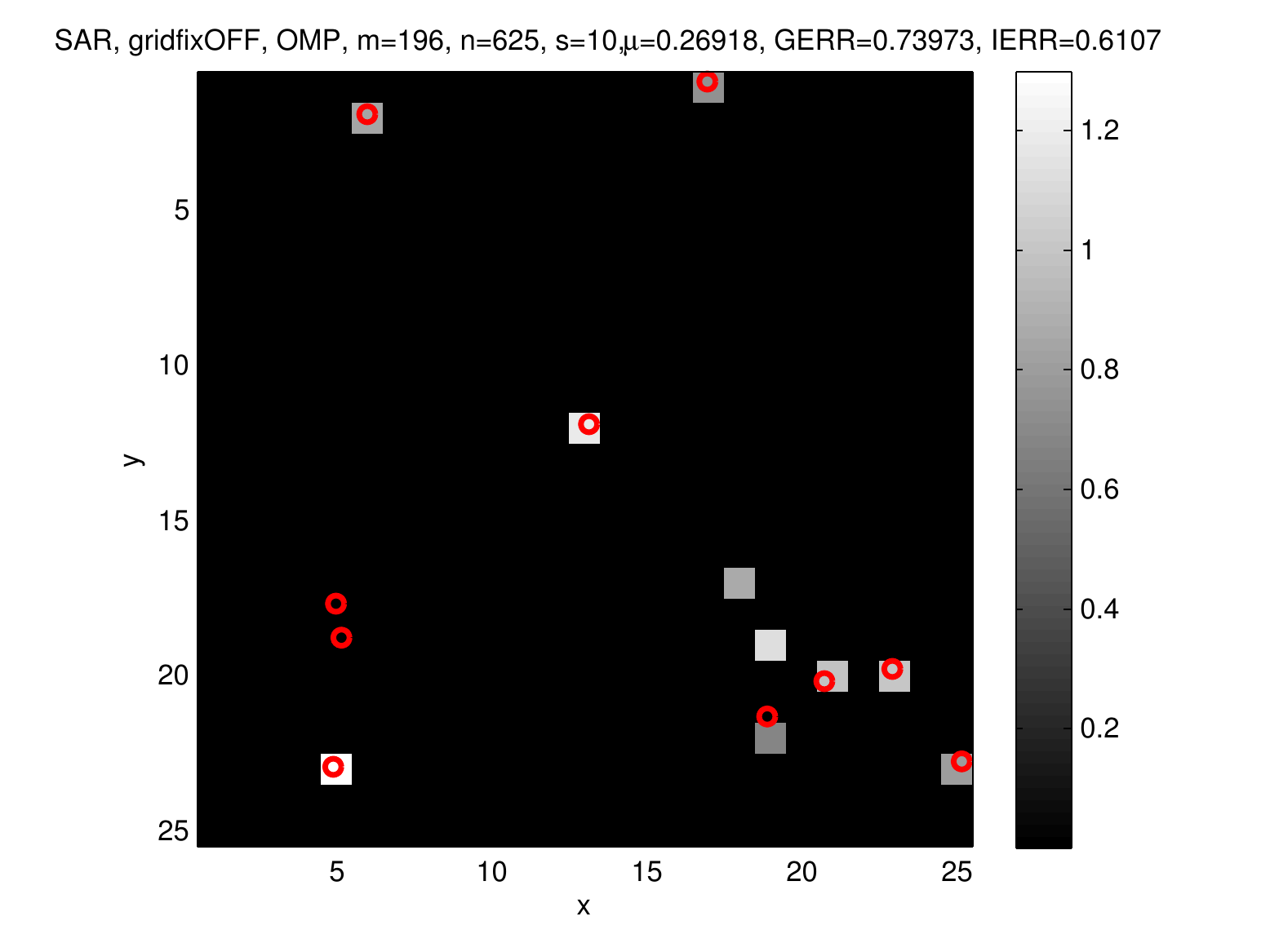}}
            \subfigure[BPLOT]{\includegraphics[width=8cm]{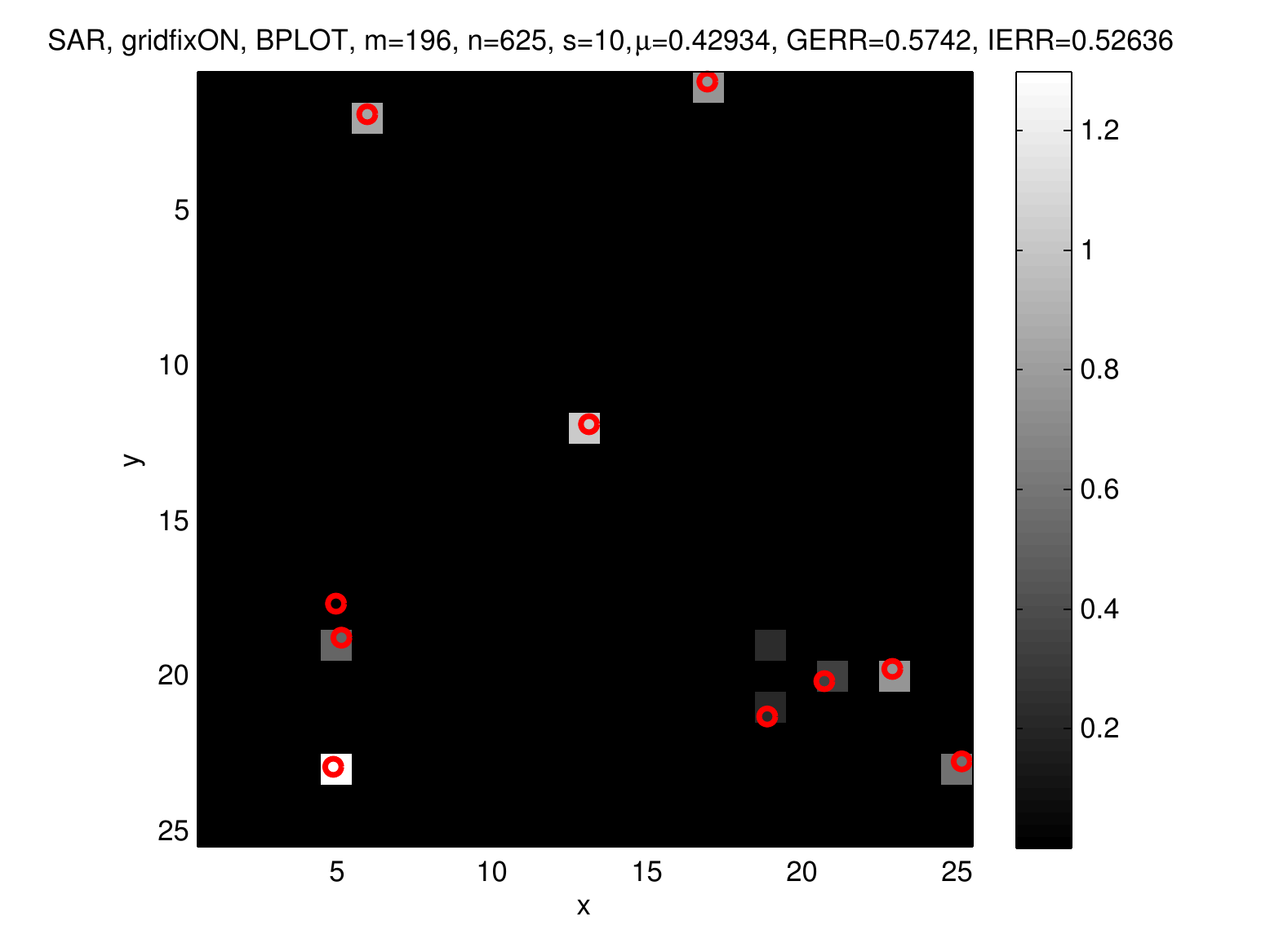}}
                \subfigure[SCOMP]{\includegraphics[width=8cm]{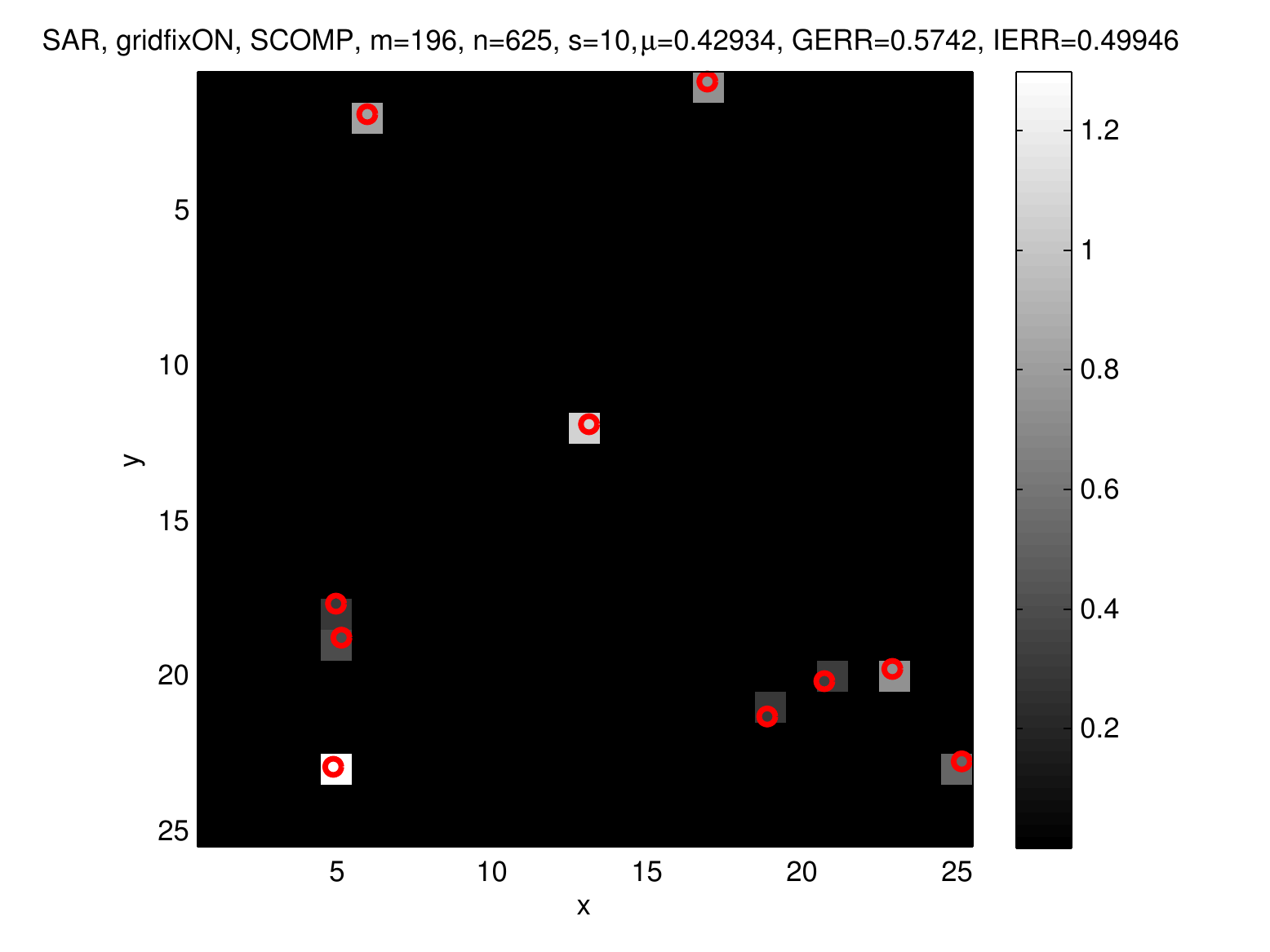}}
  \caption{ PDMF  SAR imaging  with SAR scheme A  by (a) BP, (b) OMP, without
  grid correction, 
   and (c) BPLOT, (d) SCOMP, with grid correction. }
  \label{fig5-1}
  \end{figure}
  
      \begin{figure}[t]\centering
                        \subfigure[SCOMP-NLS]{\includegraphics[width=8cm]{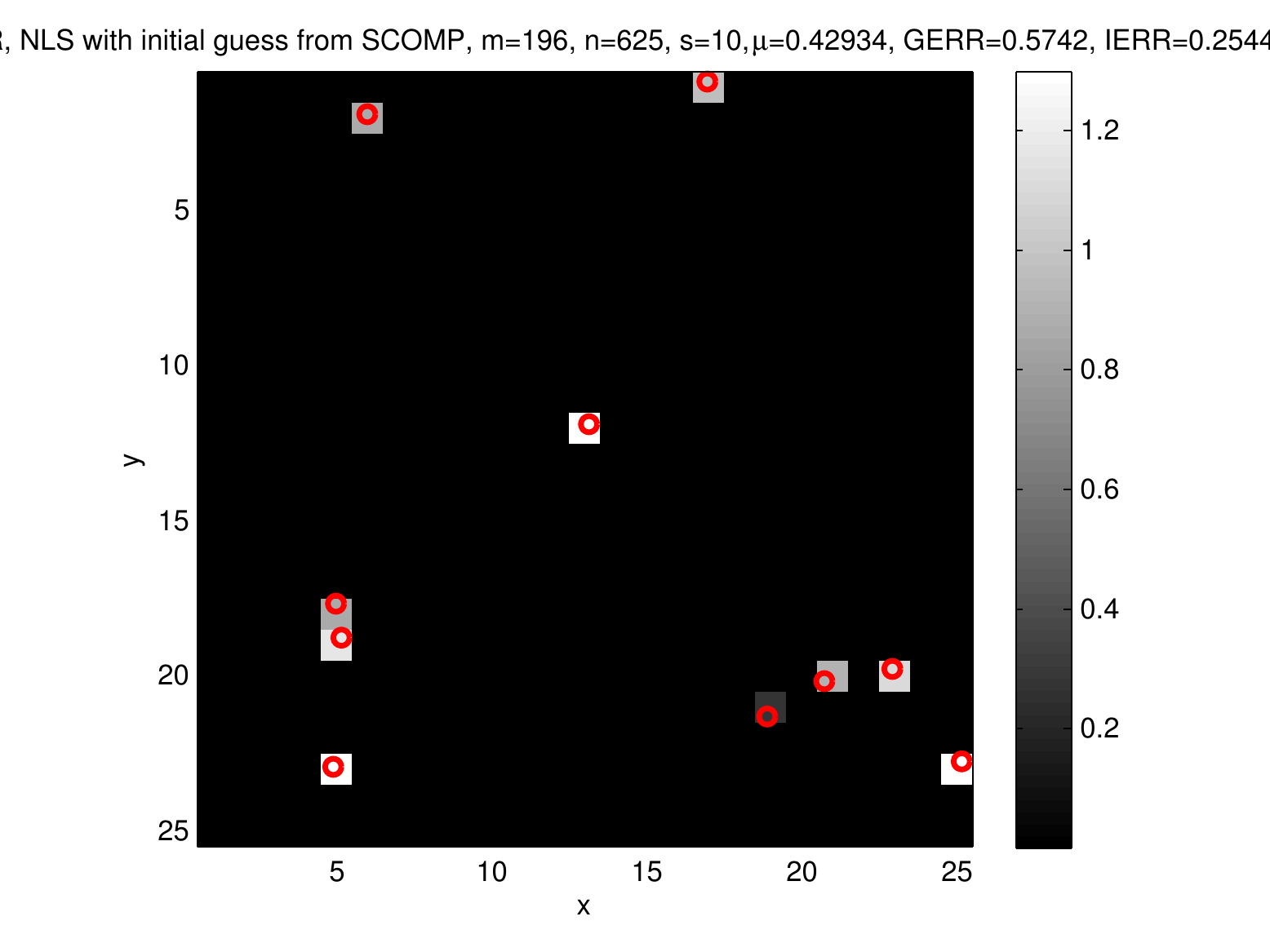}}
        \subfigure[Success rate versus sparsity]{\includegraphics[width=8cm]{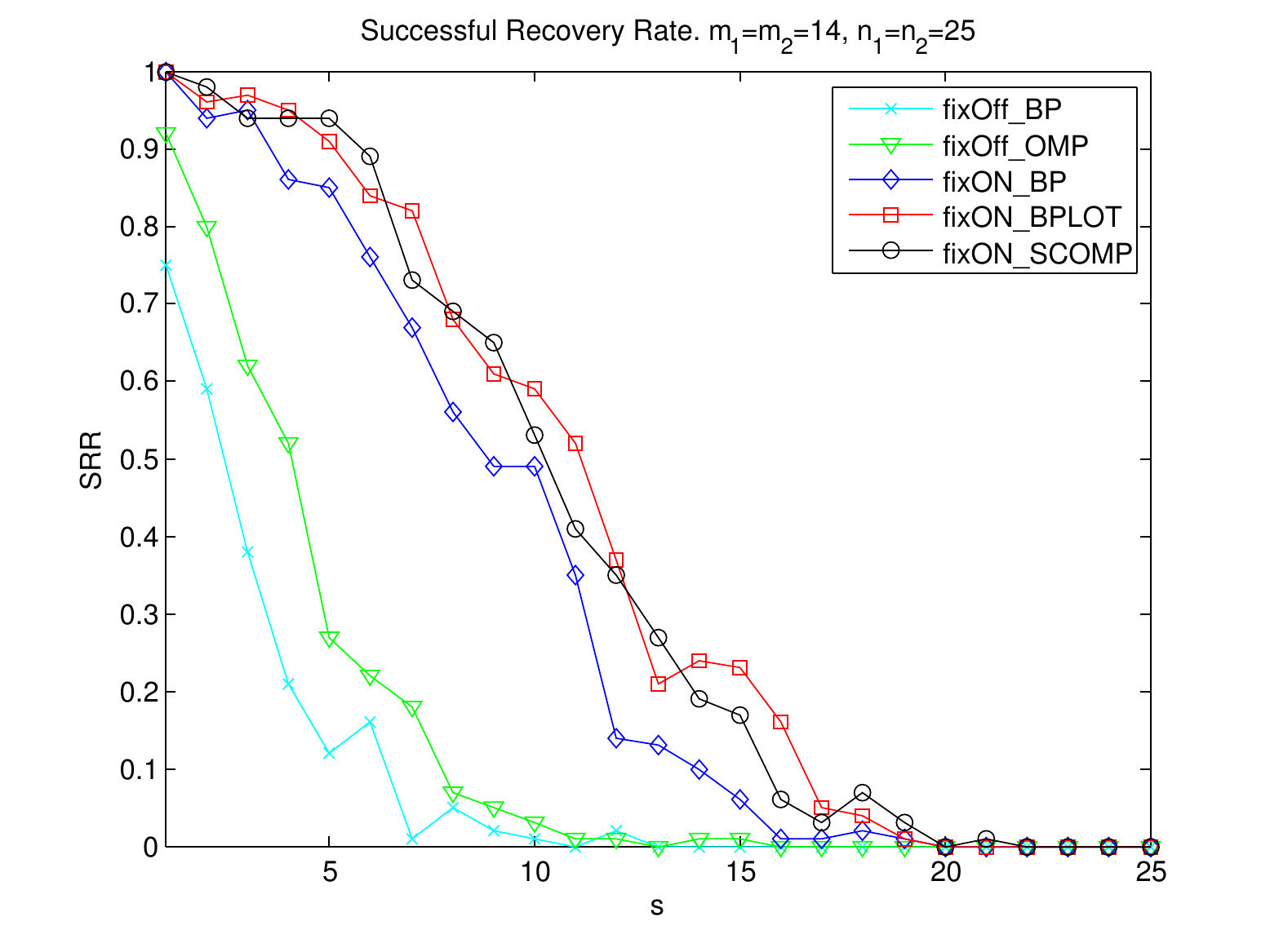}}
  \caption{PDMF  SAR imaging with scheme A ($m_1=m_2=14,\nu_0=1/2,\nu_*=1$).  (a) SCOMP-NLS produces error of $25.4\%$. (b) Success rate versus sparsity. The legend is same as in Fig. \ref{fig-new12-2}. }
  \label{fig5-2}
  \end{figure}

\commentout{
    \begin{figure}[t]\centering
    \subfigure[BP]{\includegraphics[width=8cm]{fig_SAR_m1_7_m2_28_v_05_1/fig_SAR_gridfixOFF_BP-eps-converted-to.pdf}}
        \subfigure[OMP]{\includegraphics[width=8cm]{fig_SAR_m1_7_m2_28_v_05_1/fig_SAR_gridfixOFF_OMP-eps-converted-to.pdf}}
            \subfigure[BPLOT]{\includegraphics[width=8cm]{fig_SAR_m1_7_m2_28_v_05_1/fig_SAR_gridfixON_BPLOT-eps-converted-to.pdf}}
                \subfigure[SCOMP]{\includegraphics[width=8cm]{fig_SAR_m1_7_m2_28_v_05_1/fig_SAR_gridfixON_SCOMP-eps-converted-to.pdf}}
  \caption{ PDMF  SAR  scheme A  by (a) BP, (b) OMP, without
  grid correction, 
   and (c) BPLOT, (d) SCOMP, with grid correction. }
  \label{fig7-1}
  \end{figure}
  }
      \begin{figure}[t]\centering
                        \subfigure[$m_1=28, m_2=7$ ]{\includegraphics[width=8cm]{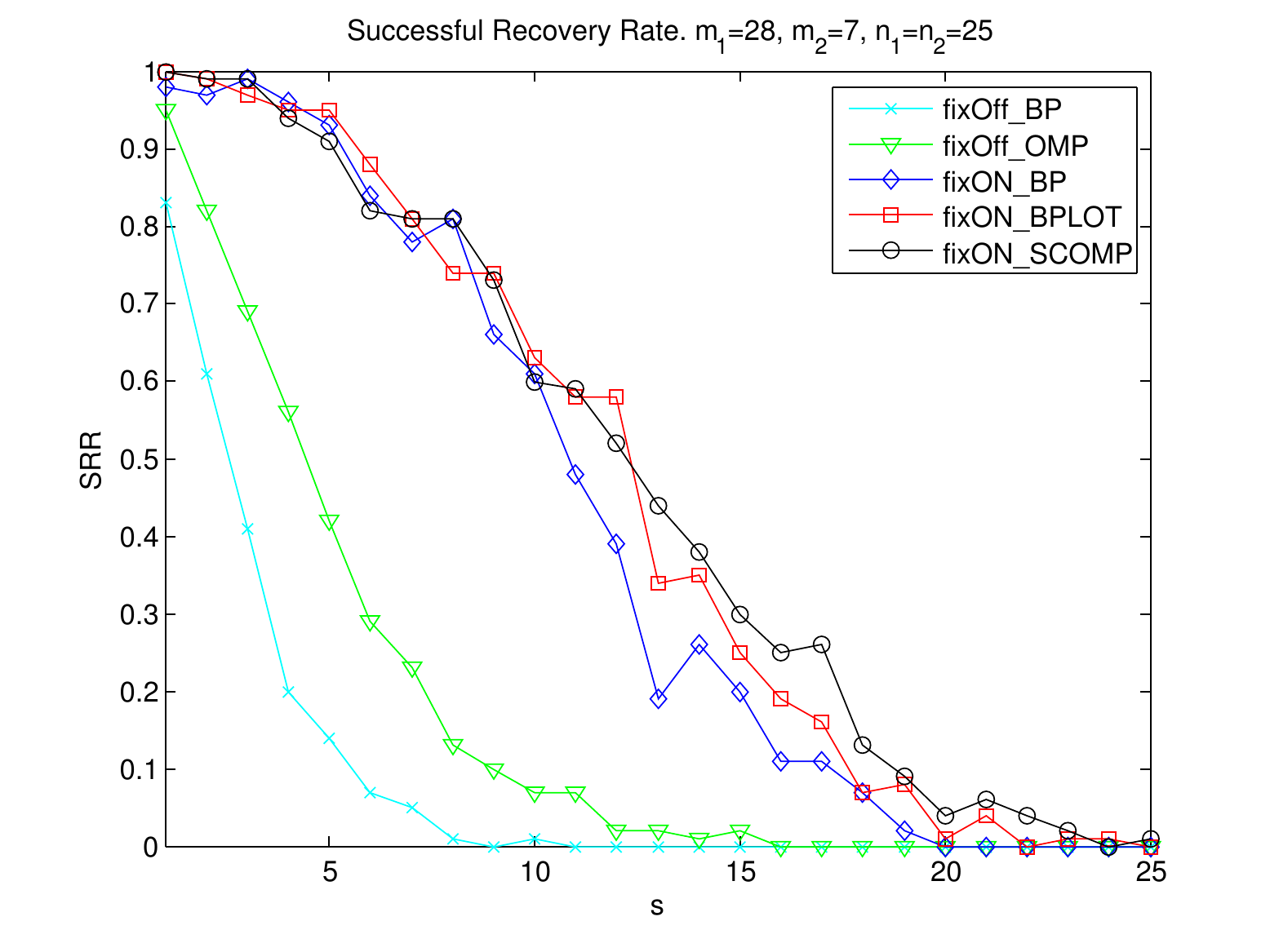}}
        \subfigure[$m_1=7, m_2=28$]{\includegraphics[width=8cm]{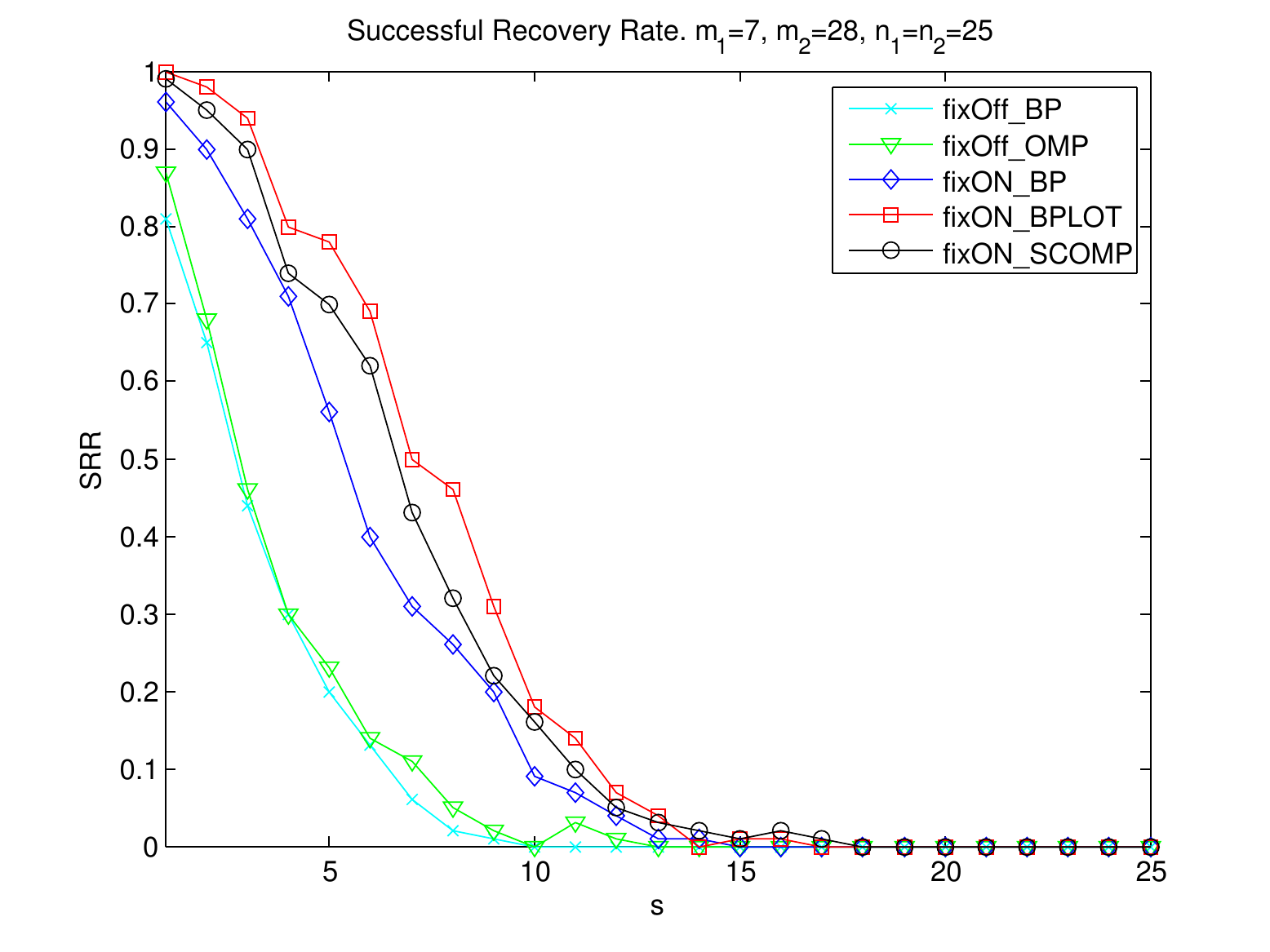}}
  \caption{Success rate of PDMF SAR scheme A with  (a) $m_1=28, m_2=7$  (b) $m_1=7, m_2=28$. The legend is same as in Fig. \ref{fig-new12-2}.  }
  \label{fig7-2}
  \end{figure}

  \section{Conclusion}\label{sec7}
We explored  compressed sensing approach to
monostatic radar with chirped
signals or  multi-frequency UNB waveforms. Particular attention is on the off-grid targets and the resulting  intrinsically nonlinear
gridding error. 

 We used the 
Taylor expansion of phase factor to approximate the signals
from the off-grid targets and reduce  the gridding error. 
We proposed a new algorithm, SCOMP, to solve the 
resulting grid-corrected system and gave
a performance guarantee (Theorem \ref{thm1}). 
Our theory, however, does  not fully account for  the numerical  performance of 
the proposed schemes, especially in the regime of low $Q$ which remains to be
further analyzed  (Remark \ref{rmk3}).  

In addition, we proposed technique (LOT) to enhance BP for the off-grid setting. The resulting method BPLOT can sometimes outperform SCOMP (Fig.\ref{fig-new12-2}(a) and \ref{fig6}(b)).  We extended SCOMP and the performance guarantee (Theorem \ref{thm22})   to 
Spotlight SAR and proposed the UNB multi-frequency version of implementation. Our numerical experiments show significant improvement over
the standard CS methods,
especially in locating sparse targets to the grid accuracy. The recovery of target amplitudes 
can be further improved by applying the nonlinear least squares 
with the SCOMP/BPLOT estimates as the initial guess. 

Our numerical study indicates that in both  radar ranging and SAR imaging,
our methods perform best  with $Q=1$. The latter corresponds to the setting where the grid spacing
is around the resolution threshold of the probe, no more no less. Excessive  bandwidth
for  the same grid spacing hinders  the radar  performance due to overall  enhanced  level of gridding error. 

When full frequency diversity is not available,  a good performance can be maintained  up to about 2/3 fractional bandwidth.
Further reduction in the probe bandwidth significantly degrades   performance. 
Therefore the signals much be of ultra-wideband (UWB), defined as at least
$1/4$ fractional bandwidth \cite{UWB},  if Spotlight SAR (\ref{45}) is to be implemented with
chirped signals and sparse measurements. 

Implementing  the proposed CS Spotlight SAR 
with multi-frequency UNB waveforms, instead of UWB pulses,  has the added  benefits  of simpler transmitters,  increased signal-to-noise ratio due to less unwanted thermal noise and increased signal-to-interference ratio due to avoiding the electromagnetic spectrum occupied by other civilian and military applications. The last of these benefits is a natural fit for the
CS paradigm which opens the door for fully diversified,  but sparse measurements
in the frequency domains. 

In the case of extreme deficiency in frequency diversity $(\nu_*-\nu_0)/\nu_0\ll 1$, the gridding error dominates the data and our methods eventually break down.  
 In this case  
 SAR  imaging of off-grid targets with {sparse measurement} requires
a different approach than the proposed methods.

We plan to extend our methodology to the case of range-Doppler radar
and SAR imaging of moving targets in the future.

\begin{appendix}
  \section{Proof of Lemma \ref{lem1}}\label{app1}
\begin{proof}  
  We prove the coherence bound for the matrix $\bA=[\bF\ \bG]$. 
  
  The $k$-th column vector $A_k$ 
  of $\bA$ is given by
  $$A_{jk}= \begin{cases}
                 \e^{-2\pi\i Q k \bar{t}_j},  &k\leq n
		\\
		( \bar{t}_j -1/2)\e^{-2\pi\i Q (k-n) \bar{t}_j}, & k>n.                \end{cases}  $$
           Note that     $ \norm{A_k}^2 = m , k\leq n$ and
           $ \Exp[\norm{A_k}^2] = {m}/{12}, k>n.$
  Consequently, the scalar product of two distinct columns of $\bA$
  has three possible forms:
  $$
b_{kk'}= \sum_{j=1}^m a(\bar{t}_j) \exp\bracket{2\pi\i Q(k-k')\bar{t}_j} 
  \ , \quad a(t_j) = 1,\ (\bar{t}_j-1/2),\text{ or } \ (\bar{t}_j-1/2)^2\ , $$
  for $k,k'=1, \ldots, n$.  
  When $a(t_j) = 1$ or $(\bar{t}_j-1/2)^2$  both columns are drawn from $\bF$ or $\bG$  and
  thus $k\neq k'$. When $a(t_j) =\bar{t}_j-1/2$, one column is drawn
  from $\bF$ and the other from $\bG$. In the last case, $k$ and $k'$
  are arbitrary.

  Let $S_m=\sum_{j=1}^m U_j$, $T_m=\sum_{j=1}^m V_j$
where
  $$
   U_j = a(\bar{t}_j)\cos\bracket{2\pi Q(k-k')\bar{t}_j}\ , \quad
   V_j = a(\bar{t}_j)\sin\bracket{2\pi  Q(k-k')\bar{t}_j}    
  $$ are independent  (for different $j$) random variables in $[-1,1]$. We have
  \begin{align*}
  \abs{b_{kk'}} 
  &\leq \abs{b_{kk'} - \Exp(b_{kk'})}
  + \abs{\Exp(b_{kk'})} \notag \\
  &=  \abs{ S_m+\i T_m - \Exp S_m-\i \Exp T_m }
  + \abs{\Exp(S_m+\i T_m)} .
  \end{align*}
Recall  the Hoeffding inequality.
  \begin{proposition}
    Let $U_1,\ldots,U_m$ be independent random variables, and $S_m=\sum_{j=1}^m U_j$.
    Assume that
    $U_j\in[u,v]$, $j=1,2,\ldots,m$ almost surely, then we have
    \begin{equation}\Pr(\abs{S_m-\Exp S_m}\geq mt)\leq 
    2\exp\bracket{-{2m^2 t^2\over\sum_j (v-u)^2}}
    \end{equation}
    for all positive $t$.
  \end{proposition}
  Choosing $t=K/\sqrt{m}$ for some constant $K$, we have
  $$\Pr(\abs{S_m-\Exp S_m}\geq \sqrt{m}K)\leq 2\exp\bracket{-K^2/2}.$$
  
  Note that the quantities $S_{m}$ depend on $
k-k'$ but 
there
are at most $n-1$ different values. 
  The union bound  yields
  $$\Pr(\max_{k\neq k'}\abs{S_m-\Exp S_m}\geq \sqrt{m}K)\leq 2(n-1)\exp\bracket{-K^2/2},$$
  and similarly
  $$\Pr(\max_{k\neq k'}\abs{T_m-\Exp T_m}\geq \sqrt{m}K)\leq 2(n-1)\exp\bracket{-K^2/2}.$$
 We have
  \begin{align*}
  & \Pr(\max_{k\neq k'}\abs{ b_{kk'} - \Exp b_{kk'} }
   < \sqrt{2m}K)
   \notag \\ = \quad &
  \Pr(\max_{k\neq k'}\abs{ S_m+\i T_m - \Exp S_m -\i\Exp T_m } < \sqrt{2m}K) 
  \notag \\  > \quad & \Big(1 - 2(n-1)\exp\bracket{-K^2/2} \Big)^2 \ > \ (1-\delta)^2
  \end{align*}
  if  $\delta>2n\exp\bracket{-K^2/2}$. 

  Now let us estimate the mean $\Exp(b_{kk'})$
  or $\Exp(S_m+\i T_m)$ for $k\neq k'$. Note that $\bar{t}_j$, $j=1,\ldots,m$, are 
  independently and uniformly distributed in
  $[0,1]$. We have three different cases:
  \begin{enumerate}
   \item For $a(\bar{t}_j)=1$, 
     \[
      \Exp(b_{kk'})= m\int_{0}^{1}\e^{2\pi\i Q(k-k')t}\d t= 0.
      \]
   \item For $a(\bar{t}_j)=\bar{t}_j-1/2$,
     \[
     \Exp(b_{kk'})= m\int_{0}^{1}(t-1/2) \e^{2\pi\i Q(k-k')t}\d  t
     = m \e^{\pi i Q(k-k')}\frac{(-1)^{(k-k')Q}}{2\pi\i (k-k') Q},\quad k\neq k',
     \] and thus 
     \[
     \abs{\Exp(b_{kk'})} \leq \frac{m}{2\pi Q},\quad k\neq k'.
     \]
     On the other hand, 
     \[
     \Exp(b_{kk}) = m\int_{0}^{1}(t-1/2)\cdot 1\d t=0.
     \]
   \item For $a(\bar{t}_j)=(\bar{t}_j-1/2)^2$, 
     \[
     \Exp(b_{kk'})= m\int_{0}^{1}(t-1/2)^2 \e^{2\pi\i Q(k-k') t}\d t
     = m \e^{\pi i Q(k-k')} \frac{(-1)^{(k-k')Q}}{(2\pi(k-k')Q)^2},\quad k\neq k'
     \]
and thus 
\[
     \abs{\Exp(b_{kk'})} \leq \frac{m}{(2\pi Q)^2}.
     \]
  \end{enumerate}
For $ k,k'\leq n$, since $ \norm{A_k}^2 = m$, 
  \beq
  \label{a1}
  b_{kk'}\leq \frac{C}{m}\sqrt{2m}K 
  = C\frac{\sqrt{2}K}{\sqrt{m}},\quad k,k'\leq n
  \eeq
  for some universal constant $C$, with probability greater than $(1-\delta)^2$.

On the other hand, for  $k> n$, 
\[
\|A_k\|^2_2=\sum_j (t_j-1/2)^2
\]
which is a sum of $m$ i.i.d. random variables  of mean $1/12$ on $[0,1/4]$. Applying Hoeffding inequality with $t={1/24}$, we have
\[
\IP\Big(|\|A_k\|^2_2-{m\over 12}|\geq {m\over 24}\Big)\leq
2e^{-{m/18}}
\]
and thus
\[
\IP\Big(\|A_k\|^2_2\leq {m\over 24}\Big)\leq 2e^{-{m/18}}.
\]
We conclude from these observations that
 \beq
 \label{a2}
 b_{kk'}&\leq& \frac{C}{m}\bracket{\sqrt{2m}K + \frac{m}{2\pi Q} }
  = C \cdot\bracket{ \frac{\sqrt{2}K}{\sqrt{m}} + \frac{1}{2\pi Q} },\quad
  k\leq n<k'\\
   b_{kk'}&\leq& \frac{C}{m}\bracket{\sqrt{2m}K + \frac{m}{(2\pi Q)^2} }
  = C \cdot\bracket{ \frac{\sqrt{2}K}{\sqrt{m}} + \frac{1}{(2\pi Q)^2} },\quad
  k, k'>n\label{a3}
  \eeq
  with probability at least $(1-\delta)^2-4e^{-m/18}$. 
  (\ref{a1})-(\ref{a3}) are what we set out to prove. 
 
 \end{proof} 
 \section{  Proof of Theorem \ref{thm:omp}}
\label{app:omp}
\begin{proof}
  We  prove the theorem  by induction. 
  Without loss of generality, we assume that the columns of $\bA$ have unit 2-norm. 
 
In the first step,
\beq
 \label{14'}  |F^*_{J_1}Y| + |G^*_{J_1}Y|& =&  |X_{J_1}F^*_{J_1}F_{J_1}+ X_{J_2}F_{J_1}^{*}F_{J_2}+ ... + X_{J_s}F_{J_1}^{*}F_{J_s} +\\
 &&X'_{J_1}F^*_{J_1}G_{J_1}+ X'_{J_2}F_{J_1}^{*}G_{J_2}+ ... + X'_{J_s}F_{J_1}^{*}G_{J_s} + F_{J_1}^*E|\nn \\
&&+|X'_{J_1}G^*_{J_1}G_{J_1}+ X'_{J_2}G_{J_1}^{*}G_{J_2}+ ... + X'_{J_s}G_{J_1}^{*}G_{J_s}+\nn\\
&&X_{J_1}G^*_{J_1}F_{J_1}+ X_{J_2}G_{J_1}^{*}F_{J_2}+ ... + X_{J_s}G_{J_1}^{*}F_{J_s} + G_{J_1}^*E|\nn\\  
  &\geq&  X_{\rm max} - X_{\rm max}(2s-1)\mu - 2\|E\|_2.\nn
\eeq 
On the other hand, $\forall l \notin\hbox{supp}(X)$,
\beq
\label{15}   |F^*_lY| + |G^*_lY|  & =& |X_{J_1}F^*_lF_{J_1}+ X_{J_2}F_l^{*}F_{J_2}+ ... + X_{J_s}F_l^{*}F_{J_s} +\\
 &&X'_{J_1}F^*_lG_{J_1}+ X'_{J_2}F_l^{*}G_{J_2}+ ... + X'_{J_s}F_l^{*}G_{J_s} + F_l^*E|\nn \\
&&+|X'_{J_1}G^*_lG_{J_1}+ X'_{J_2}G_l^{*}G_{J_2}+ ... + X'_{J_s}G_l^{*}G_{J_s}+\nn\\
&&X_{J_1}G^*_lF_{J_1}+ X_{J_2}G_l^{*}F_{J_2}+ ... + X_{J_s}G_l^{*}F_{J_s} + G_l^*E|\nn\\  
  & \le& 2X_{\rm max} s\mu +2\|E\|_2.\nn
\eeq
Hence,  if
 \[
 (4s-1)\mu +  \frac{4\|E\|_2}{X_{\rm max}} < 1,
 \]
 then the right hand side of (\ref{14'}) is greater than
 the right hand side of (\ref{15}) which implies that
 the first  index selected by OMP must belong to $\hbox{supp}(X)$. 

To continue the induction process, we  need the following result.

\begin{proposition}
Let  $Y=\bF X+\bG X'+E$ where $\supp (X')\subseteq\supp (X)=\cS$. Let $\cS^k$ be a set of $k$ indices
containing both $\supp{(\hat X)}$ and $\supp (\hat X')$. Define
\beq
\label{1001}
Y'=Y-\bF \hat X-\bG \hat X'.
\eeq
Clearly, 
$Y'=\bF(X-\hat X)+ \bG(X'-\hat X')+E$. 
If $\cS^k\subseteq \cS$ and  the sparsity $s$ of $X$ satisfies $4s<1+1/\mu$, then
$\bF(X-\hat X)+\bG(X'-\hat X')$ has a unique sparsest representation
$\bF Z+\bG Z'$ with  $Z=X-\hat X$ and $Z'=X'-\hat X'$.
\label{prop22}
\end{proposition}
\begin{proof}
 Clearly $\supp (Z), \supp(Z')\subseteq \supp (X)$. 
Since 
\[
\|Z\|_0+\|Z'\|_0\leq  2s< {1\over 2} (1+{1\over \mu})
\]
we conclude that $Z$ and $Z'$ are the unique sparsest representation of $\bF(X-\hat X)+\bG(X'-\hat X')$. \end{proof}

Proposition \ref{prop22} says  that selection of a column, followed by the formation of the residual signal, leads to a situation like before, where the ideal noiseless signal has no more representing columns than before, and the noise level is the same.

Suppose
that the set $\cS^k\subseteq \supp (X)$ of $k$ distinct indices
has been selected and
that $\hat X$ in Proposition \ref{prop22} solves the following least squares problem
\beq
\label{ls}
\left(\begin{matrix}
\hat X\\
\hat X'
\end{matrix}\right)=\hbox{arg}\min\|Y-[\bF\ \bG] Z\|_{2},\quad \hbox{s.t.}\quad  \supp (Z)\subseteq \cS^k. 
\eeq

Let $\bF_{\cS^k}$ and $\bG_{\cS_k}$ be, respectively,  the column submatrices of $\bF$ and $\bG$ indexed by the set $\cS^k$. By (\ref{1001}) and (\ref{ls}),  $\bF^*_{\cS^k}Y'=\bG^*_{\cS_k}Y'=0,$ which implies
that no element of $\cS^k$ gets  selected at
the $(k+1)$-st step. 

In order to ensure that some element in $\supp (X) \setminus \cS^k$ gets selected at the $(k+1)$-st
step we only need to repeat the calculation (\ref{14'})-(\ref{15}) to obtain the condition 
\beq
\nn
(4s-1)\mu + \frac{4\|E\|_2}{|X_{J_{k+1}}|+|X'_{J_{k+1}}|} < 1
\eeq
which  follows from  
\beq
\label{snr}
(4s-1)\mu + \frac{4\ep}{X_{\rm min}} < 1.
\eeq

By the $s$-th step, all elements of the support set
are selected and by the nature of the least squares
solution the $2$-norm of the residual is at most $\ep$. 
Thus the stopping criterion is met and the iteration
stops after $s$ steps.

On the other hand, it follows from the calculation
\beqn
2\|Y'\|_2 &\geq & \big| F^*_{J_{k+1}} Y'\big|+\big|G^*_{J_{k+1}}Y'\big|\\
&=&  |X_{J_{k+1}}+ \sum_{i=k+2}^sX_{J_i}F_{J_{k+1}}^{*}F_{J_i}+\sum_{i=k+1}^sX'_{J_i}F_{J_{k+1}}^{*}G_{J_i}+ F_{J_{k+1}}^*E|\nn \\
&+& |X'_{J_{k+1}}+ \sum_{i=k+1}^sX_{J_i}G_{J_{k+1}}^{*}F_{J_i}+\sum_{i=k+2}^sX'_{J_i}G_{J_{k+1}}^{*}G_{J_i}+ G_{J_{k+1}}^*E|\nn \\
&\geq&|X_{J_{k+1}}|+|X'_{J_{k+1}}|-(|X_{J_{k+1}}|+|X'_{J_{k+1}}|)\mu\nn\\
&&- 2(s-k-1)\mu(|X_{J_{k+2}}|+|X_{J_{k+2}}|)-2\|E\|_2\\&\geq& (1-\mu (2s-2k-1))(|X_{J_{k+1}}|+|X'_{J_{k+1}}|)-2\|E\|_2
\eeqn
and (\ref{snr})
that $\|Y'\|_{2}> \ep$ for $k=0,1,\cdots, s-1$. Thus
the iteration does not stop until $k=s$. 

By 
(\ref{ls}), we have 
\[
\|Y-\bF \hat X-\bG \hat X'\|_{2}\leq \|Y-\bF X-\bG X'\|_{2}  \leq\ep
\]
 and
\[
\|\bF(X-\hat X)+\bG(X'-\hat X')\|^2_{2}\leq 2 \|Y-\bF X-\bG X'\|^2_{2}+2\|Y-\bF \hat X-\bG \hat X'\|^2_{2}\leq 2 \ep^2
\]
implying that 
\[
\|\hat X-X\|^2_{2}+\|\hat X'-X'\|_2^2\leq 2\ep^2/\lambda^2_{\rm min}
\]
where 
\[
\lambda_{\rm min}=  \hbox{\rm the $2s$-th largest  singular
value of $\bA$.}
\].

The desired error bound can now be obtained from
 the following result (Lemma 2.2, \cite{DET06}). 
\begin{proposition} 
 Suppose $2s<1+\mu^{-1}$. Every $m\times (2s)$ column submatrix of $\bA$ has the $2s$-th singular value bounded below by $ \sqrt{1-\mu (2s-1)}$.
 \label{prop4}
\end{proposition}

By Proposition \ref{prop4},  $\lambda_{\rm min}
\geq  \sqrt{1-\mu (2s-1)}$ and thus the desired estimate (\ref{err:omp}) follows.
\end{proof}

\commentout{
 \section{Proof of Lemma \ref{lem2}} \label{app2}
 \begin{proof}
Observe that $\bA\bA^*=\bF\bF^*+\bG\bG^*$  with
  \begin{align*}
   (\bF\bF^\ast)_{i,j} &
   = \sum_{k=1}^{n} \e^{2\pi\i Q(\bar{t}_j-\bar{t}_i)k} 
   = \e^{2\pi\i Q(\bar{t}_j-\bar{t}_i)}\cdot {1-e^{2\pi i Q (\bar{t}_j-\bar{t}_i) n}\over 1-e^{2\pi i Q (\bar{t}_j-\bar{t}_i)}} 
   \\
   (\bG\bG^\ast)_{i,j} &
   = \sum_{k=1}^{n}( \bar{t}_i-1/2)(\bar{t}_j-1/2)\e^{2\pi\i Q(\bar{t}_j-\bar{t}_i)k} 
   = ( \bar{t}_i-1/2)(\bar{t}_j-1/2)\e^{2\pi\i Q(\bar{t}_j-\bar{t}_i)}{1-e^{2\pi i Q (\bar{t}_j-\bar{t}_i) n}\over 1-e^{2\pi i Q (\bar{t}_j-\bar{t}_i)}}.
  \end{align*}
 Thus  
  \beq
  \label{30}
   |(\bA\bA^\ast)_{i,j}| = \big(1+( \bar{t}_i-1/2)(\bar{t}_j-1/2)\big) \left|{\sin{[\pi Q (\bar{t}_j-\bar{t}_i) n]}\over \sin{[\pi Q (\bar{t}_j-\bar{t}_i)]}}\right| 
  \eeq
which can be controlled if the random variables
 $ Z_{ij} =Q(\bar{t}_j-\bar{t}_i)$ are bounded away from integers. 
 $ Z_{ij}$ has the density 
 \[
 f_Z = \ind_{[0, Q]}\ast \ind_{[0,Q]}
 \]
  since
  $\bar{t}_i, \bar{t}_j$ are independently and uniformly distributed in $[0,1]$. Clearly $f_Z$ is a triangular function on $[0,2Q]$ and satisfies $|f_Z|\leq 1/Q$. 
 Let 
  $$
   \zeta= \min_{ \substack{ i\neq j}}
   \min_{k\in\mathbb{Z}}\set{\abs{Z_{ij}-k}} .
  $$
  Hence, for small $b>0$,
  $$
   \Pr(\zeta > b) \ > \ (1-C_1 b)^{m\cdot(m-1)}
  $$
  where the power counts for all possible pairs $(\bar{t}_i,\bar{t}_j)$, $1\leq i\neq j\leq m$.
  Thus, with probability greater than $(1-C_1 b)^{m\cdot(m-1)}$ we have
  \[
  \abs{(\bA\bA^*)_{i,j}} < \frac{2}{\pi b},\quad i\neq j. 
  \]
   By choosing $b = \frac{2(m-1)}{n\pi}$ and applying
   Gershgorin circle theorem, 
  we have $ \norm{\bA\bA^\ast - n\bI_m}_2 < n$, or 
  equivalently $\norm{\bA}_2^2 \leq 2n$.
\end{proof}
}

\commentout{
\section{Proof of Lemma \ref{lemm4}}
\label{app4}
\begin{proof}
  The scalar product of two distinct columns of $\bA=[\bF\ \bG \ \bH]$
is of the form:
\beqn
b_{\bp\bp'}&=& \sum_{j=1}^{m_2}\sum_{k=1}^{m_1} a_{j}(\theta_k) e^{-2\pi\i\ell\nu_{j}\bdhat_k\cdot (\bp-\bp')} 
 \eeqn
 where
 \beqn
a_{j}(\theta)&\in& \{1,\ \nu_j\cos\theta/\sigma_1,\ \nu_j\sin\theta/\sigma_2,\ \nu_j^2\cos^2\theta/\sigma_1^2,\ 
  \nu_j^2\sin^2\theta/\sigma_2^2,\ \nu_j^2\cos\theta\sin\theta/(\sigma_1\sigma_2)\}.\nn
  \eeqn

Consider the summation over $k$ in $b_{\bp\bp'}$. Applying the Hoeffding inequality as in the proof of Lemma \ref{lem1}
we obtain
\commentout{
Consider the first summation over $k=1,...,p$. Let
\[
P_k=\cos{(2\pi \ell\nu_j\bdhat_k\cdot(\bp-\bp'))},\quad
Q_k=\sin{(2\pi\ell\nu_j\bdhat_k\cdot(\bp-\bp'))}
\]
and
\[
S_{m_1}=\sum_{k=1}^{m_1} P_k,\quad T_{m_1}=\sum_{k=1}^{m_1} Q_k.
\]
Then  the summation can be
bounded by
\beq
\label{57'}
\lt|\sum_{k=1}^{m_1}e^{-2\pi\i\ell\nu_j\bdhat_k\cdot (\bp-\bp')}\rt|\leq
\sqrt{|S_{m_1}-\IE S_{m_1} |^2+|T_{m_1} -\IE T_{m_1}|^2}+\sqrt{|\IE S_{m_1}|^2+|\IE T_{m_1}|^2}
\eeq

We apply the Hoeffding inequality to both $S_{m_1}$ and $T_{m_1}$. With
\[
t=K/\sqrt{m_1},\quad K>0
\]
we obtain 
\beq
\label{hoeff2}
\IP\lt[m_1^{-1}\lt|S_{m_1} -\IE S_{m_1}\rt|\geq K/\sqrt{m_1}\rt]
&\leq& 2e^{-{K^2/2}}\\
\IP\lt[m_1^{-1}\lt|T_{m_1} -\IE T_{m_1}\rt|\geq K/\sqrt{m_1}\rt]
&\leq& 2e^{-{K^2/2}}\label{hoeff22}.
\eeq

Note that the quantities $S_{m_1}, T_{m_1}$ depend on $
\bp-\bp'$ but they possess the symmetry:  $S_{m_1}(\bp-\bp')
=S_{m_1}(\bp'-\bp), T_{m_1}(\bp-\bp')=-T_{m_1}(\bp'-\bp)$. 
Furthermore, a moment of reflection reveals that thanks
to the square symmetry of the lattice 
there
are at most $n-1$ different values $|S_{m_1}|$ and $|T_{m_1}|$
among the $n(n-1)/2$  pairs  of $(\bp, \bp')$.

We use (\ref{hoeff2})-(\ref{hoeff22}) and the union bound to obtain
\beqn
&&{\IP\lt[\max_{\bp\neq \bp'}m_1^{-1}\lt|S_{m_1}-\IE S_{m_1}\rt|\geq K/\sqrt{m_1}\rt]} \leq 2(n-1) \cdot \e^{-K^2/2}\\
&&{\IP\lt[\max_{\bp\neq \bp'}m_1^{-1}\lt|T_{m_1}-\IE T_{m_1}\rt|\geq K/\sqrt{m_1}\rt]}\leq 2(n-1) \cdot \e^{-K^2/2}
\eeqn
where the factor $4 n$ is due to the structure of square lattice.
Hence, by (\ref{57'}) 
\beq
\lefteqn{{\IP\lt[\max_{\bp\neq \bp'}m_1^{-1}\lt|\sum_{k=1}^{m_1}e^{-2\pi\i \ell\nu_j\bdhat_k\cdot (\bp-\bp')}- \IE\lt[\sum_{k=1}^{m_1}e^{-2\pi\i\ell\nu_j\bdhat_k\cdot (\bp-\bp')}\rt] \rt|< \sqrt{2\over m_1} K\rt]}}\\
&>&(1-2(n-1) \e^{-K^2/2})^2.\label{10.1}
\eeq
By  (\ref{m-2})
 the right hand side of (\ref{10.1})
is greater than $(1-\delta)^2$. 
}
\beq
{{\IP\lt[\max_{\bp\neq \bp'}{1\over m_1}\lt|\sum_{k=1}^{m_1}a_{j}(\theta_k)e^{-2\pi\i \ell\nu_j\bdhat_k\cdot (\bp-\bp')}- \IE\lt[\sum_{k=1}^{m_1}a_{j}(\theta_k)e^{-2\pi\i\ell\nu_j\bdhat_k\cdot (\bp-\bp')}\rt] \rt|< {\sqrt{2}K\over \sqrt{m_1}}\rt]}}
 > (1-\delta)^2.\label{10.1}\nn
\eeq

To estimate 
\beq
{1\over m_1}\IE\lt[\sum_{k=1}^{m_1}a_{j}(\theta_k)e^{-2\pi\i\ell\nu_j \bdhat_k\cdot (\bp-\bp')}\rt]
&=&\int_0^{2\pi} a_j(\theta)\e^{-2\pi\i\ell\nu_j \bdhat\cdot (\bp-\bp')} \phi(\theta) d\theta,\quad \bdhat=(\cos\theta,\sin\theta). \label{39-3}\nn
\eeq
\commentout{
we expand $f$ in the Fourier series
\[
f(\theta)=\sum_{l} c_le^{\i l\theta}
\]
and denoting the angle of $\bp'-\bp$ by $\theta_{\bp'-\bp}$
we can write (\ref{39-3}) as
\beq
&&\sum_l c_l \e^{\i l \theta_{\bp'-\bp}} \int^{2\pi }_0e^{\i l\theta}
e^{2\pi\i \ell \nu_j |\bp-\bp'|\cos\theta }d\theta\nn\\
&=&{2\pi}\sum_l c_l \e^{\i l (\theta_{\bp'-\bp}+\pi/2)} J_l(2\pi\ell\nu_j)
\label{302}
\eeq
where $J_l$ is the Bessel function of order $l$. 
\commentout{
which can be written as  a finite sum of integrals of the form
\beq
\label{42-5}
\int_a^b \e^{-2\pi\i\ell\nu_j\bdhat\cdot (\bp-\bp')} \phi(\theta) d\theta,\quad \phi(\theta)\neq 0,\forall\theta\in (a,b).
\eeq
}The large argument asymptotic for $J_l$
\[
J_l(z)=\sqrt{2\over \pi z}\lt\{ \cos{(z-l\pi/2- \pi/4)}
+\cO(|z|^{-1})\rt\},\quad z\gg 1
\]
suggests the decay estimate (\ref{21-4}). 
}
we use  the method
of stationary phase (Theorem XI. 14 and XI. 15 of \cite{RS}).

\begin{proposition} 
\label{prop:sph}
Let $g_{\bp, \bp'}(\theta)=\bdhat\cdot (\bp-\bp')/|\bp-\bp'|$
which is in $ C^\infty([-\pi,\pi]),\forall\bp,\bp'\in\cL$. 

(i) Suppose  ${d\over d\theta}g_{\bp,\bp'}(\theta)\neq 0, \forall
\theta \in [a,b],\forall \bp,\bp'\in\cL$.
Then for all $\phi\in C^h_0([a,b])$
\beq
\lt|\int \e^{-2\pi\i\ell\nu_j|\bp-\bp'| g_{\bp,\bp'}(\theta)} \phi(\theta) d\theta\rt|
\leq c_h(1+\om|\bp-\bp'|)^{-h} \|\phi\|_{h,\infty}\label{91}
\eeq
for some constant $c_h$ independent of $\phi$.
Moreover, since $\{g_{\bp,\bp'}:\bp,\bp'\in \cL\}$ 
is a compact subset of $C^{h+1}([a,b])$, the constant $c_h$
can be chosen uniformly for all $\bp,\bp'\in \cL$. 

(ii) Suppose ${d\over d\theta} g_{\bp,\bp'}(\theta)$  vanishes at $\theta_*\in (a,b)$. Since $ {d^2\over d\theta^2} g_{\bp,\bp'}(\theta_*)\neq 0$,  there exists a constant $c_t, t>1/2$
such that  
\beq
\lt|\int \e^{-2\pi\i\ell\nu_j |\bp-\bp'| g_{\bp,\bp'}(\theta)}\phi(\theta)d\theta
\rt|\leq c_t (1+\om|\bp-\bp'|)^{-1/2}\|\phi\|_{t,\infty} \label{92}
\eeq
where the constant $c_t$ is independent of $ \bp,\bp'\in\cL$. 

 \end{proposition}
Note that the estimates (\ref{91}) and (\ref{92}) can be made independent of $\nu_j\in [\nu_0,\nu_*]$
by slight adjustment of the numerical constants.

 Combining the above estimates with proper normalization of columns we 
 obtain (\ref{mut})
with probability greater than $1-2\delta m_2$ since the estimates are independent of $\nu_j$. 

\end{proof}
}

\end{appendix}

\bigskip

{\bf Acknowledgement.} We thank an anonymous referee for a suggestion that inspired
our formulation of Algorithm 3 (LOT) in Section 6.  The research is partially supported by
the NSF grant DMS-0908535.

\bigskip
 \bibliographystyle{amsalpha}

\begin{thebibliography}{99}


\bibitem{BS07}
R.~Baraniuk and P.~Steeghs, ``Compressive radar imaging.''
{\em IEEE Radar Conf.}
(2007), 128-133.


\commentout{
\bibitem{BM11}
C.~R.~Berger and J.~M.~F.~Moura, ``Noncoherent compressive sensing with application to distributed radar.''
{\em Inform. Sci. and Syst.}
{\bf } (2011), 1-6.
}

\bibitem{Candes08}
E.~J.~Cand{\`e}s, ``The restricted isometry property and its implications for compressed
sensing.''
{\em Comptes Rendus Mathematique}
{\bf 346} (2008), 589-592.


\bibitem{CEN11}
E.J. Cand\`es, Y.C. Eldar, D. Needell, and P. Randall, ``Compressed sensing with coherent and redundant dictionaries," {\em  Appl. Comput. Harmon. Anal.,} {\bf 31} (2011), pp. 59Ð73.

\commentout{
\bibitem{CP09}
E.~J.~Cand{\`e}s and Y.~Plan,
``Near-ideal model selection by $\ell_1$ minimization.''
{\em Ann. Stat.}
{\bf 37} (2009), 2145-2177.


\bibitem{CU1957}
L.~Carlitz and S.~Uchiyama, ``Bounds for exponential sums.''
{\em Duke Math J.}
{\bf 24} (1957), 37-41.
}



\bibitem{CB08}
M.~Cheney and B.~Borden, ``Imaging moving targets from scattered waves.''
{\em Inverse Probl.}
{\bf 24} (2008), 035005.


\bibitem{CSPC11}
Y.~Chi, L.~L.~Scharf, A.~Pezeshki and A.~R.~Calderbank,
``Sensitivity to basis mismatch in compressed sensing.''
{\em IEEE T. Signal Proces.}
{\bf 59} (2011), 2182-2195.

\commentout{
\bibitem{DM09}
W.~Dai and O.~Milenkovic,
``Subspace pursuit for compressive sensing signal reconstruction.''
{\em IEEE T. Inform. Theory}
{\bf 55} (2009), 2230-2249.
\bibitem{Dau}
I. Daubechies, {\em Ten Lectures on Wavelets.}
SIAM, Philadelphia, 1992.
}

\bibitem{DET06}
D.L. Donoho, M. Elad and V.N. Temlyakov,
``Stable recovery of sparse overcomplete
representations in the presence of noise,''
{\em IEEE Trans. Inform. Theory} {\bf 52} (2006) 6-18. 


\bibitem{DB11}
M.F. Duarte and R.G. Baraniuk, ``Spectral compressive sensing,"
{\em  Appl. Comput. Harmon. Anal.,} {\bf 32} (2012).


\bibitem{Ender10}
J.~Ender, ``On compressive sensing applied to radar.''
{\em Signal Proces.}
{\bf 90} (2010), 1402-1414.



\bibitem{cis-simo}
A.~Fannjiang, ``Compressive inverse scattering {I}. {H}igh-frequency {SIMO/MISO} and
	{MIMO} measurements.''
{\em Inverse Probl.}
{\bf 26} (2010), 035008.



\bibitem{cis-siso}
A.~Fannjiang, ``Compressive inverse scattering {II}. {M}ulti-shot {SISO} measurements
	with {Born} scatterers.''
{\em Inverse Probl.}
{\bf 26} (2010), 035009.
\commentout{
\bibitem{CIS-TV}
A. Fannjiang, `` TV-min and greedy pursuit for constrained joint sparsity and application to inverse scattering,"
{\em  Math. Mech. Complex Syst.} {\bf 1}(2013), pp. 81-104

}
\bibitem{FL1}
A.~Fannjiang and W.~Liao, 
``Mismatch and resolution in compressive imaging," {Wavelets and Sparsity XIV,} 
edited by Manos Papadakis, Dimitri Van De Ville, Vivek K. Goyal, 
{\em Proc. SPIE} Vol. 8138, 0Y1-9,2011.


\bibitem{FL2}
A.~Fannjiang and W.~Liao,
``{Coherence-Pattern} guided compressive sensing with unresolved grids," 
{\em SIAM J. Imag. Sci.} 
{\bf 5} (2012), 179-202.

\bibitem{Asilomar12}A. Fannjiang and W. Liao, ``Super-resolution by compressive sensing algorithms," IEEE Proc. Asilomar conference on signals, systems and computers, 2012.
\commentout{
\bibitem{FSY09}
A.~Fannjiang, T.~Strohmer and P.~Yan,
``Compressed remote sensing of sparse objects.''
{\em SIAM J. Imag. Sci.}
{\bf 3} (2010), 595-618.
}

\bibitem{GPG}
M. J. Gerry, L. C. Potter, I. J. Gupta, and
A. van der Merwe, ``A parametric model for
synthetic aperture radar measurements,"
{\em IEEE Trans. Antennas Propag.} {\bf 47},
pp. 1179Ð1188, 1999.
\bibitem{HS09}
M.~Herman and T.~Strohmer, ``High resolution radar via compressed sensing.''
{\em IEEE Trans. Signal Process.}
{\bf 57} (2009), 2275-2284.

\bibitem{HS2}
M.~Herman and T.~Strohmer, ``General deviants: an analysis of perturbations in compressed sensing," {\em IEEE J. Sel. Topics Signal Process.} {\bf 4} (2010), pp. 342 - 349.

\commentout{
\bibitem{Hoeffding63}
W.~Hoeffding, ``Probability inequalities for sums of bounded random variables.''
{\em J. Am. Stat. Assoc.}
{\bf 58} (1963), 13-30.
}

\bibitem{Jak}
C. V. Jakowatz et al., {\em Spotlight-Mode Synthetic
Aperture Radar: A Signal Processing
Approach}. New York: Springer, 1996.


\bibitem{Kel}
J. B. Keller, ``Geometrical theory of
diffraction", {\em J. Opt. Soc. Amer.}, {\bf  5} (1962)
pp. 116Ð130.


\bibitem{MOJ}
D. C. Munson, Jr., J. D. O'Brien, and
W. K. Jenkins, ``A tomographic formulation
of spotlight-mode synthetic aperture radar,"
{\em 
Proc. IEEE} {\bf  71}, pp. 917-925, 1983.

\commentout{
\bibitem{NT11}
J.~L.~Nelson and V.~N.~Temlyakov,
``On the size of incoherent systems.''
{\em J. Approx. Theory}
{\bf 163} (2011), 1238-1245.

}

\bibitem{PEPC10}
L.~C.~Potter, E.~Ertin, J.~T.~Parker and M{\"u}jdat {\c{C}}etin, 
``Sparsity and compressed sensing in radar imaging.''
{\em Proc. IEEE}
{\bf 98} (2010), 1006-1020.


\bibitem{Rauhut08}
H.~Rauhut, ``Stability results for random sampling of sparse trigonometric polynomials.''
{\em IEEE Trans. Inform. Theory}
{\bf 54} (2008), 5661-5670.



\commentout{

\bibitem{Romberg09}
J. Romberg, ``Compressive sensing by random convolution.''
{\em SIAM J. Imag. Sci.}
{\bf 2} (2009), 1098-1128.


\bibitem{SF09}
T.~Strohmer and B.~Friedlander, ``Compressed sensing for {MIMO} radar - algorithms and performance.''
{\em Signals, Systems and Computers}
{\bf } (2009), 464-468.
}

\bibitem{UWB}
J. D. Taylor, {\em Introduction to Ultra-wideband Radar Systems,} CRC Press, Boca Raton, 1995. 


\bibitem{YALL1}
J. Yang and Y. Zhang,
``Alternating direction algorithms for L1 problems in compressive sensing,"
{\em SIAM J. Sci. Comput.} {\bf 33} (2011),  250 - 278.


\bibitem{ZLG11}
H. Zhu, G. Leus and G. B. Giannakis, 
``Sparsity-cognizant total least-squares for
perturbed compressive sampling," {\em IEEE Trans. Sign. Process.} {\bf 59}
(2011), 2002-2016. 
\end{thebibliography}

\end{document}